\documentclass[11pt]{article}
\usepackage[utf8]{inputenc}
\usepackage{blindtext}
\usepackage{tikz}
\usepackage{url}
\usepackage{xr-hyper}
\usepackage{hyperref}
\usepackage{bbm}
\usepackage{tikz}
\usepackage{caption}
\usepackage{subcaption}
\usetikzlibrary{decorations.pathreplacing}
\usepackage[a4paper, left=2.5cm, right=2.5cm, top=3cm,bottom=3cm,marginparsep=0.5cm]{geometry}
\usepackage{natbib}
\setcitestyle{authoryear}
\usepackage[T1]{fontenc}
\usepackage[export]{adjustbox}
\usepackage{todonotes}
\usepackage{amssymb,amsmath,amsthm,enumerate} 
\usepackage{graphicx}
\usepackage{amsthm}
\usepackage{titling}
\usepackage[affil-it]{authblk}
\usepackage[justification=centering]{caption}
\newtheorem{thm}{Theorem}[section]
\newtheorem{proposition}[thm]{Proposition}
\newtheorem{corollary}[thm]{Corollary}

\newtheorem{remark}[thm]{Remark}

\newtheorem{lemma}[thm]{Lemma}

\DeclareMathOperator*{\sgn}{sgn}
\DeclareMathOperator*{\E}{\mathbb{E}}
\DeclareMathOperator*{\PP}{\mathbb{P}}
\DeclareMathOperator*{\R}{\mathbb{R}}

\usepackage{mathtools}
\mathtoolsset{showonlyrefs}
\hypersetup{hidelinks}
\newcommand\blfootnote[1]{%
  \begingroup
  \renewcommand\thefootnote{}\footnote{#1}%
  \addtocounter{footnote}{-1}%
  \endgroup
}
\newcommand{\w}{\mathfrak{w}}
\newcommand{\rr}{\mathfrak{r}}
\begin{document}

\title{\textbf{Optimal execution and speculation \\with trade signals}}
\date{\vspace{-20pt}}
\author[a]{Peter Bank}
\author[b,c]{\'{A}lvaro Cartea}
\author[a,b]{Laura K\"orber}

\affil[a]{\footnotesize Institut f\"ur Mathematik, Technische Universit\"at Berlin, Berlin, Germany}
\affil[b]{\footnotesize Oxford-Man Institute of Quantitative Finance, University of Oxford, Oxford, UK}
\affil[c]{\footnotesize Mathematical Institute, University of Oxford, Oxford, UK}
\maketitle

\begin{abstract}
    We propose a price impact model where changes in prices are purely driven by the order flow in the market. The stochastic price impact of market orders and the arrival rates of limit and market orders are functions of the market liquidity process which reflects the balance of the demand and supply of liquidity.  Limit and market orders mutually excite each other so that liquidity is mean reverting.  We use the theory of Meyer-$\sigma$-fields to introduce a short-term signal process from which a trader learns about imminent changes in order flow. Her trades impact the market through the same mechanism as other orders.  With a novel version of Marcus-type SDEs  we efficiently describe the intricate timing of market dynamics at moments when her orders concur with that of others. In this setting, we examine an optimal execution problem and derive the Hamilton--Jacobi--Bellman (HJB) equation for the value function of the trader. The HJB equation is solved numerically and we illustrate how the trader uses the signals to enhance the performance of execution problems and to execute speculative strategies. 
\end{abstract}
\begin{description}
\item[Mathematical Subject Classification (2022)] 91B70, 93E20
\item[Keywords] Optimal execution, speculation, trade signal, Meyer $\sigma$-field, stochastic optimal control, Marcus SDEs
\end{description}
\section{Introduction}\blfootnote{P.~Bank and L.~K\"orber acknowledge support by Deutsche Forschungsgemeinschaft through IRTG 2544, P.~Bank also acknowledges support through SFB TRR 388, Projects B03 and B05. L.~K\"orber is grateful to the Oxford-Man-Institute of Quantitative Finance for their generous hospitality. We are also grateful to participants at the Mathematical Finance Seminar at Imperial College, the German Probability and Statistics Days in Essen, the SIAM Conference on Financial Mathematics and Engineering  in Philadelphia and the Women in Mathematical Finance Conference at Rutgers University.}%
In this paper, we propose a reduced-form model where the evolution of prices is determined by the order flow in the market. In our model, changes in the asset's price are caused by market order flow, which is modelled as a pure jump process. Thus, we do not require exogenous dynamics for the evolution of some unaffected fundamental price as is generally proposed in standard models including \cite{AlmgrenChriss:2000} and \cite{ObizhaevaWang:2013}. Here, the price impact of a market order depends on the market's liquidity which is given by the difference between the volume of posted limit orders and the volume of market orders and cancellations. When this difference is high, i.e., when the market is liquid, the price impact of the next market order is low; and when this difference is low, the price impact is high. Our model may be viewed as a permanent price impact model with finite \emph{market depth} similar to that in \cite{HubermanStanzl:04}. However, our impact specification depends on liquidity and we do not require an exogenous fundamental price.

For simplicity, we do not distinguish between the volumes posted on the bid or ask side of the book, a model extension we leave for future research. This allows us to specify the \emph{market's tightness} with a constant spread. The \emph{market's resilience} is captured by letting the arrival rates of market orders, cancellations and limit orders depend on the liquidity in the market: when the level of liquidity is low, the arrival rate of further market orders is low and the arrival rate of limit orders is high, and vice-versa when liquidity is high, see \cite{CarteaJaimungalWang:19} who provide empirical evidence of this effect based on data from the Nasdaq stock exchange. As a consequence, resilience is an endogenous feature in our model as opposed to models such as the one by \cite{ObizhaevaWang:2013} where resilience is an exponential relaxation of price impact towards zero.

We introduce a trader who uses market orders to complete an execution programme, where her objective is to maximize expected utility of terminal wealth and she does not provide liquidity to the book. The trader receives a private signal about the arrival of other limit and market orders and uses it to execute informed trades before other orders arrive in the book. Delicate timing-issues arise when, in reaction to a signal, the trader pre-emptively sends orders just before other orders arrive in the market, while having the option to trade after other orders arrive. As a result, our continuous-time market dynamics have to accommodate three simultaneous jump-shocks in a timely manner. For this, we present a novel Marucs-type SDE system with which we compute the dynamics of the trader's portfolio value. The system of SDEs highlights the positive effect of the arrival of liquidity and reveals the adverse effects of the trader's price impact and transaction costs. Interestingly, the wealth dynamics also reveal how, for a trader with a long position, the arrival of market buy orders is beneficial if their upwards pressure on market prices is not offset by their detrimental effect on market liquidity. An analogous effect is found for short positions and market sell orders.

As mentioned, the trader's market orders affect the price of the asset and the liquidity in the market in the same way as the \emph{external} orders sent by other market participants. Due to the resilience of the book in our model, the trader's orders will trigger more liquidity provision, an indirect market impact which may incentivise pump-and-dump schemes. In such cases,  to ensure a well-functioning market, we introduce a trading halt mechanism that imposes a minimum liquidity requirement for the market to continue operating so that the profitability of pump-and-dump schemes is limited. When a liquidity taking order exhausts liquidity below the minimum liquidity requirement,  trading is halted, and positions are executed in an ``auction'' where the price is randomly drawn from a Gaussian distribution.

In a market with trading halts, the trader's optimal execution problem is well-posed and the value function is non-degenerate. The optimization problem is a singular stochastic control problem when trading is continuous and it results in an impulse control problem when trading at a minimum lot size. Under some regularity assumptions, we prove continuity of the value function. This is non-trivial because traders may be confronted with abruptly differing liquidity evolutions even when starting with very similar initial liquidity. This is due to the jumps corresponding to external. As a result, pathwise closeness cannot be expected and the dynamics of liquidity have to be controlled carefully, which we do with the Marcus-SDE dynamics that provides a remarkably convenient tool for this stability analysis.

We apply the dynamic programming approach to derive the Hamilton–Jacobi–Bellman Quasi-Variational Inequality~(HJBQVI) for both continuous and discrete trading. As a novelty, our HJBQVI contains a double integral straddling a sup-operator that yields the optimal signal-based trade.

Mathematically, the trader's information flow with signal is modeled as a Meyer-$\sigma$-field, see \cite{Lenglart:80}. The technique of Meyer-$\sigma$-fields was first applied in stochastic optimization by \cite{ElKaroui:81} in the context of optimal stopping. \cite{BankBesslich:2020} apply the theory of Meyer-$\sigma$-fields in a singular stochastic control problem which is solved with convex analysis tools. In Merton's optimal investment problem, \cite{BankKoerber:2022} use a Meyer-$\sigma$-field to incorporate a short-term signal about jumps in the price of the asset and use dynamic programming methods to solve for the optimal investment strategy.

We use simulations to study the trader's execution strategies. The trader uses signals about liquidity taking orders to optimize the times of execution and the trading volume. Specifically, upon receiving the signal of an imminent arrival of a liquidity taking order, the trader may submit a market order before the external order arrives and impacts the price. As the bid-ask spread widens, a signal about liquidity provision becomes irrelevant for the execution programme of a trader because there is no incentive to execute a market order before liquidity increases in the book, so price impact of market orders decreases.

For a narrow bid-ask spread, we find that the trader uses the signal about liquidity provision to start speculative roundtrip trades when liquidity is low and completes the roundtrip after liquidity has recovered upon receiving a signal about a liquidity taking order.

Trading on information from signals increases the average terminal wealth and it also increases the variance of the distribution of terminal wealth. In other words, to extract value from the signal, the trader uses strategies that increase the expected gain and that increase the risk of the financial performance of the execution programme.

There exists a broad range of literature on informed trading. Important seminal works are those by \cite{Kyle:85} and \cite{Back:92}. More recent work on optimal trading with market signals is in \cite{CarteaJaimungal:16} who examine optimal execution with a general Markovian signals and derive closed form optimal strategies; see also \cite{CasgrainJaimungal:19} who incorporate latent factors. Similarly, \cite{LehalleNeumann:22} and \cite{NeumannVoss:22} study optimal trading with signals when there is transient price impact and \cite{Belaketal:18} use non-Markovian finite variation
signals. For a market maker, \cite{CarteaWang:19} show how to use a signal of the trend of the price of an asset in optimal liquidity provision strategies; similarly \cite{LehalleMounjid:2017} study optimal market maker strategies with a signal on liquidity imbalance. More recently, \cite{Carteaetal:22} use signatures of the market to generate signals and \cite{CarteaSanchez:22} study how a broker provides liquidity to informed and uninformed traders.

The feature that the trader's market orders influence the arrival rate of other limit and market orders in the same way as external orders is a novelty compared with the existing literature on stochastic optimal control with Hawkes processes. For instance, \cite{AlfonsiBlanc:2016} solve an optimal execution problem in a modified version of the model in \cite{ObizhaevaWang:2013} where the external order flow is modelled as a Hawkes process that is not influenced by the trader.
Similarly, \cite{CarteaJaimungalRicci:2018} study a framework where the fill rates of the trader's controlled limit order flow are driven by an external, uncontrolled Hawkes process. The work of \cite{CayeMuhlekarbe:2014} proposes a modification of the \cite{AlmgrenChriss:2000} model where the past orders of the trader have a self-exciting effect on the price impact. In \cite{FuHorstXia:2020}, the trader influences the base intensity of the mutually exciting external market order dynamics so that the order flow of the trader influences the intensity process in a different way from the external order flow.

The remainder of this paper is organized as follows: Section~\ref{chapter:model} presents the market model and introduces a trader who receives signals on the order flow. Section~\ref{chapter:optimization} examines the problem of optimal investment and execution in a market with trading halts and illustrates the performance of the optimal strategy with trading signals. Results and proofs are collected in the Appendix.
\section{The model}\label{chapter:model}
In this section, we present a model of stock prices where innovations in prices are driven by the flow of market orders. First, we introduce a market model where the arrival of market orders and limit orders is driven by a marked Poisson point process. The arrival rate of orders is a function of the liquidity in the market. At times of high liquidity, market orders arrive more frequently, while at times of low liquidity, the arrival rate of limit orders increases. This feature of the model introduces resilience in the supply and demand of liquidity. Next, we introduce a strategic trader who receives signals about order flow and executes market orders. In our model, all orders arriving in the market, external or those of the trader, have the same effect on the dynamics of liquidity and prices.

\subsection{The uncontrolled model}\label{section:uncontrolled model}
In the following, we present a model for the price dynamics $(P_t)_{t\geq 0}$ of a single asset where price changes are driven by the market order flow $(M_t)_{t\geq 0}$. Market liquidity $\lambda_t$ measures the capacity of the market to fill an incoming market order at every point in time $t\geq 0$. 
The change in liquidity is the difference between the volume of incoming and cancellations of limit orders $(L_t)_{t\geq 0}$ and the volume of incoming market orders. For simplicity, we assume that the liquidity of the buy and sell side are the same. So, the liquidity process $(\lambda_t)_{t\geq 0}$ satisfies the dynamics
\begin{align}\label{eq:dynamics_liquidity}
d\lambda_t&=dL_t-|dM_t|.
\end{align}
Many specifications for the order flows $L$ and $M$ are conceivable. As a tractable possibility we let them be driven by a marked Poisson point process and put
\begin{align}
dL_t&=\int_{E\times \mathbb{R}_+}\mathbbm{1}_{\{y\leq f(\lambda_{t-})\}}\rho(e)N(dt,de,dy),\label{eq:dynamics_LO}
\end{align}
and 
\begin{align}\label{eq:dynamics_MO}
dM_t&=\int_{ E \times \mathbb{R}_+}\mathbbm{1}_{\{y\leq g(\lambda_{t-})\}}\eta(e)N(dt,de,dy).
\end{align}
Here, $N$ is a marked Poisson point process on $[0,\infty) \times E \times \mathbb R_+$ with compensator $ dt \otimes \nu(de) \otimes dy$ where we assume $\nu(E)=1$. The mappings $\eta, \rho \in L^1(\nu)$ determine the volume of the order associated with a mark $e$ that the point process $N$ sets in the Polish mark space $E$, and $\nu(de)$ is the probability for this mark.  
When $\rho(e)>0$, a new limit order of size $\rho(e)$ is posted, and $\rho(e)<0$ corresponds to a cancellation of limit orders of size $|\rho(e)|$. Similarly, $\eta(e)>0$ corresponds to a buy market order of size $\eta(e)$ and $\eta(e)<0$ is a sell market order of size $|\eta(e)|$. Modeling the dynamics of limit and market orders with the same marked Poisson point process allows us to specify the dependence structure between the arrival of orders in a tractable way. More complex specifications are conceivable to address momentum in order flow or time-of-day effects but beyond the scope of this paper. It is worth pointing out, however, that the SDE dynamics of our model to be derived below are fully versatile and able to accommodate more involved order flow dynamics. 

To exclude the simultaneous arrival of limit orders and market orders (which would be cumbersome for bookkeeping), we assume \begin{align}\label{eq:noSimultaneousOrders}
    \eta(e)\rho(e)=0 \text{ for all } e\in E.
\end{align} 

The functions $f$ and $g$ specify how liquidity affects the posting frequency of limit and market orders. We assume that  
\begin{align}\label{asp:fandgLipschitzMonotone}
    \text{$f$ and $g$ are nonnegative, Lipschitz continuous and, respectively, falling and growing.}
\end{align} 
Lipschitz continuity of $f$ and $g$ yields in particular that the Hawkes-like liquidity dynamics~\eqref{eq:dynamics_liquidity}--\eqref{eq:dynamics_MO} have a unique solution; see Lemma~\ref{lemma:uniquesol_model} in the appendix. The monotonicity introduces resilience of liquidity provision. Indeed, it ensures that, as liquidity decreases, fewer market orders will arrive while limit orders will arrive more frequently --- and vice versa when liquidity increases, imposing some degree of self-stabilization for the liquidity balance $\lambda$.

Next, we specify the price dynamics of the asset as
\begin{align}\label{eq:dynamics_price}
  dP_t = I(\Delta_t M,\lambda_{t-}), 
\end{align}
with $\Delta_t M=M_t-M_{t-}$, for $M$ the càdlàg process solving dynamics \eqref{eq:dynamics_MO}, and with the price impact function $I(\cdot,\,\cdot)$ defined by 
\begin{align}\label{eq:definition_I}
I(\Delta,\lambda):=\sgn(\Delta)\int_0^{|\Delta|} \iota(\lambda-z)dz=\int_0^\Delta \iota(\lambda-|z|)dz,\quad \lambda\in\mathbb{R},
\end{align}
where $\sgn(\,\cdot\,)$ denotes the sign function and $\int_0^\Delta=-\int_0^{|\Delta|}$ for $\Delta<0$. Here, the function $I(\Delta,\lambda)$ describes the price impact of a market order of size $\Delta\in\mathbb{R}$ that arrives when the liquidity of the market is $\lambda$. The sign of $I(\Delta,\lambda)$ is determined by that of $\Delta$: buy orders increase the price and sell orders decrease it. 
The function $$\iota:\mathbb{R}\rightarrow\mathbb{R}_+ \text{ is Lipschitz continuous and non-increasing}$$ to ensure that market orders have less impact when liquidity is high. Similarly, the absolute value of the function $I$ is increasing in the size $|\Delta|$ because, everything else being equal, as the volume of market orders increases, so does the impact of the order on the price of the asset. Finally, definition \eqref{eq:definition_I} makes the price dynamics \eqref{eq:dynamics_price} consistent when orders are split, i.e.,  $I(\Delta,\lambda)=I(\Delta_1,\lambda)+I(\Delta_2,\lambda-|\Delta_1|)$ for $\Delta_1+\Delta_2=\Delta$ with $\sgn(\Delta_1)=\sgn(\Delta_2)$. In fact, when taking order splitting to the limit, this observation reveals that an equivalent specification of the price dynamics can be given in the form of the Marcus-type SDE system (cf.~\cite{Marcus:80}) 
\begin{align}\label{eq:infinitesimalPriceLiquidityDynamics}
     \diamond\;d\lambda_t = \diamond\, (dL_t-|dM_t|)\quad \text{and} \quad \diamond dP_t = \iota(\lambda_t) \diamond dM_t.
\end{align}

\begin{remark}\label{rem:Marcus}
    We briefly explain the Marcus-type SDE dynamics in~\eqref{eq:infinitesimalPriceLiquidityDynamics}. The symbol $\diamond$ indicates that, at any time $t$, a jump by the drivers (here $M$, its total variation $V[0,.](M)=\int_0^.|dM_t|$), and $L$) will make the system (here $(\lambda,P)$) jump, too. The system's jump is determined by the endpoint at $s=1$ of  a ``fictitious'' flow, that runs in line with the given dynamics on a separate time line $s \in [0,1]$ and starts with the pre-jump values. In our example, this flow is given by
    \begin{align}
        (l_0,p_0):=(\lambda_{t-},P_{t-}), \quad   dl_s = \Delta_t L \, ds-|\Delta_t M| \,ds, \quad dp_s =\iota(l_s) \Delta_t M ds,s \in [0,1].
    \end{align}
    Its solution is easily found to be 
    \begin{align}
        l_s=l_0+(\Delta_t L -|\Delta_t M|)s, \quad p_s=p_0+\int_0^s \iota(l_r) \Delta_t M dr, \quad s \in [0,1], 
    \end{align}
    and so, recalling that $\Delta_t L=0$ whenever $\Delta_t M \not=0$ by~\eqref{eq:noSimultaneousOrders}, the post-jump values become
    \begin{align}
        \lambda_t&:= l_1 = 
        \lambda_{t-}+\Delta_t L -|\Delta_t M|,\\
        P_t & := p_1 = 
        P_{t-}+\int_0^1\iota(l_{0}-|\Delta_t M|r) \Delta_t M dr =P_{t-}+I(\Delta_t M, \lambda_{t-}),
    \end{align}
    in exact accordance with~\eqref{eq:dynamics_liquidity} and~\eqref{eq:dynamics_price} .
\end{remark}

Our model exhibits a direct link between price volatility and market liquidity because price volatility is determined by the arrival rate and the size of price changes, both of which depend on the liquidity in the market. More precisely, market liquidity affects the arrival rate of market orders, i.e., the arrival rate of price changes, through the dynamics as in \eqref{eq:dynamics_MO} and the size of price changes in \eqref{eq:dynamics_price} through the price impact function $I$ as in \eqref{eq:definition_I}. For instance, when liquidity in the market decreases, the arrival rate of market orders decreases and the size of the price impact of market orders increases.
The next lemma makes this link between volatility and liquidity explicit and states an elasticity condition under which the volatility of prices decreases with market liquidity, see \cite{PastorStambaugh:03} for empirical evidence of this negative relation between price volatility and liquidity.
\begin{lemma}\label{lemma:link_volaliquidity}
The predictable quadratic variation of the price $P$ is given by
\begin{align}\label{eq:pred_quadvar_price}
    d\langle P \rangle_t = \sigma^2(\lambda_{t-})dt,
\end{align}
where
\begin{align}\label{eq:function_price_vola}
    \sigma(\lambda):=\left(g(\lambda)\int_E I(\eta(e),\lambda)^2\nu(de)\right)^{1/2}.
\end{align}
The function in \eqref{eq:function_price_vola} is strictly decreasing in $\lambda$ under the elasticity condition
\begin{align}\label{eq:elasticitycond}
0< \left.\frac{\partial_\lambda g(\lambda)}{g(\lambda)}\middle/ \middle(-\frac{\partial_\lambda \mathbb I^2(\lambda)}{\mathbb I^2(\lambda)}\right) < 1,
\end{align}
where $\mathbb I(\lambda):=\big(\int_E I(\eta(e),\lambda)^2\nu(de)\big)^{1/2}$ is the $L^2$-norm of the price impact size from market orders at liquidity level $\lambda$.
\end{lemma}
For a proof, see the Appendix~\ref{sec:uniquesol_model}.

The elasticity condition compares the relative change in the arrival rate $g$ of market orders with the absolute value of the relative change in the impact size of market orders under the $L^2$-norm. Thus, condition \eqref{eq:elasticitycond} says that when liquidity decreases, the price impact increases at a higher rate than the arrival rate of market orders decreases. Hence, price volatility rises with low liquidity because the increase of the price impact of market orders at low liquidity more than compensates for the decrease in the arrival rate of market orders.

\subsection{The controlled model with trading signal}\label{section:controlled_model}
 Next, we introduce a trader who follows an optimal execution programme. For simplicity, we assume that the trader sends market orders and does not provide liquidity to the market.

To describe the trader's information structure, we introduce the (right-continuous) filtration \begin{align*}
    \mathcal{F}_t=\sigma(N([0,s]\times E\times [0,y]), s\in[0,t], E\in\mathcal{E},y\in\mathbb{R}_+), \quad t\geq 0
\end{align*} generated by the point process $N$. The usual information framework of the predictable information flow $\mathcal{P}(\mathcal{F})$  contains only past information about the order flow and a trader with information $\mathcal{P}(\mathcal{F})$ can only trade after the arrival of external limit and market orders are revealed. Thus, she can only react to market events after they have happened.

By contrast, in our setting the trader receives a short-term signal about imminent changes in the order flow. She uses this information to anticipate jumps in liquidity and in prices, and executes signal-based trades before the external order arrives in the market. The signal is given by the process
\begin{align}\label{def:signal}
Z_t=\int_{\{t\}\times E \times \mathbb{R}_+}z(e,y)N(ds,de,dy), 
\quad t \geq 0,
\end{align}
where $N$ is the same marked Poisson point process that also drives the external order flows \eqref{eq:dynamics_LO} and \eqref{eq:dynamics_MO}. The function $z\in L^0(\nu(de)\otimes dy)$ determines what the trader learns about any mark $(e,y)$ set by $N$, with $z(e,y)=0$ amounting to the same as receiving no signal. Signals $\bar{z}$ different from zero will be received with frequency $\mu(d\bar{z}):=(\nu\otimes \mathrm{Leb}) \circ z^{-1}(d\bar{z})$ which we assume to be a $\sigma$-finite measure on $z(E \times \mathbb{R}_+)\setminus\{0\}$. We introduce the Meyer-$\sigma$-field (cf.~\cite[Def. 2]{Lenglart:80})
\begin{align}
\Lambda:=\mathcal{P}(\mathcal{F})\vee \sigma(Z),
\end{align}
which adds the information $\sigma(Z)$ of the signal $Z$ to the predictable information $\mathcal{P}(\mathcal{F})$. It is this $\sigma$-field which will formally describe the information flow based on which the investor chooses her strategy.

\begin{remark}\label{remark:signal_practice}
    A key competitive edge of high-frequency traders are strategies that use their speed advantage and that employ trading signals, which include those that use limit order book information to predict order flow, see \cite{Lewis:14}. Other signals are designed to anticipate large meta orders split and routed to different exchanges at slightly different points in time. In our case study below, we consider a signal of whether an impending next order will provide or will consume liquidity.
\end{remark}

The trader's strategy is described by a process $(C_t)_{t\geq 0}$ of locally bounded variation starting at $C_{0-}=0$. Its changes represent changes in the trader's inventory so that at any time $t$ the net number of shares sold or bought up to this moment is given by $C_t$. We assume admissible strategies $C$ to induce uniformly bounded inventories:
\begin{align}
    Q^C_t := q_0+C_t \text{ with } |Q^C_t| \leq \bar{q}\label{eq:inventorybound}
\end{align}
for some constant $\bar{q}\in [0,\infty)$ (which could be, for instance, the total number of shares outstanding). Trades can be of any size in $D=\mathbb{R}$, or --- in line with market practice --- a multiple of a minimal lot size $\delta>0$, so $ D=\{...,-2\delta,-\delta,0,\delta,2\delta,...\}$. The set of admissible controls is therefore
\begin{align}\label{eq:definition_admcontrols}
\mathcal{C}(q)=\left\{ C\text{ $D$-valued, $\Lambda$-measurable of bounded total variation with \eqref{eq:inventorybound}, $C_{0-}=0$}\right\}.
\end{align}
Notice that we allow $C$ to be merely l\`adl\`ag and so there may be `left jumps' $\Delta^l_t C:=C_t-C_{t-}$ and `right jumps' $\Delta^r_t C:=C_{t+}-C_t$ at the same time~$t\geq 0$. They correspond to the proactive signal-based trades and reactive state-based trades we will introduce shortly.

The following lemma provides a representation of the trader's strategy and clarifies how she uses the signal $Z$; its proof is given in the Appendix~\ref{sec:DecompositionStrategies}.
\begin{lemma}\label{Lemma:decomposition_q}
A process $C=(C_t)_{t \in [0,T]}$ is a $\Lambda$-measurable process of locally bounded variation if and only if it admits the decomposition
\begin{align}\label{eq:representation_Q}
C_{t}&=C^c_{t}+\sum_{0\leq s \leq t} \Delta^l_s C+\sum_{0\leq s < t} \Delta^r_s C, \quad t \in [0,T],
\end{align}
where $C^c$ is a continuous, adapted process of locally bounded variation, $\Delta^r C$ is an adapted process of bounded variation, and where 
\begin{align}\label{eq:representation_deltaleft}
\Delta^l_t C&= 
\int\limits_{ \{t\}\times \{z\not=0\}}\Gamma_s(z(e,y))N(ds,de,dy) + \Delta_t C^{l,p} \text{ for } t \in [0,T] 
\end{align}
almost surely for some predictable, right-continuous pure-jump process $C^{l,p}$ of bounded variation and some $\mathcal{P}(\mathcal{F})\otimes\mathcal{B}(\mathbb{R})$- measurable field  $\Gamma$ satisfying  
\begin{align}\label{eq:integrability_gamma}
\int\limits_{[0,T]\times \{z\not=0\}}|\Gamma_s(z(e,y))|N(ds,de,dy)<\infty\quad  \text{almost surely.}
\end{align}
Almost surely, predictable left-jumps only happen when there are no marks: 
\begin{align}\label{eq:marksnotcharged}
\PP\left[\left\{t \in [0,T]\,:\, \Delta_t C^{l,p}\not=0 \right\} \subset \left\{t \in [0,T]\,:\, N(\{t\}\times E \times [0,\infty))=0\right\}\right]=1.
\end{align}
\end{lemma}
Thus, upon observing a signal $Z_t\not=0$, the trader sends the order $\Delta^l_t C=\Gamma^l_t(Z_t)$ before the full information about external orders becomes known to all market participants and its effect materialises, i.e., liquidity and prices change as a result. Hence, we call such a trade \emph{signal-based}. 

In contrast, the right jump $\Delta^r_t C$ is the trader's action after the arrival of an external order in the market. It incorporates the full information about the post-shock market state, i.e., the market state after the arrival of external orders. In addition, $\Delta^r_t C$ can also represent an order sent out of the trader's own volition, typically motivated by the state of her execution programme. The latter also applies to the predictable left trades $\Delta^{l,p}_t C$. Therefore, we call these trades \emph{state-based}.

\begin{remark}\label{remark:hawkes}
In our model, the controlled orders $\Delta^l_t C$, $\Delta^r_tC$, $dC^c_t$ affect market liquidity and, with it, the arrival rate of external shocks, in the same way as the external limit orders, market orders and cancellations. This is a key difference between our model and those proposed in the existing literature on stochastic control for Hawkes processes, see e.g. \cite{AlfonsiBlanc:2016}, \cite{CarteaJaimungalRicci:2018}, and \cite{FuHorstXia:2020}. 
\end{remark}

Let us now introduce the controlled dynamics of the state process 
\begin{align*}
S^C_t:=(\lambda^C_t,Q^C_t,P^C_t,X^C_t), \quad t\geq 0,
\end{align*} 
starting at $S^C_{0-}:=(\lambda,q,p,x)$. Here, $\lambda^C_t$ is the current level of market liquidity, $Q^C_t$ is the trader's current inventory, $P^C_t$ is the current asset price, and $X^C_t$ is the trader's cash process resulting from her strategy $C$.

When no external market or limits orders arrive and when there are no jumps in the trader's strategy, the state process satisfies the continuous dynamics 
\begin{align}\label{eq:cont_dynamics_S}
    dS^C_t&=(-|dC^c_t|,dC^c_t,\iota(\lambda^C_t)dC^c_t,-P^C_tdC^c_t-\zeta|dC^c_t|) \\&\text{ when }\Delta^l_t C=\Delta^r_t C=0, N(\{t\}\times E\times \mathbb{R}_+)=0,
\end{align}
where $\zeta\geq 0$ is the half-spread quoted in the book.

Next, we describe how the state changes when there are jumps in the order flow. A jump can result, for instance, from an order of size $\Delta>0$ sent by the trader. Thus, the asset price changes by $I(\Delta,\lambda)$ where $\lambda$ denotes the pre-trade liquidity level, see \eqref{eq:dynamics_price}.
When the trader sends an order of size $\Delta>0$, the price impact from her trade affects the trader's post-trade cash position in a nonlinear way. To describe this, it is convenient to introduce the function 
\begin{align}\label{eq:definition_Xi}
\Xi(\Delta,\lambda):=\int_0^{|\Delta|}I\left(z,\lambda\right)dz,
\end{align}
which reflects the impact costs of the trade $\Delta$ when liquidity is $\lambda$. So the total costs of buying $\Delta\geq 0$ shares when the pre-transaction mid-price  is $p$ amount to $(p+\zeta)\Delta+\Xi(\Delta,\lambda)$; similarly, the revenue from selling $\Delta\geq 0$ shares is $(p-\zeta)\Delta-\Xi(\Delta,\lambda)$.
As in the definition of $I$ in~\eqref{eq:definition_I}, the function $\Xi$ is consistent with order splitting, i.e.,  $\Xi(\Delta,\lambda)=\Xi(\Delta_1,\lambda)+\Xi(\Delta_2,\lambda-|\Delta_1|)$ for $\Delta_1+\Delta_2=\Delta$ with $\sgn(\Delta_1)=\sgn(\Delta_2)$. Similar to the infinitesimal version of the price and liquidity dynamics~\eqref{eq:infinitesimalPriceLiquidityDynamics}, this splitting yields a Marcus-SDE style description also for the investor's cash position; cf.~\eqref{eq:MarcusDynamicsCashControlled} below.

Finally, the precise timing when market orders, limit orders, and cancellations arrive is important and not interchangeable. This is because the different orders arrive when the level of liquidity is~$\lambda^C_{t-}$, $\lambda^{C}_{t-}-|\Gamma_t(Z_t)|$ and $\lambda^{C}_t:=\lambda^{C}_{t-}-|\Gamma_t(Z_t)|-|\Delta_t M^C|+\Delta_t L^C$, respectively, and so the impact on the price of the asset varies and the effect on the trader's cash process is different, too. We explain this step-by-step, see Figure \ref{fig:triplejump} for an illustration.
\begin{center}
      \begin{tikzpicture}
\draw[black] (0,0)-- (1,0)--(1,1)--(0,1)--(0,0);
\draw[black] (6,0)-- (6,1)--(7,1)--(7,0)--(6,0);
\draw[black] (10,0)-- (10,1)--(11,1)--(11,0)--(10,0);
\node at (0.5,1.6) {Pre-trade};
\node at (0.5,1.3) {state};
\node at (0.5,0.5) {$S^C_{t-}$};
\node at (6.5,1.6) {Post-shock};
\node at (6.5,1.3) {state};
\node at (6.5,0.5) {$S^C_{t}$};
\node at (10.5,1.6) {Post-trade};
\node at (10.5,1.3) {state};
\node at (10.5,0.5) {$S^C_{t+}$};
\node at (1.9,-0.4) {Signal-based trade};
\node at (2,-0.8) {$\Delta^l_t C$};
\node at (5,-0.4) {External orders};
\node at (5,-0.8) {$\Delta_t M^C, \Delta_t L^C$};
\node at (8.5,-0.4) {State-based trade};
\node at (8.5,-0.8) {$\Delta^r_t C$};
\node at (3.5,-1.9) {Signal $Z_t$ informs trader};
\draw[->] (1.1,0.5) -- (5.9,0.5);
\draw[->] (7.1,0.5) -- (9.9,0.5);
\draw[->] (8.5,-0.2) -- (8.5,0.4);
\draw[->] (5,-0.2) -- (5,0.4);
\draw[->] (2,-0.2) -- (2,0.4);
\draw[->] (2,-1.6) -- (2,-1.1);
\draw[-] (2,-1.6) -- (5,-1.6);
\draw[-] (5,-1.1) -- (5,-1.6);
\end{tikzpicture}
\captionof{figure}{Evolution of our state process $S^C$}
\label{fig:triplejump}
\end{center}
\paragraph{From pre-trade state $S^C_{t-}$ to post-shock state $S^C_t$.}

We assume the pre-trade state $\mathbf{s}_-:=S^C_{t-}$ is
\begin{align}\label{eq:pretradestate}
\mathbf{s}_-=(\boldsymbol{\lambda}_-,\mathbf{q}_-,\mathbf{p}_-,\mathbf{x}_-):=(\lambda^C_{t-},Q^C_{t-},P^C_{t-},X^C_{t-}).
\end{align}
The impending external limit and market orders at time $t$ are
\begin{align}
   \rho:=\Delta_t L^C&=\int_{\{t\}\times E\times \R_+}\mathbbm{1}_{\{y\leq f(\lambda^C_{s-})\}}\rho(e)N(ds,de,dy),\label{eq:dynamics_LO_controlled}\\
   \eta:=\Delta_t M^{C}&=\int_{ \{t\}\times E\times \mathbb{R}_+}\mathbbm{1}_{\{y\leq g(\lambda^C_{s-})\}}\eta(e)N(ds,de,dy).\label{eq:dynamics_MO_exog_controlled}
\end{align}
\begin{remark}
    Equations~\eqref{eq:dynamics_LO_controlled} and~\eqref{eq:dynamics_MO_exog_controlled} specify the dynamics of incoming limit and market orders analogously to~\eqref{eq:dynamics_LO} and~\eqref{eq:dynamics_MO}. However, their jump frequencies are now determined by the liquidity $\lambda^C_{t-}$. This is another way by which the investor affects market dynamics because her market orders will attract more limit orders, an effect that may be exploited in a pump-and-dump scheme.
\end{remark}
As derived in Lemma~\ref{Lemma:decomposition_q}, when the trader receives a signal $Z_t$, she trades the quantity $\gamma:=\Gamma_t(z)$. Without a signal, she can still opt to send the predictable order $\gamma:=\Delta^{l,p}_t C$. In either case, market liquidity updates to 
\begin{align}\label{eq:lambdaleft}
    \boldsymbol{\lambda}(\gamma,\eta,\rho;\mathbf{s}_-):=\boldsymbol{\lambda}_-- |\gamma|-|\eta|+\rho,
\end{align}
and the trader's inventory becomes
\begin{align}
    \mathbf{q}(\gamma,\eta,\rho;\mathbf{s}_-):=\mathbf{q}_-+\gamma.
\end{align}
Due to price impact, the price changes to
\begin{align}
    \mathbf{p}(\gamma,\eta,\rho;\mathbf{s}_-):=\mathbf{p}_-+I\left(\gamma,\boldsymbol{\lambda}_-\right)+I\left(\eta,\boldsymbol{\lambda}_--|\gamma|\right),
\end{align}
where the market order of size $\eta$ arrives when liquidity is $\boldsymbol{\lambda}_--|\gamma|$. 
The trader's cash position becomes
\begin{align}
    \mathbf{x}(\gamma,\eta,\rho;\mathbf{s}_-):=\mathbf{x}_--\mathbf{p}_-\gamma-\zeta|\gamma|-\Xi(\gamma,\boldsymbol{\lambda}_-).
\end{align}
Thus, we define the state-update function $\mathfrak{s}$ as 
\begin{align}\label{eq:def_mathfrak_s}
    \mathfrak{s}(\gamma,\eta,\rho;\mathbf{s}_-):=(\boldsymbol{\lambda},\mathbf{q},\mathbf{p},\mathbf{x})(\gamma,\eta,\rho;\mathbf{s}_-),
\end{align}
to write the post-shock state as
\begin{align}
    S^C_t:=\mathfrak{s}\left(\Gamma_t(Z_t)\mathbbm{1}_{\{Z_t\not=0\}}+\Delta^{l,p}_tC,\Delta_tM^{C},\Delta_tL^C;S^C_{t-}\right).
\end{align}
\paragraph{From post-shock state $S^C_t$ to post-trade state $S^C_{t+}$.}

After the realised external shock and the post-shock state $S^C_t$ become fully known to the trader, she executes the state-based trade $\Delta^r_t C$ and the post-trade state is defined as
\begin{align}
    S^C_{t+}:=\mathfrak{s}(\Delta^r_t C,0,0;S^C_{t}),
\end{align}
where $\mathfrak{s}$ is as in \eqref{eq:def_mathfrak_s}. 

It is convenient to have an equivalent alternative description of the above state dynamics in the form of an SDE-system. For this one can use the framework of~\cite{Chevyrev:24} which introduces a fictitious time interval $[0,1]$ at every jump time to specify a ``manual'' describing the SDE-sytem's reaction to jumps. Our idea is to use this fictitious time interval to coordinate the three jumps that can affect the market at a given time: a (typically) signal-based trade $dC^l_t$, either a limit order/cancellation $dL^C_t$ or an exogenous market order $dM^C_t$ (recall~\eqref{eq:noSimultaneousOrders}), and a state-based trade $dC^{r,c}_t=dC^{r}_t+dC^c_t$ (to which we have added the continuous part of the control for the sake of simplicity and completeness; cf. also Lemma~\ref{Lemma:decomposition_q}). 

Recall that the exogenous orders only depend on the pre-jump liquidity $\lambda^C_{t-}$:
\begin{align}
    dL^C_t&=\int_{E\times \R_+}\mathbbm{1}_{\{y\leq f(\lambda^C_{t-})\}}\rho(e)N(dt,de,dy),\label{eq:MarcusDynamics_LO_controlled}\\
   dM^{C}_t&=\int_{ E\times \mathbb{R}_+}\mathbbm{1}_{\{y\leq g(\lambda^C_{t-})\}}\eta(e)N(dt,de,dy).\label{eq:MarcusDynamics_MO_exog_controlled}
\end{align}   
In line with the order splitting impact specification, we let a Marcus-style fictitious flow (cf.\ Remark~\ref{rem:Marcus}) continuously drive liquidity $\lambda^C$, stock and cash position $(X^C, Q^C)$, and price $P^C$ by these shocks in the following way:
\begin{align}  
\lambda^C_{0-} &=\lambda,  & \diamond\; d\lambda^C_t &= \diamond \left(  (-|dC^{l}_t|)\vec{+}(dL^C_t-|dM^C_t|)\vec{+}(-|dC^{r,c}_t|)\right),\label{eq:MarcusDynamicsLiquidityControlled}\\
   X^C_{0-} &=x,  &\diamond\; dX^C_t & = -P^C_t \diamond \left(dC^{l}_t \vec{+} 0 \vec{+} dC^{r,c}_t\right)-
   \zeta \diamond \left(|dC^{l}_t|\vec{+}0\vec{+}|dC^{r,c}_t|\right), \label{eq:MarcusDynamicsCashControlled}\\
   Q^C_{0-} &=q,  &\diamond\; dQ^C_t &= \diamond \left( dC^l_t \vec{+} 0 \vec{+} dC^{r,c}_t\right),\label{eq:MarcusDynamicsPositionControlled}\\
   P^C_{0-} &=p,  &\diamond\; dP^C_t &= \iota(\lambda^C_t) \diamond \left(dC^{l}_t \vec{+} dM^C_t \vec{+} dC^{r,c}_t\right).\label{eq:MarcusDynamicsPriceControlled}
\end{align}
Here, the $\vec{+}$-operator means that the jumps it connects are to be run through in the order indicated (with $0$ signifying a pause) when following the fictitious Marcus flow to determine the resulting jumps; cf.\ Remark~\ref{rem:Marcus}. This is done with the understanding that the first of the three $\vec{+}$-summands is taken care of over the first third of the fictitious time interval $[0,1]$, the second one over the second third, and the third one in the remaining third. This way liquidity $\lambda^C$ impacts prices $P^C$ in sync with all incoming orders, while the external orders have no impact on the investor's cash or stock positions as indicated by the $\vec{+}0$-pause terms in the middle of the drivers of both $X^C$ and $Q^C$. The ``intermediate'' values $\lambda^C_t$, $Q^C_t$, $X^C_t$, $P^C_t$ at time $t$ (rather than their right limits at time $t+$) can be obtained from these dynamics by evaluating the Marcus-style dynamics when its fictitious flow at time $t$ has run two thirds of its course (i.e., when at time $s=2/3$ two of the three jump contributions at time $t$ have been taken care of). This way we get the SDE system to fully describe also the dynamics of such l\`adl\`ag processes.

The advantages of the SDE-system description are borne out by the following result which describes the dynamics of the realizable portfolio value
\begin{align}\label{eq:liquidationWealth}
    W^C_t:=\w(X^C_t,Q^C_t,P^C_t,\lambda^C_t), \quad t \in [0,T],
\end{align}
where the function 
\begin{align}\label{eq:liquidationWealthFunction}
    \w(x,q,p,\lambda):=x+qp-\left(\zeta|q|+\Xi(q,\lambda)\right)
\end{align}
adjusts the mark-to-market or book value $x+qp$  of the agent's cash and stock position $(x,q)$ by the transaction costs $\zeta|q|$ and the impact costs $\Xi(q,\lambda)$ that would be incurred if the stock position was liquidated at present liquidity $\lambda$.

\begin{proposition}\label{prop:wealthDynamics}
    The realizable portfolio value $W^C$ associated with a strategy $C \in C$ follows for $t \in [0,T]$ the dynamics
    \begin{align}
            {W}^C_t =\;&\w(x,q,p,\lambda)\\
    & + \int_0^t \left(I(|{Q}^C|,{\lambda}^C)-|{Q}^C|\iota({\lambda}^C)\right) \diamond \left(0\vec{+}d{L}^C\vec{+}0\right)\\
     & - \int_0^t \left(I(|{Q}^C|,{\lambda}^C)-2\mathbbm{1}_{\{{Q}^C>0\}}|{Q}^C|\iota({\lambda}^C)\right) \diamond \left(0\vec{+}d{M}^{C,+}\vec{+}0\right)\label{eq:WealthDynamics}\\&-\int_0^t\left(I(|{Q}^C|,{\lambda}^C)-2\mathbbm{1}_{\{{Q}^C<0\}}|{Q}^C|\iota({\lambda}^C)\right) \diamond \left(0\vec{+}d{M}^{C,-}\vec{+}0\right)\\
    & - \int_0^t 2\left(I(|{Q}^C|,{\lambda}^C)-|{Q}^C| \iota({\lambda}^C)+\zeta\right) \left(\mathbbm{1}_{\{{Q}^C>0\}}\diamond d{Q}^{C,+}+
    \mathbbm{1}_{\{{Q}^C<0\}}\diamond d{Q}^{C,-}\right),
    \end{align}
  where we use the notation $dM^C=dM^{C+}-dM^{C-}$ and  $dQ^C=dQ^{C,+}-dQ^{C-}
  $ for the Hahn-decompositions of $M^C$ and $Q^C$.
\end{proposition}
The proof of this result is deferred to  Appendix~\ref{proof:propWealthDynamics}.

The above result does not depend on our choice for the dynamics of $L^C$ and $M^C$, but only on the market mechanics we have specified above. As a result, the following observations hold true generically. For instance,   we have $I(|{Q}^C|,{\lambda}^C)-|Q^C|\iota(\lambda^C)=\int_0^{|Q^C|}\iota(\lambda^C-z)-\iota(\lambda^C)\geq 0$ because $\iota$ is non-increasing, and so added liquidity (``$d{L}^{C,+}>0$'') boosts the realizable portfolio value while any cancellation (``$d{L}^{C,-}>0$'') has the opposite effect. For the same reason, expanding an existing position (e.g., ``$d{Q}^{C,+}>0$'' when ${Q}^C>0$) also reduces the realizable portfolio value, an effect amplified by the spread $\zeta$. External orders in opposition to the investor's position (e.g., ``$d{M}^{C,+}>0$'' when ${Q}<0$) penalize the liquidation value in proportion to the liquidation price impact $I(|{Q}^C|,{\lambda}^C)$. Only external market orders in the direction of the investor's current position can result in both positive or negative consequences. For instance, an external buy order ``$d{M}^{C,+}>0$'' will be welcome for a trader with a long position ${Q}^C>0$, yet only if its favorable price impact dominates its detrimental effect on market liquidity in the sense that $2|{Q}^C|\iota({\lambda}^C)>I(|{Q}^C|,{\lambda}^C)+2\zeta$. 

The above dynamics show furthermore that direct market manipulation is not profitable. In particular, immediate round trips will result in a loss. Note that this finding does not contradict \cite{HubermanStanzl:04} since they consider impact factors that are independent of the trader's orders. 

While there is no way to turn a profit by one's own trading, the trader can still try to game the system through her short- to medium-term influence on market liquidity. Indeed, this is made possible by our choice for the dynamics of incoming limit and market orders~\eqref{eq:dynamics_LO_controlled} and~\eqref{eq:dynamics_MO_exog_controlled}: In this specification, the jump frequencies of these external orders are determined by the liquidity $\lambda^C$, which the trader can influence. This is an indirect way in which the trader affects market dynamics because the market orders issued by her will attract more limit orders, an effect that could be exploited in a pump-and-dump scheme; see Section~\ref{section:circuitbreaker} below for further considerations in that respect. 

As a final observation, the investor's influence on liquidity does not allow her to directly trigger or preempt impending exogenous orders by making $\lambda^C$ jump. This is because the dynamics of $L^C$ and $M^C$ are solely governed by the pre-jump liquidity level~$\lambda^C_{-}$ that cannot be changed in an ad hoc manner.

\section{The optimal investment and execution problem}\label{chapter:optimization}
In this section, we show how a trader who receives the private signal $Z$ executes a position with market orders over some finite time horizon $[0,T]$. The trader's performance criterion is the expected utility of terminal wealth.

The self-exciting nature of our system requires special care to avoid blow-ups. Thus, we introduce a lower bound on liquidity so that liquidity taking orders (i.e., market orders and limit order cancellations) are limited by the supply of liquidity in the market. The lower bound on liquidity can be interpreted as a threshold that triggers a trading halt imposed by the exchange to ensure a well-functioning market. As we show in this section's main result Theorem~\ref{thm:continuity_all}, the trading halt ensures that the value function of the optimization problem is non-degenerate. Finally, we illustrate the market model with signals and showcase the performance of the optimal strategy.

\subsection{Trading halts for sufficient market liquidity}\label{section:circuitbreaker}

In practice, we observe that trading in a stock is temporarily halted if its price dynamics become too volatile or undergo an abrupt change exceeding a predefined level; see, e.g., \cite{ChenPetukhovXingWang}. Such a mechanism is called a (single stock) circuit breaker or, more simply, a trading halt. According to \cite{Kyle:88}, the purpose of such mechanisms is, in particular, to reduce price volatility and to protect customers from illiquid, disorderly markets. We note that in our reduced form model the liquidity level $\lambda$ can be viewed as a measure for market quality. Consequently, a trading halt is triggered when a market participant's order lets market quality deteriorate beyond a given threshold $\underline{\lambda}$, the liquidity trigger, i.e., at time
\begin{align}\label{eq:tauC}
    \tau^C:=\inf\left\{ t\geq 0: \min\lbrace\lambda^C_{t-}-|\Delta^l_t C|,\lambda^C_t\rbrace<\underline{\lambda}\right\},
\end{align}
with $\inf \emptyset =+\infty$.  In this case, the market order or the cancellation of a limit order that reduces liquidity to the level $\underline{\lambda}$ is partially executed with the available liquidity in the market $\lambda_t-\underline{\lambda}$ and trading is halted temporarily. Any orders that are sent thereafter cannot be executed and thus will have no effect on the state process. 

\begin{remark}
We briefly comment on how our stylized trading halt mechanism relates to the ones used in exchanges such as NYSE or Nasdaq to prevent excessive market movements. At the time of writing,  trading in a stock is halted in Nasdaq when within five minutes its price moves by more than 10\%.  

Lemma~\ref{lemma:link_volaliquidity} shows how volatility is linked to our market liquidity process $\lambda$ and clarifies how our liquidity-based trigger uses a bound on volatility to tame market fluctuations. Moreover, the proper functioning of a market necessitates a healthy amount of liquidity and the price band trigger of Nasdaq can also be viewed as a proxy indicator for a situation for poor liquidity levels. In our reduced-form model, such a dire market condition can be directly related to its liquidity process $\lambda^C$, so a passage time like $\tau^C$ is a natural time to halt trading. Of course, constructing this liquidity indicator from real market data is challenging (for example, what levels of the limit order books to use? How to address the difference of liquidity posted on bid vs.\ ask side?) and so the price band rule of Nasdaq can be viewed as a more readily implementable substitute for a liquidity oriented trigger of trading halts such as ours. Also, while not directly price oriented, our liquidity-based trading halts will often coincide with big price changes that breach a price band. 

We also observe that monitoring a rolling five minute price window leads to a path-dependency which severely complicates any control problem in this context by making the state-process infinite-dimensional. Our choice of liquidity-based trading halt trigger thus also significantly improves mathematical tractability.
\end{remark}

To formalize the dynamics of the state process with trading halt, let us denote it by
$$\tilde{S}^C_t:=(\tilde{\lambda}^C_t,\tilde{Q}^C_t,\tilde{P}^C_t,\tilde{X}^C_t).$$ Clearly, while the trading halt has not been triggered, the state process in the market with and without trading halt coincide so that 
\begin{align}
    \tilde{S}^C_t=S^C_t \text{ for } 0\leq t<\tau^C \wedge T.
\end{align} 
When a trading halt is triggered at time $\tau^C \leq T$, we proceed analogously to Section~\ref{section:controlled_model} and put
\begin{align}
    \tilde{S}^C_{\tau^C} := \tilde{\mathfrak{s}}(\Delta^l_{\tau^C}C,\Delta_{\tau^C} M^C,\Delta_{\tau^C} L^C;\tilde{S}^C_{\tau^C-}) \text{ and } \tilde{S}^C_{\tau^C+} := \tilde{\mathfrak{s}}(\Delta^r_{\tau^C}C,0,0;\tilde{S}^C_{\tau^C}) \text{ on } \{\tau^C \leq T\},
\end{align}
where $\tilde{\mathfrak{s}}$ is a suitably amended state update function that caps orders when they trigger the trading halt. Specifically, this capping is accomplished by
\begin{align}\label{eq:def_upsilon}
\Upsilon(\Delta,\lambda):=\begin{cases}
\sgn(\Delta)(\lambda-\underline{\lambda})^+,&\text{if }\lambda-|\Delta|<\underline{\lambda},\\\hfill
\Delta,&\text{if }\lambda-|\Delta|\geq\underline{\lambda},
\end{cases}
\end{align}
and, consequently,
\begin{align}
&\tilde{\mathfrak{s}}(\Delta,\eta,\rho;\textbf{s}_-):=(\tilde{\boldsymbol{\lambda}},\tilde{\mathbf{q}},\tilde{\mathbf{p}},\tilde{\mathbf{x}})(\Delta,\eta,\rho;\boldsymbol{\lambda}_-,\mathbf{q}_-,\mathbf{p}_-,\mathbf{x}_-)
\end{align}
with 
\begin{align} \tilde{\boldsymbol{\lambda}}&=
          \left(\boldsymbol{\lambda}_- -|\Delta|+\mathbbm{1}_{\{\boldsymbol{\lambda}_--|\Delta|\geq \underline{\lambda}\}}(-|\eta|+\rho)\right)\vee\underline{\lambda},\\
     \tilde{\mathbf{q}}&=\textbf{q}_-+\Upsilon(\Delta,\boldsymbol{\lambda}_-),\\
     \tilde{\mathbf{p}}&=\textbf{p}_-+I\left(\Upsilon(\Delta,\boldsymbol{\lambda}_-),\boldsymbol{\lambda}_-\right)+\mathbbm{1}_{\{\boldsymbol{\lambda}_--|\Delta|\geq \underline{\lambda}\}}I\left(\Upsilon(\eta,\boldsymbol{\lambda}_--|\Delta|),\boldsymbol{\lambda}_--|\Delta|\right),\\
     \tilde{\mathbf{x}}&=\textbf{x}_--\textbf{p}_-\Upsilon(\Delta,\boldsymbol{\lambda}_-)-\zeta|\Upsilon(\Delta,\boldsymbol{\lambda}_-)|-\Xi(\Upsilon(\Delta,\boldsymbol{\lambda}_-),\boldsymbol{\lambda}_-).
\end{align}
Trading is halted when market liquidity becomes too low, thus
\begin{align}
    \tilde{S}^C_t := \tilde{S}^C_{\tau^C+} \text{ for } \tau^C<t\leq T.
\end{align}
For the sake of completeness, let us also define the external market order flow $\tilde{M}^{C}_t$ in the market with trading halt as
\begin{align}
 \label{eq:MO_with_circuitbreaker}
    \tilde{M}^{C}_t&:=
        \int_{[0,t\wedge \tau^C]\times E\times\R_+}\Upsilon(\eta(e),\tilde{\lambda}^C_{s-}-|\Gamma_s(z(e,y))|)\mathbbm{1}_{\{y\leq g(\tilde{\lambda}^C_{s-})\}}N(ds,de,dy),
\end{align}
and proceed similarly for the liquidity providing limit orders 
\begin{align}\label{eq:LO_with_circuitbreaker}
    \tilde{L}^{C,+}_t&:=\int_{[0,t\wedge \tau^C]\times E\times\R_+}\rho^+(e)\mathbbm{1}_{\{\tilde{\lambda}^C_{s-}-|\Gamma_s(z(e,y))|\geq\underline{\lambda}\}\cap\{y\leq f(\tilde{\lambda}^C_{s-})\}}N(ds,de,dy),
\end{align}
and cancellations
\begin{align}\label{eq:cancel_with_circuitbreaker}
    \tilde{L}^{C,-}_t&:=\int_{[0,t\wedge \tau^C]\times E\times\R_+}\Upsilon(\rho^-(e),\tilde{\lambda}^C_{s-}-|\Gamma_s(z(e,y))|)\mathbbm{1}_{\{y\leq f(\tilde{\lambda}^C_{s-})\}}N(ds,de,dy).
\end{align}
The order flow is frozen after $\tau^C$ and at time $\tau^C$ the function $\Upsilon$ in \eqref{eq:MO_with_circuitbreaker} and \eqref{eq:cancel_with_circuitbreaker} and the indicator in \eqref{eq:LO_with_circuitbreaker} ensure that trading is halted after liquidity has been depleted beyond $\underline{\lambda}$. 

Let us finally record that $(\lambda^C,\tilde{\lambda}^C, \tilde{Q}^C,\tilde{X}^C,\tilde{P}^C,\tilde{L}^C,\tilde{M}^C)$ jointly solve the SDE-sytem:
\begin{align}
     {\lambda}^C_{0-} &= \lambda, & \diamond\; d\lambda^C &=  \diamond \; \left((-|dC^l|)\vec{+}(d\tilde{L}^C-|d\tilde{M}^C|)\vec{+}(-|dC^{r,c}|)\right),\\
    \tilde{\lambda}^{C}_{0-}&=\lambda', & \diamond \;d\tilde{\lambda}^C &= -\mathbbm{1}_{\{ \lambda^C \geq \underline{\lambda}\}} \diamond \left((-|dC^l|)\vec{+}(d\tilde{L}^C-|d\tilde{M}^C|)\vec{+}(-|dC^{r,c}|)\right),\\
        \tilde{Q}^C_{0-}&=q, & \diamond\; d\tilde{Q}^C&= \mathbbm{1}_{\{\lambda^C \geq \underline{\lambda}\}} \diamond dC, \\
    \tilde{X}^C_{0-}&=x, & \diamond\;d\tilde{X}^C &= -\tilde{P}^C \diamond d\tilde{Q}^C-\zeta \diamond |d\tilde{Q}^C|,\\
    \tilde{P}^C_{0-}&=p, & \diamond\; d\tilde{P}^C &= \iota(\tilde{\lambda}^C)\mathbbm{1}_{\{\lambda^C\geq\underline{\lambda}\}} \diamond \left(dC^l\vec{+}d\tilde{M}^C\vec{+}dC^{r,c}\right),\\ 
    \tilde{L}^C_{0-}&=0, & d\tilde{L}^C&=\mathbbm{1}_{\{y\leq f(\lambda^{C}_{-}), \; \lambda^{C}_- \geq \underline{\lambda}\}}\rho(e)N(dt,de,dy),\\
    \tilde{M}^C_{0-}&=0, & d\tilde{M}^C&=\mathbbm{1}_{\{y\leq g(\lambda^{C}_{-}), \; \lambda^{C}_-\geq\underline{\lambda}\}}\eta(e)N(dt,de,dy).
\end{align}
Also here the ``intermediate'' values $\lambda^C_t$, $\tilde{\lambda}^C_t$, $\tilde{Q}^C_t$, $\tilde{X}^C_t$, $\tilde{P}^C_t$ at time $t$ (rather than $t+$) can be obtained from these dynamics by considering the Marcus-style dynamics at the moment when its fictitious flow at time $t$ has run two thirds of its course (i.e., when, in the notation of Remark~\ref{rem:Marcus}, at fictitious time $s=2/3$ two of the three jump contributions at time $t$ have been taken care of).

At the terminal time $T$, any remaining inventory $\tilde{Q}^C_{T+}$ has to be liquidated. If trading has been halted or will be during this liquidation, an auction determines the new mid-price from which the remaining liquidation is executed with price impact as before. For simplicity, we assume that this new mid-price deviates from the pre-halt mid-price by a shock
\begin{align}\label{eq:YProperties}
    Y \text{ independent of $\mathcal{F}_T$ with  }  \mathcal{L}_Y(a):=\log \E[\exp(a Y)]<\infty \text{ symmetric in } a \in \R.
\end{align}
For instance, any centered and independent Gaussian $Y$ will do. As a result,   our trader's terminal wealth is
\begin{align}\label{eq:terminalWealth}
   \tilde{W}^C_{T+}:=\w\left(\tilde{X}^C_{T+},\tilde{Q}^C_{T+}, \tilde{P}^C_{T+},\tilde{\lambda}^C_{T+}\right)+Y\rr(\tilde{Q}^C_{T+},\tilde{\lambda}^C_{T+}).
\end{align}
Here, the function $\w$ of~\eqref{eq:liquidationWealthFunction} provides the cost-adjusted liquidation value  of the investor's position. The $Y$-term adjusts this further to account for the part 
of the investor's position that goes into auction and is thus additionally exposed to the mid-price shock $Y$. This part is given by the function 
\begin{align}\label{eq:2bAuctioned}
    \rr(q,\lambda):= \left(|q|-(\lambda-\underline{\lambda})^+\right)^+.
\end{align}
Indeed, if at liquidity level $\lambda$ a trade of size $q$ has to be executed, this will be possible without triggering a halt only if $|q|$ does not exceed the available liquidity $(\lambda-\underline{\lambda})^+$. Otherwise, this available liquidity is exhausted and the remaining bits of the order $\rr(q,\lambda)$ are auctioned.

\subsection{The utility maximization problem}\label{section:utilityMaximization}

Let us consider a trader who maximises her expected utility from terminal wealth over the trading horizon $[0,T]$ as given by the performance criterion
\begin{align}
    J(C):=\mathbb{E}[U(\tilde{W}^C_{T+})] \quad \text{ with } \quad U(w)=-\exp(-\alpha w),
\end{align}
where $\tilde{W}^C_{T+}$ in~\eqref{eq:terminalWealth} specifies the investor's terminal wealth and $\alpha>0$ measures her (constant) absolute risk aversion. 

For continuity of the resulting value function, the trader can take at any time the remaining liquidity without triggering a trading halt and to trade exactly to the extreme inventory levels $-\bar{q}$ and $\bar{q}$. Obviously, this is not a problem when any order size is allowed,  but when trading is in lots of size $\delta$, this may not be possible. Thus, we let 
\begin{align}
    D(\lambda,q):=\mathbb{R} \quad \text{ or } \quad D(\lambda,q):=\{...,-2\delta,-\delta,0,\delta,2\delta,...\}\cup \{\lambda-\underline{\lambda},\underline{\lambda}-\lambda,\bar{q}-q,-\bar{q}-q\}
\end{align}
and define the set of admissible strategies when starting with liquidity $\lambda$ and in position $q$ as
\begin{align}\label{def:tildeC}
    \tilde{\mathcal{C}}(\lambda,q):=\big\{C &\text{ $\Lambda$-measurable of bounded variation with $C_{0-}=0$ satisfying~\eqref{eq:inventorybound} } \\&\text{ and } \Delta^l_t C \in D(\lambda^C_{t-},Q^C_{t-}), \;\Delta^r_t C \in D(\lambda^C_{t},Q^C_{t}) \text{ for } t \in [0,T]\big\}.
\end{align}

Now, the value function can be defined as
\begin{align}\label{eq:value_fct}
v(T,\mathbf{s})=v(T,\lambda,q,p,x):=\sup_{C\in\tilde{\mathcal{C}}(\lambda,q)}\E\left[U(\tilde{W}^C_{T+})\right].
\end{align}
\begin{thm}[Regularity of the value function]\label{thm:continuity_all}
The value function $v$ in~\eqref{eq:value_fct} for our utility maximization problem with trading halt has the multiplicative structure
\begin{align}\label{eq:multiplicativeValueFunction}
   v(T,\lambda,q,p,x)=\exp(-\alpha(x+pq))v_0(T,\lambda,q),
\end{align}
where $v_0=v(.,.,.,0,0)$ is the value function for initial data $x=0$, $p=0$.
Moreover, it is non-degenerate in the sense that it is strictly less than $U(+\infty)=0$. 

Suppose that the distributions of cancellations $\nu \circ (\rho^-)^{-1}$ and of external market orders $\nu \circ |\eta|^{-1}$ at most have an atom at $0$ and both exhibit a finite second moment. Assume, in addition, that the posting of market orders is controlled in the sense that their maximum arrival rate is finite, i.e., $g(\infty)<\infty$. Finally, let the size of limit orders be controlled in the sense that 
\begin{align}\label{eq:rhoExponentialIntegrable}
    \int_E \exp(c\rho^+(e))\nu(de)<\infty
    \quad\text{ for all }c\in [0,\infty).
\end{align}
Then the value function $v=v(T,\lambda,q,p,x)$ is continuous on $[0,\infty) \times [\underline{\lambda},\infty) \times [-\bar{q},\bar{q}] \times \R^2$.
\end{thm}

We defer the proof of this result to the Appendix~\ref{sec:regularity}. One cannot follow standard arguments in our setting because the external order flow, at least pathwise, will not depend continuously on the control: For instance, small differences in liquidity levels can allow external orders to affect one strategy, but not another one which has almost the same liquidity level. This may further widen the difference in liquidity for these strategies, making it ever more difficult for one of them to, for instance, copy the other's trades and do so under similar market conditions and thus with similar profitability. For reasons like this, the relevant analysis is rather challenging and requires some care; see in particular Lemma~\ref{lem:strategyComparison} below.

Let us briefly comment on the extra assumptions we need for our continuity result for the value function. Continuity of the distributions of cancellation and market order sizes is needed because with an atom say at $1$ lot, a trader with liquidity $\underline{\lambda}+1$ has no risk of seeing a halt triggered by an external order whereas there is a strictly positive frequency of this disruption happening when liquidity is just below this level. The integrability and frequency assumptions on external order flow are of little practical concern and are made for technical reasons as we need to bound the order flow fluctuations; see in particular the proof of Corollary~\ref{cor:continuityLiquidity}.

In Section~\ref{section:HJB}, we show how one can use standard techniques from dynamic programming and the theory of Hamilton--Jacobi--Bellman (HJB) equations to derive the numerical approximations of solutions to the trader's problem which we describe in the next section.

\subsection{Performance of signal-based trading strategies}\label{section_illustration}
In the following, we report the results from a numerical scheme to solve the dimension-reduced HJB \eqref{eq:HJB_reduced} from Subsection \ref{chapter:HJB} below and illustrate the trader's signal-based strategy in an optimal acquisition problem --- Subsection \ref{chapter:numericalscheme} provides details of the numerical scheme we used.

In our benchmark, the trader receives, with probability $\hat{p}$,  a private signal about limit orders, market orders, and cancellations of limit orders. Compared to a trader without the signal, the trader with the signal optimizes her execution times and trading volumes by splitting the parent order into fewer and larger child orders than those of the trader without signal. She executes her orders when a signal about a liquidity taking order arrives.
Signals about liquidity provision are not relevant for the execution programme of a risk averse trader
because there is no incentive to execute before the liquidity providing order increases the liquidity
in the book and price impact of market orders decreases.

Speculative trades are not profitable in the benchmark due to the size of the bid-ask spread. Whereas when the bid-ask spread is narrow, speculation based on signals is more profitable and the trader generally initiates speculative trades upon receiving a signal about liquidity provision and the level of liquidity is low. After waiting for liquidity to improve and the cost of price impact to decline, she unwinds the speculative position upon receiving a signal about the imminent arrival of a liquidity taking order.
\paragraph{Parameter specification --- Benchmark case.}
In the simulations, we consider six types of external market orders: buy orders of one, two, and three lots; and sell orders of one, two and three lots. Similarly, limit orders and their cancellations are of size one, two, and three lots. 

We choose the arrival rates $f$ and $g$ as
\begin{align}\label{eq:def_fg}
f(\lambda)=\theta_f\exp(-\kappa_f\,\lambda)\quad \text{and}\quad g(\lambda)=\theta_g\exp(\kappa_g\,\lambda),
\end{align}
where $\theta_f=40$, $\theta_g=20$,  and $\kappa_f=\kappa_g=0.01$, i.e., if liquidity $\lambda$ is at level 0, the market expects to receive 20 external market orders, 30 limit orders, and 10 cancellations of limit orders over the trading horizon of length $T=1$. 

The price impact function is given by \begin{align}\label{eq:iota_casestudy}
\iota(\lambda)=\theta_\iota+\kappa_\iota\lambda,
\end{align}
where $\theta_\iota=0.01$ and $\kappa_\iota=-0.0002$. 

For numerical reasons, we introduce an upper bound $\overline{\lambda}$  for the liquidity parameter $\lambda$ so that we obtain a bounded domain $[\underline{\lambda},\overline{\lambda}]$ for $\lambda$, where $\underline{\lambda}$ is the lower bound from the trading halt mechanism as in \eqref{eq:tauC}.
We set the lower bound at $\underline{\lambda}=-40$ and the upper bound at $\overline{\lambda}=40$, such that $\iota(\lambda)\geq 0$ for every $\lambda\in[\underline{\lambda},\overline{\lambda}]$. 
Moreover, we consider the inventory $q$ in the bounded domain $[-\overline{q},\overline{q}]$ with $\overline{q}:=10$.

With these parameter values, the elasticity condition \eqref{eq:elasticitycond} holds. More precisely, for the given specification of $\iota$ as in \eqref{eq:iota_casestudy}, the denominator $-\frac{\partial_\lambda \mathbb I^2(\lambda)}{\mathbb I^2(\lambda)}$ in the elasticity condition \eqref{eq:elasticitycond} lies in the range $\approx [0.022, 0.177]$ for $\lambda\in[-40,40]$. Hence, condition~\eqref{eq:elasticitycond} requires $\smash{\kappa_f=\frac{\partial_\lambda g(\lambda)}{g(\lambda)}<0.022}$, and thus yields an upper bound for this exponential arrival rate of market orders. \label{page:comment_elasticity_casestudy}

The independent auction mark-up $Y$ on the mid-price in~\eqref{eq:terminalWealth} is $N(0,\sigma^2)$-distributed with  
$\sigma=0.3$. The bid-ask spread is set to one cent, i.e., $2\zeta=0.01$.

Finally, the trader's risk aversion is $\alpha=0.1$.

\paragraph{Signal design.} The signal is given by the variable $Z_t$. A value of $Z_t=1$ signals an incoming limit order and a value of $Z_t=-1$ signals liquidity taking, i.e., an incoming market order or the cancellation of a limit order. With fixed probability $\hat{p}\in[0,1]$, the signal alerts the trader of an imminent order and so an external order has a $1-\hat{p}$ chance of taking the trader by surprise.

Thus, the signal informs the trader about the sign of the imminent change in liquidity, but does not provide any information about the size of an order, i.e., the trader anticipates the direction, but not the magnitude of changes in liquidity. Particularly, in the case of a signal $Z_t=-1$ about liquidity taking, the trader does not know whether a buy market order, a sell market order, or the cancellation of a limit order will be posted.

Note that definition \eqref{def:signal} allows for a huge variety of conceivable signals and the above specification is considered to simplify the implementation.

\paragraph{Formal model specification as in Section \ref{section_illustration}.}\label{page:formalmodel}
In this part, we follow the notations in Chapter \ref{chapter:model}. For simplicity, we only treat the benchmark; the specification of the remaining settings are similar.

First, the mark space is given by $E=\{-3,-2,-1,1,2,3\}\times\{-3,-2,-1,1,2,3\}\times\{-1,1\}$. The components $e_1$ and  $e_2$ of a mark $e=(e_1,e_2,e_3) \in E$ describe the volume of a possible market and of a limit order, respectively;  component $e_3$ decides which of these two possibilities materializes. The mappings $\eta$ and $\rho$ are $$\eta(e)=e_1\mathbbm{1}_{\{-1\}}(e_3)\quad\text{and}\quad \rho(e)=e_2\mathbbm{1}_{\{1\}}(e_3) \text{ for } e=(e_1,e_2,e_3) \in E.$$ 
The model is driven by the homogeneous Poisson point process $N(dt,de,dy)$ with compensator $dt\otimes\nu(de)\otimes dy$. Here, we have $\nu(de)=P^1(de_1)\otimes P^2(de_2)\otimes P^3(de_3)$, where $P^1$ is the distribution of the mark that determines the market order size, $P^2$ is the distribution of the mark that determines the limit order size and $P^3$ is a coin flip deciding whether a market or a limit order arrives in the book.

In our case study, we use the specifications
\begin{align}
      &P^1(\{-3\})=P^1(\{3\})=0.15,\;\; P^1(\{-2\})=P^1(\{2\})= P^1(\{-1\})=P^1(\{1\})=0.175,\\
      &P^2(\{-3\})=0.1,\;\, P^2(\{-2\})=P^2(\{-1\})=0.075,\;\, P^2(\{1\})=P^2(\{2\})=P^2(\{3\})=0.25,\\
      &P^3(\{-1\})=P^3(\{1\})=0.5.
 \end{align}

Here, as the arrival rates of market and limit orders are halved by the coin flip, we double $f$ and $g$ as in \eqref{eq:def_fg}, i.e., take $\hat{f}=2f$ and $\hat{g}=2g$, to ensure that, in expectation, the market still receives 20 external market orders, 30 limit orders, and 10 cancellations of limit orders over the trading horizon $[0,1]$.

The signal process $Z$ is
\begin{align}
    Z_t&:=\int_{\{t\}\times E \times \R_+}z(e,y)N(ds,de,dy) \text{ with }z(e,y):=(e_3,y).
\end{align}
We disintegrate for $z=z(e,y)\in\{-1,1\}\times \mathbb{R}_+$ the jump measure
\begin{align}
    \nu(de)\otimes dy=\int_{\{-1,1\} \times [0,\infty)}K(z;de,dy)\mu(dz),
\end{align}
into 
\begin{align}
    \mu(dz)&= (P^3(\{-1\})\text{Dirac}_{-1}+P^3(\{1\})\text{Dirac}_{1})(dz_1)\otimes dy,\\
    K(z_1,y;.,.)&= P^1\otimes P^2\otimes \text{Dirac}_{z_1}\otimes \text{Dirac}_y,
\end{align}
where we generically denote $z=(z_1,y)$.

\paragraph{Value function and certainty equivalent.} We implement the numerical scheme from subsection~\ref{chapter:numericalscheme} to solve the HJB~\eqref{eq:HJB_reduced}. Figure \ref{fig:value_fct_p02} shows the value function $v$ as in \eqref{eq:value_fct} for $p=x=0$ and when the probability of receiving a signal is $\hat{p}=0.2$.
\begin{figure}
\centering
\begin{minipage}{.5\textwidth}
\hspace{-1.5cm}
  \includegraphics[scale=0.7]{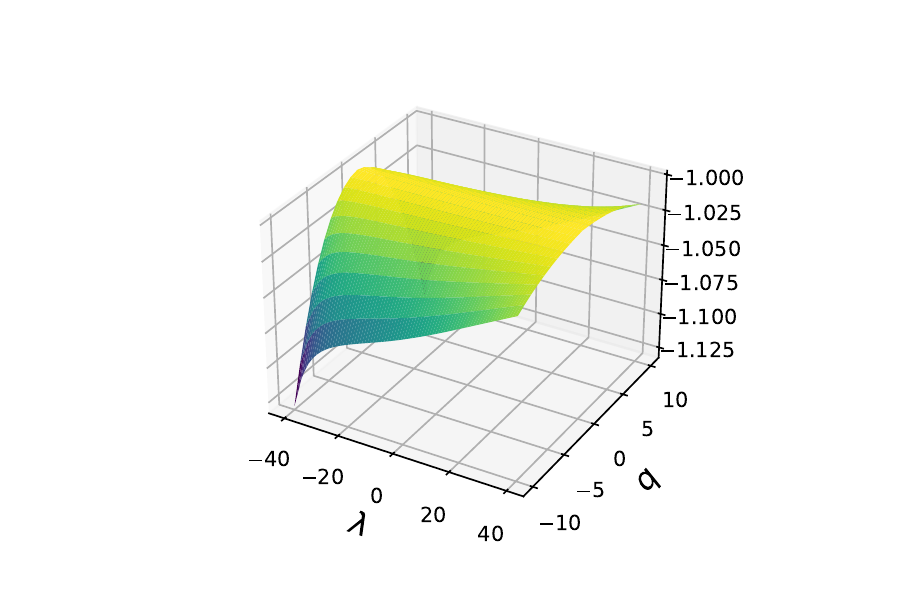}
    \caption{Value function $v$ as in \eqref{eq:value_fct} for $p=x=0$.}
    \label{fig:value_fct_p02}
\end{minipage}%
\begin{minipage}{.5\textwidth}
  \hspace{-1.5cm}
  \includegraphics[scale=0.7]{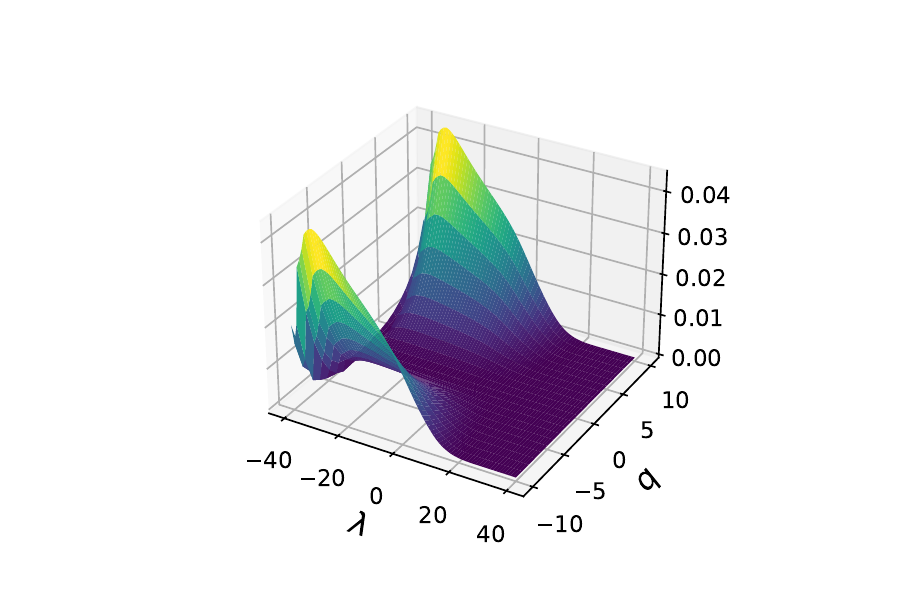}
\caption{Certainty Equivalent for a trader who receives a signal with probability $\hat{p}=0.2$ compared to a trader without signal for spread~$0.01$.}
\label{fig:VF_certaintyequiv}
\end{minipage}
\end{figure}
The value function is symmetric about $q=0$ due to the symmetry of \eqref{eq:definition_I} and \eqref{eq:definition_Xi} with respect to buy and sell orders. Hence, the expected utility of terminal wealth for an acquisition or liquidation programme is the same --- everything else being equal. 

Moreover, the value function is increasing in the liquidity $\lambda$ because execution costs decrease as the liquidity in the market increases. When liquidity approaches the liquidity trigger 
$\underline{\lambda}$, the value function is the steepest because of the additional risk of entering an auction through a trading halt.

Similarly, the value function is particularly steep for large values of $|q|$ because a large execution programme is linked to higher risk in price and liquidity changes.

Figure \ref{fig:VF_certaintyequiv} illustrates the signal-specific certainty equivalent of the value function when the probability of receiving a signal is $\hat{p}=0.2$. 
It corresponds to the additional amount of initial wealth that is needed when the trader does not receive a signal to achieve the same expected utility as in the case with signal, i.e., the certainty equivalent $CE$ is such that 
\begin{align}\label{eq:def_CE_1}
v_{\hat{p}=0}(T,q,\lambda,p,x+CE)=v_{\hat{p}=0.2}(T,q,\lambda,p,x).
\end{align}
We rearrange \eqref{eq:def_CE_1} to write 
\begin{align}\label{eq:def_CE}
CE=-\frac{1}{\alpha}\log\left(\frac{v_{\hat{p}=0.2}(T,q,\lambda,p,x)}{v_{\hat{p}=0}(T,q,\lambda,p,x)}\right).
\end{align}

The certainty equivalent is symmetric about $q$ and does not depend on $p$ or $x$. Moreover, the signal is worth the most for large values of inventory, long or short. At most, the trader can make up to four cents from the information of the signal, which corresponds to four times the bid-ask spread, see Figure \ref{fig:VF_certaintyequiv}. 

The certainty equivalent is decreasing in the liquidity component $\lambda$ because as liquidity decreases, the execution times become more relevant because price impact is more costly. 

Moreover, the certainty equivalent increases slightly towards the liquidity trigger $\underline{\lambda}=-40$. Here, the signal warns the trader of potentially reaching the lower bound $\underline{\lambda}$ and reduces the risk of executing a final trade in an auction.

For liquidity greater than twenty, the certainty equivalent is zero, i.e., the signal does not add any value because, with and without a signal, a trader will immediately complete the execution programme.

\paragraph{Benchmark --- The signal optimises times of execution and trading volume.}
\begin{figure}[t!]
    \centering
\includegraphics[scale=0.5]{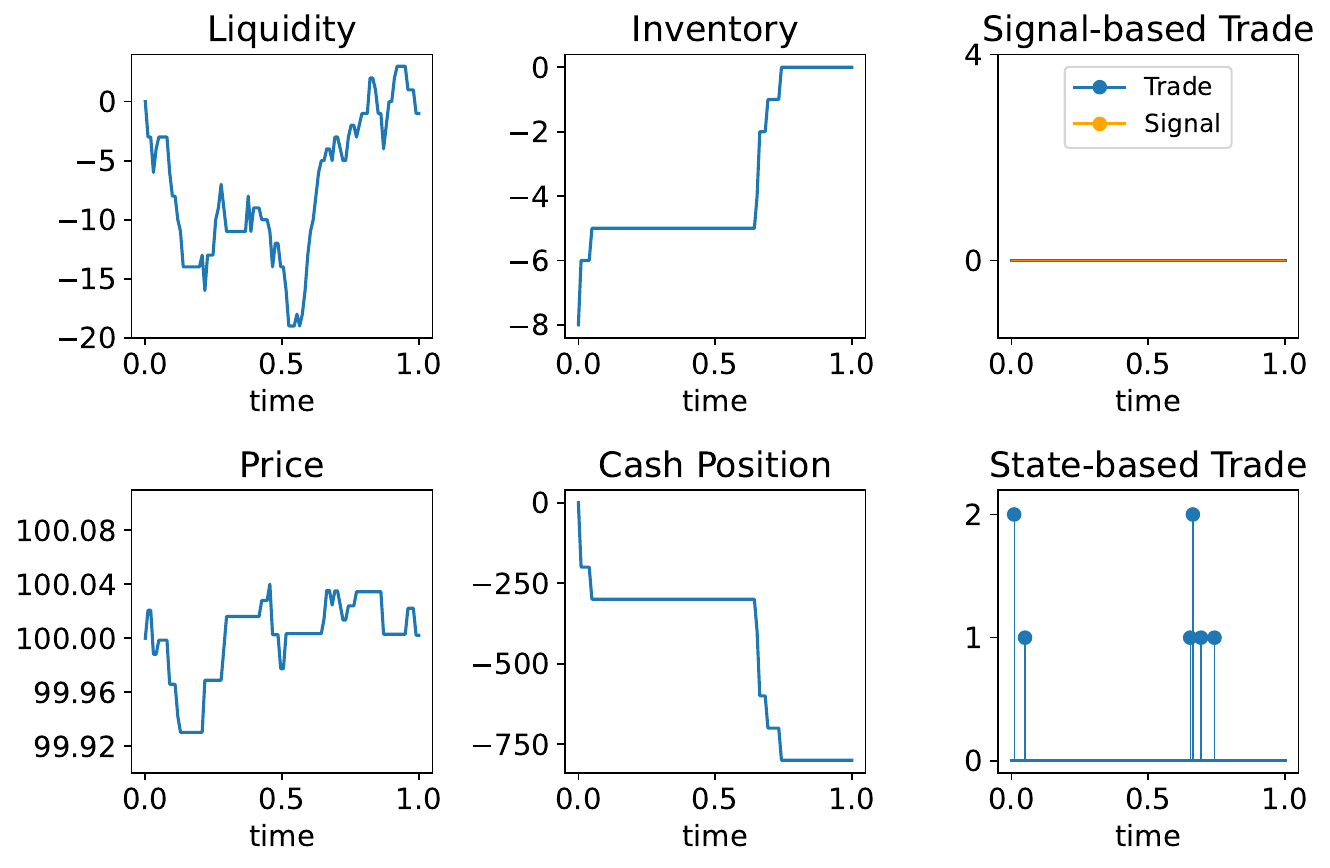}
    \caption{Optimal acquisition of eight lots. Pathwise plot when trader does not receive a signal and with spread 0.01.}
    \label{fig:paths_execution_0}
\end{figure}
\begin{figure}[t!]
    \centering
\includegraphics[scale=0.5]{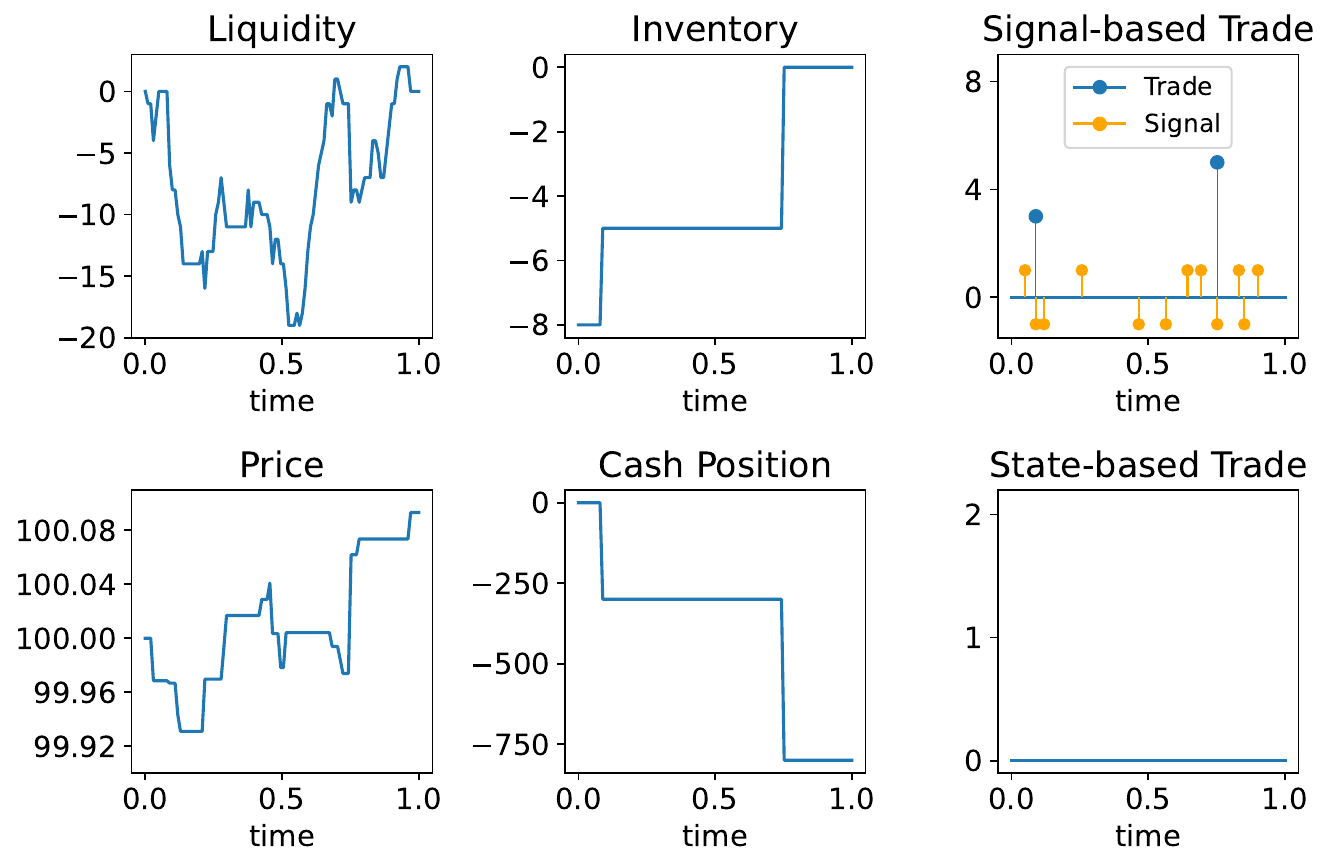}
    \caption{Optimal acquisition of eight lots. Pathwise plot when trader receives signal with probability $\hat{p}=0.2$ and with spread 0.01.}
\label{fig:paths_execution_02}
\end{figure}
\begin{figure}[t!]
\centering\vspace{-1cm}
\begin{minipage}{.49\textwidth}
\includegraphics[scale=0.5]{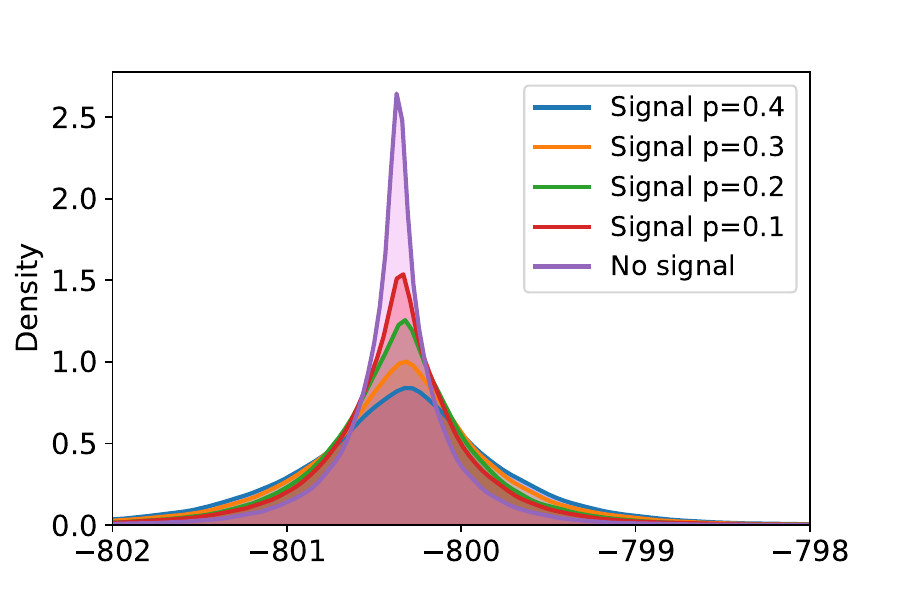}
\end{minipage}
\begin{minipage}{.49\textwidth}\vspace{21pt}
   \includegraphics[scale=0.5]{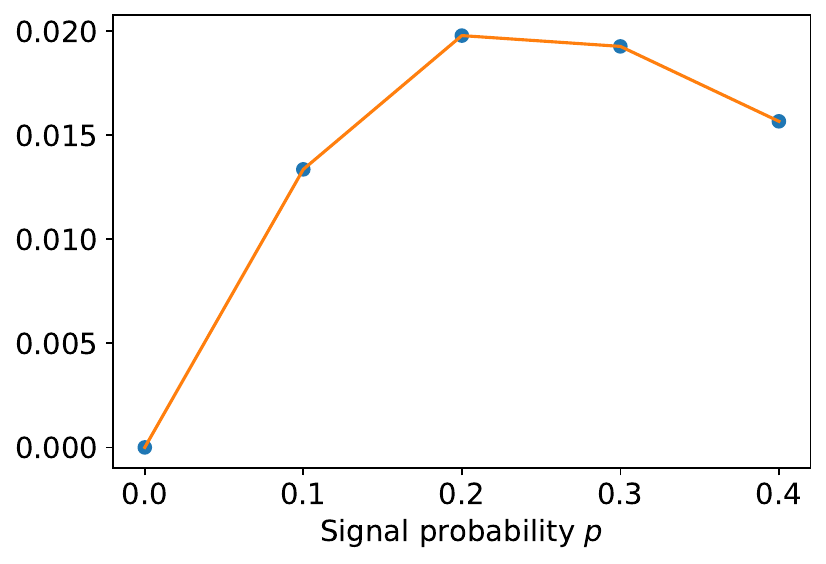}
\end{minipage}
\caption{Optimal acquisition of eight lots. Performance of trader without signal compared to a trader who receives a signal with probability $\hat{p}\in\{0,1,0.2,0.3,0.4\}$ and with spread 0.01.\\
    Left: Distribution of terminal wealth.\; Right: Signal-Sharpe-ratio.}
    \label{fig:hist_ssr_benchmark}
\end{figure}
In the benchmark with spread 0.01, we look at the optimal acquisition problem where the trader acquires eight lots of the asset over the time horizon $[0,T]$, so the initial inventory is $q=-8$. The trader can buy and sell any multiple of lots and can execute speculative trades by trading away from the target of buying only eight lots over the trading window.

When there is no private signal, the trader sends child orders of at most two lots when liquidity is sufficiently high, see Figure \ref{fig:paths_execution_0}. However, when the trader receives private signals, she times her executions with those of the signal with information about liquidity taking orders and optimizes the trading volume of her trades, see Figure \ref{fig:paths_execution_02} (simulated with same seed as that for results in Figure \ref{fig:paths_execution_0}). Here, she executes her acquisition programme by sending two child orders of sizes three and five lots upon receiving a signal about liquidity taking when the level of liquidity is high. When liquidity is low during $t\in[0.15,0.7]$, the trader does not act on the signal and prefers to wait for the liquidity to increase. Note that only the signals about liquidity taking orders are relevant in her execution programme.

\paragraph{More signals lead to lower execution costs and riskier strategies.}
Figure \ref{fig:hist_ssr_benchmark} illustrates the performance of signal-based strategies depending on the probability of receiving a signal. 

Figure \ref{fig:hist_ssr_benchmark} (left) shows the distribution of terminal wealth for different probabilities of receiving a signal compared to the distribution without a signal. When the probability of receiving a signal increases, the densities shift to the right, i.e., the signal increases the expected terminal wealth by optimizing the execution. On the other hand, as the probability of receiving a signal increases, also the variance of terminal wealth increases because the trader takes more risks due to the additional information through the private signal. More precisely, knowing that she will be informed about liquidity shocks through the signal, the trader waits longer for liquidity to recover from shocks before trading to minimize price impact costs. However, this is linked to more risk because unsignaled external orders can still arrive which increases the variance of the trader's terminal wealth.

To quantify the value of the private signal $Z$, we introduce the Signal-Sharpe-ratio (SSR)
\begin{align}
   SSR(Z):= \frac{\bar{W}(Z)-\bar{W}(0)}{\sigma(Z)},
   \end{align}
   where 
   \begin{align*}
    \bar{W}(Z):=\frac{1}{n_{sim}}\sum_{j=1}^{n_{sim}}\tilde{}W^j_{T+}(Z),\text{  }\bar{W}(0):=\frac{1}{n_{sim}}\sum_{j=1}^{n_{sim}}\tilde{W}^j_{T+}(0)
    \end{align*}
    and  
    \begin{align*}\sigma(Z):=\sqrt{\frac{1}{n_{sim}}\sum_{j=1}^{n_{sim}}(\tilde{W}^j_{T+}(Z)-\bar{W}(Z))^2}.
\end{align*}
Here, $\tilde{W}^j_{T+}(Z)$ resp.\ $\tilde{W}^j_{T+}(0)$ denote the terminal wealth for scenario $j\in\{1,...,n_{sim}\}$ with signal $Z$ and without signal.

Precisely, the SSR is the excess return $\bar{W}(Z)-\bar{W}(0)$ for a trader with signal $Z$ compared to a trader without signal weighted by the risk $\sigma(Z)$ of a trader with signal $Z$.

Figure \ref{fig:hist_ssr_benchmark} (right) shows the SSR as a function of the probability $\hat{p}$ of receiving a signal. The SSR increases up to $\hat{p}=0.2$ and slowly decreases thereafter because the increase in variance dominates the increase in expected terminal wealth, which leads to a non-monotonicity in Figure~\ref{fig:hist_ssr_benchmark}~(left). Note that the mean-variance optimization is not the objective of the trader who maximizes expected utility of terminal wealth; see~\eqref{eq:value_fct}.
\paragraph{Signal is more valuable for speculation as bid-ask spread narrows.}
Next, consider a bid-ask spread of size 0.002, i.e., the bid-ask spread is one fifth of that in the benchmark. In this case, the signal incentivises speculative trades because, everything else being equal, the costs to execute round-trip trades are lower.

The speculative round-trip trades start after receiving a signal on liquidity provision, i.e., when the trader knows through the signal that the speculation can be unwound at better liquidity, see Figure~\ref{fig:paths_execution_smallspread}. After triggering liquidity provision through her own market order, the trader waits for liquidity to arrive in the market to unwind the speculative position with better liquidity upon receiving a signal about a liquidity taking order.

\begin{figure}[t!]
    \centering
\includegraphics[scale=0.5]{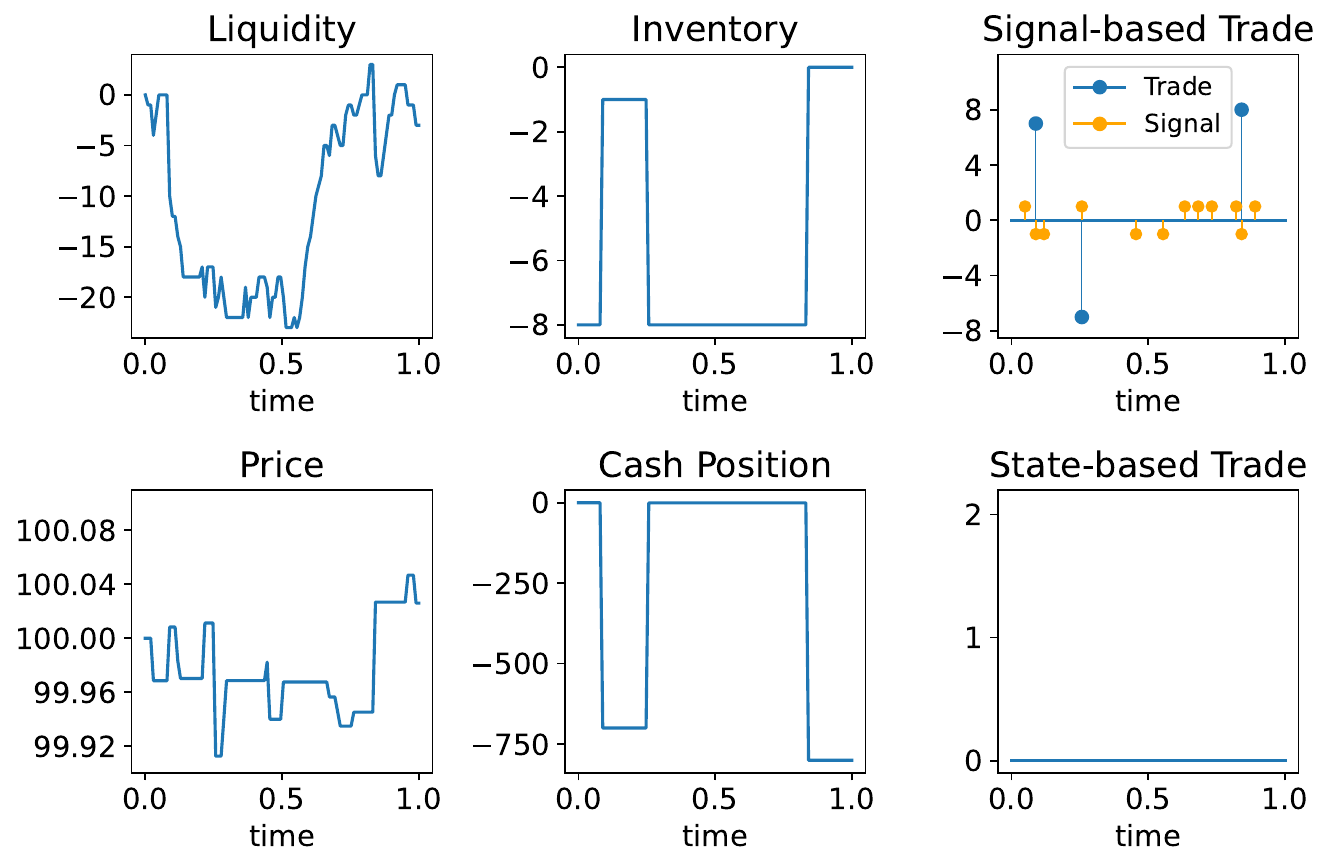}
    \caption{Optimal acquisition of eight lots. Pathwise plot when trader receives a signal with probability $\hat{p}=0.2$ and with spread $0.002$.}
    \label{fig:paths_execution_smallspread}
\end{figure}

\noindent Next, we study the number of paths with speculative trades among $n_{sim}=100,000$ simulated paths with bid-ask spread 0.01 and 0.002 in the acquisition ($q=-8$) and pure speculation ($q=0$) scenario where the paths for wide and narrow spreads are simulated with the same seed.\footnote{We say that a path contains a speculative trade if the trader trades away from the target of buying eight lots over the time horizon $[0,1]$ by trading with both buy and sell market orders.}

When the spread is wide, i.e., when the spread is $0.01$, there is no speculation because round-trip trades are too costly. On the other hand, when the spread is narrow, i.e., when the spread is $0.002$, there is speculation in about $22\%$ of paths in the acquisition problem~($q=-8$) and in about $11\%$ of paths in the pure speculation scenario ($q=0$).\label{page:22percent} There are more speculative paths in the execution example because the trader's execution of trades triggers liquidity provision and speculative trades become profitable. Especially, when the trader trades towards a position close to zero very early in the trading window, i.e., completes the acquisition problem early on, the remaining time horizon is long enough so that a speculative roundtrip starting close to zero is profitable; and vice versa if the trader completes the acquisition at a later point in the trading window.

Finally, without signal, there is no speculation for both the wide and the narrow bid-ask spread.

\begin{figure}[t!]
\vspace{-1cm}
\begin{minipage}{0.48\textwidth}
\centering
\vspace{1.4cm}
    \includegraphics[scale=0.5]{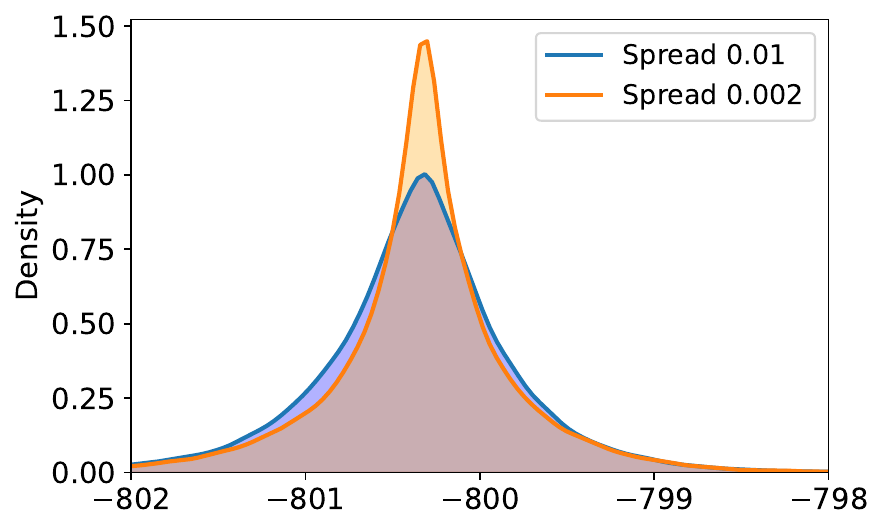}\vspace{0.5cm}
    \caption{Optimal acquisition of eight lots. Distribution of terminal wealth when trader receives a signal with probability $\hat{p}=0.2$ with a spread of $0.01$ and $0.002$.}
    \label{fig:hist_spread}
\end{minipage}
\begin{minipage}{0.48\textwidth}
\hspace{-1.8cm}
    \includegraphics[scale=0.7]{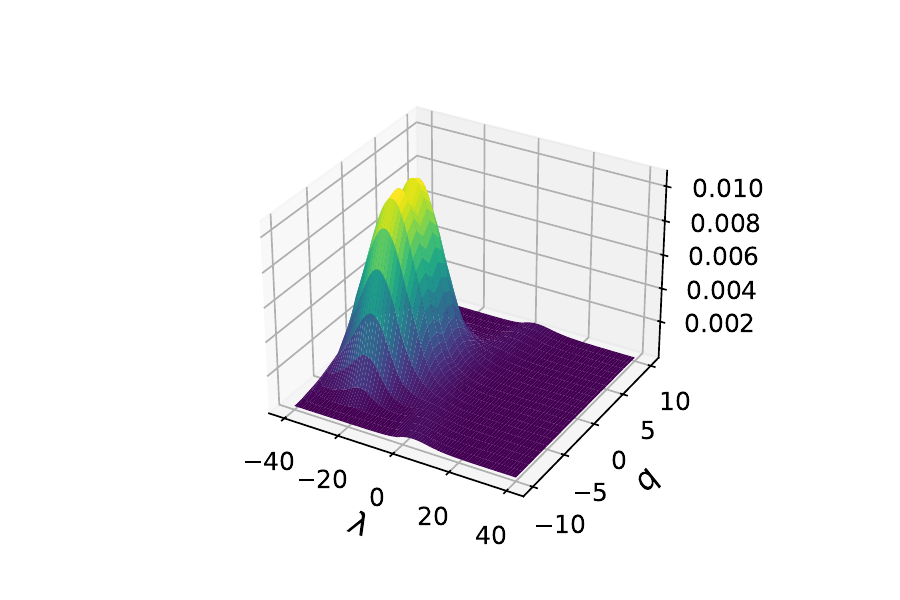}\vspace{-0.5cm}
    \caption{Certainty equivalent for a small spread of $0.002$ minus the certainty equivalent for a spread of $0.01$.}
    \label{fig:CE_spread}
\end{minipage}
\end{figure}Figure \ref{fig:hist_spread} shows the distribution of terminal wealth in the optimal acquisition problem of $q=-8$ when the trader receives a private signal with probability $\hat{p}=0.2$, and the bid-ask spread is 0.01 and 0.002. With a narrow spread, the terminal wealth is, on average, higher because the trader pays less spread and profits from speculative trades. The terminal wealth increases the most in the interval $[-801,-800.5]$, see Figure \ref{fig:hist_spread}, because low level of liquidity, which cause smaller values of terminal wealth for a spread of $0.01$ due to price impact costs, can lead to profitable roundtrips when the spread is at $0.002$.
Similarly, the variance of terminal wealth decreases because in general, the trader executes more orders to bring her inventory to zero and only deviates from this target in about $22\%$ of cases. 

Similarly, Figure \ref{fig:CE_spread} illustrates the certainty equivalent as in \eqref{eq:def_CE} for spread 0.002 minus the certainty equivalent as in \eqref{eq:def_CE} for spread 0.01. The certainty equivalent improves when liquidity is low and when the trader's inventory is close to zero which is when speculative trades are executed. 
\begin{figure}[t!]
\centering
\includegraphics[scale=0.5]{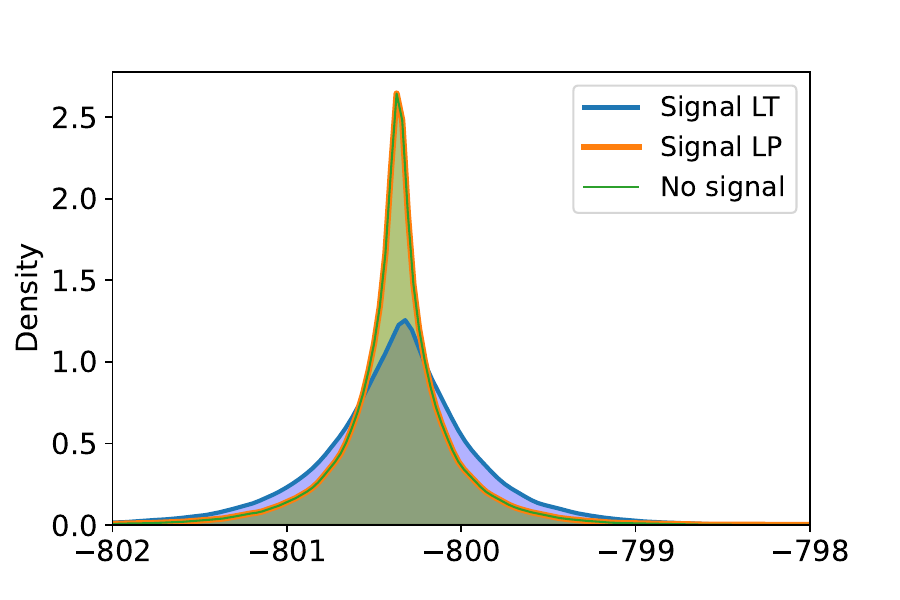}
  \caption{Optimal acquisition of eight lots. Distribution of terminal wealth for trader who receives a signal on liquidity taking~(LT) with probability $\hat{p}=0.2$ vs a trader who receives a signal on liquidity provision~(LP) with probability $\hat{p}=0.2$; spread is 0.01.}
 \label{fig:hist_sigMOLO}
 \end{figure}
\paragraph{Signal on liquidity taking orders vs signal on liquidity providing orders.}

Finally, we compare the performance of the strategy of a trader who receives a signal about the arrival of liquidity taking orders against the strategy of a trader who receives a signal about the arrival of liquidity provision. 

The trader with a signal on liquidity taking uses the signal to optimize the execution times and trading volumes, which increases the terminal wealth and its variance, see Figure \ref{fig:hist_sigMOLO}. Whereas, similarly to Figure \ref{fig:paths_execution_02}, the trader does not use the signal about liquidity provision for her execution so her terminal wealth coincides with that of the trader who receives no signal; see Figure \ref{fig:hist_sigMOLO}.

\section{The Hamilton--Jacobi--Bellman equation and its numerical solution}\label{section:HJB}

In this section, we present the derivation of the HJB and we present the numerical scheme we used for our illustrations in the preceding section.

\subsection{The Hamilton--Jacobi--Bellman equation}\label{chapter:HJB}

In the following, we derive the HJB equation that is satisfied by the value function $v$ as in \eqref{eq:value_fct}. Particularly, we investigate for which conditions the value process $V^C_t:=v(\bar{T}-t,\tilde{S}^C_{t+})$, $t\in[0,\bar{T}]$ with arbitrary, but fixed time horizon $\bar{T}$, satisfies super-martingale dynamics for every admissible $C\in\tilde{\mathcal{C}}(\lambda,q)$ and martingale dynamics for some  strategy $C^*\in\tilde{\mathcal{C}}(\lambda,q)$ (which will then be optimal). 

We adopt the notation from the representation of $C$ in Lemma~\ref{Lemma:decomposition_q} and suppose that the value function $v=v(T,\lambda,q,p,x)$ is sufficiently smooth to apply Itô's formula. We then can write at least formally for $0\leq t \leq\tau^C\wedge \bar{T}$ the dynamics of $V^C$ as
\begin{equation}\label{eq:Itoformula_V}
\begin{aligned}
    &dV^C_t=-\frac{\partial v}{\partial T}(\bar{T}-t,\tilde{S}^C_{t-})dt+\left(\frac{\partial v}{\partial p}(\bar{T}-t,\tilde{S}^C_{t-})\iota(\tilde{\lambda}^C_{t-})-\frac{\partial v}{\partial x}(\bar{T}-t,\tilde{S}^C_{t-})\tilde{P}^C_{t-}+\frac{\partial v}{\partial q}(\bar{T}-t,\tilde{S}^C_{t-})\right)dC^c_t\\
    &+\left(\frac{\partial v}{\partial x}(\bar{T}-t,\tilde{S}^C_{t-})\zeta-\frac{\partial v}{\partial \lambda}(\bar{T}-t,\tilde{S}^C_{t-})\right)|dC^c_t|\\
    &+\hspace{-9pt}\int\limits_{E\times \{z\not=0\}}\hspace{-9pt}\Big[v\!\left(\bar{T}\!-\!t,\tilde{\mathfrak{s}}\left(\Gamma_t(z(e,y)),\eta(e)\mathbbm{1}_{\{y\leq f(\tilde{\lambda}^C_{t-})\}},\rho(e)\mathbbm{1}_{\{y\leq f(\tilde{\lambda}^C_{t-})\}};\tilde{S}^C_{t-}\right)\right)\!-\!v\!\left(\bar{T}\!-\!t,\tilde{S}^C_{t-}\right)\Big]{N}(dt,de,dy)\\
    &+\hspace{-9pt}\int\limits_{E\times \{z=0\}}\hspace{-9pt}\Big[v\!\left(\bar{T}\!-\!t,\tilde{\mathfrak{s}}\left(0,\eta(e)\mathbbm{1}_{\{y\leq g(\tilde{\lambda}^C_{t-})\}},\rho(e)\mathbbm{1}_{\{y\leq f(\tilde{\lambda}^C_{t-})\}};\tilde{S}^C_{t-}\right)\!\right)\!-\!v\!\left(\bar{T}\!-\!t,\tilde{S}^C_{t-}\right)\!\Big]N(dt,de,dy)\\&
    +\left[v\left(\bar{T}-t,\tilde{\mathfrak{s}}(\Delta^{l,p}_t C,0,0;\tilde{S}^C_{t-})\right)-v\left(\bar{T}-t,\tilde{S}^C_{t-}\right)\right]\\&
    +\left[v\left(\bar{T}-t,\tilde{\mathfrak{s}}(\Delta^r_t C,0,0;\tilde{S}^C_{t})\right)-v\left(\bar{T}-t,\tilde{S}^C_{t}\right)\right],
\end{aligned}
\end{equation}
where we used that the predictable left-jumps $\Delta^{l,p}_t C$ cannot coincide with mark times; cf.~\eqref{eq:marksnotcharged}.

For $V^C$ to be a supermartingale for \emph{any} admissible $C$, we thus see from the last two lines that 
\begin{align}\label{eq:impulse_part}
    \sup\limits_{\Delta\in  D(\lambda,q)}\Big\lbrace v(T,\tilde{\mathfrak{s}}(\Delta,0,0;s))-v(T,s)\Big\rbrace\leq 0,\quad  (T,s)\in[0,\infty)\times[\underline{\lambda},\infty)\times[-\bar{q},\bar{q}]\times\mathbb{R}^2,
\end{align}
must be satisfied. Note that in case $ D(\lambda,q)=D=\mathbb{R}$ (when continuous trading is admissible) the above condition still ensures that contributions from the $dC^c_t$ terms in $dV^C_t$ lead to supermartingale dynamics. Indeed, upon formal differentiation of $v(T,\tilde{\mathfrak{s}}(\Delta,0,0;s))-v(T,s)$ from above and from below at $\Delta=0$, the condition entails in particular $\frac{\partial v}{\partial x}\zeta-\frac{\partial v}{\partial \lambda} \pm \left(\frac{\partial v}{\partial p} \iota(\lambda)-\frac{\partial v}{\partial x}p+\frac{\partial v}{\partial q}\right)\leq 0$. As a result, the netted integrands of $dC^{c,\pm}$ in~\eqref{eq:Itoformula_V} will be nonpositive as required.   

Next, we compensate the $N$-integrals which then give local martingales and drifts accumulating with rates
\begin{align}
    &\hspace{-45pt}\int\limits_{\hspace{30pt}(E\times \R_+)\cap \{z(e,y)\neq 0\}}\hspace{-45pt}\Big[v\!\left(\bar{T}\!-\!t,\tilde{\mathfrak{s}}\left(\Gamma_t(z(e,y)),\eta(e)\mathbbm{1}_{\{y\leq g(\tilde{\lambda}^C_{t-})\}},\rho(e)\mathbbm{1}_{\{y\leq f(\tilde{\lambda}^C_{t-})\}};\tilde{S}^C_{t-}\right)\right)\!-\!v\!\left(\bar{T}\!-\!t,\tilde{S}^C_{t-}\right)\Big]\nu(de)\otimes dy \label{eq:integral_zneq0}\\&+\hspace{-30pt}\int\limits_{\hspace{30pt}(E\times \R_+)\cap \{z(e,y)=0\}}\hspace{-45pt}\left[v\!\left(\bar{T}-t,\tilde{\mathfrak{s}}\left(0,\eta(e)\mathbbm{1}_{\{y\leq g(\tilde{\lambda}^C_{t-})\}},\rho(e)\mathbbm{1}_{\{y\leq f(\tilde{\lambda}^C_{t-})\}};\tilde{S}^C_{t-}\right)\right)\!-\!v\!\left(\bar{T}-t,\tilde{S}^C_{t-}\right)\right]\nu(de)\otimes dy \label{eq:integral_z=0}.
\end{align} 
We use the disintegration 
\begin{align}\label{eq:disintegration}
\nu(de)\otimes dy=\int_{z(E\times\R_+)\setminus\{0\}}K(\bar{z};de,  dy)\mu(d\bar{z}),
\end{align}
with $\mu=(\nu\otimes \mathrm{Leb}) \circ (z)^{-1}$ to rewrite and estimate the first integral in \eqref{eq:integral_zneq0} as follows:
\begin{align*}
&\hspace{-39pt}\int\limits_{\hspace{44pt}(E\times \R_+)\cap \{z(e,y)\neq 0\}}\hspace{-51pt}\left[v\!\left(\bar{T}\!-\!t,\tilde{\mathfrak{s}}\left(\Gamma_t(z(e,y)),\eta(e)\mathbbm{1}_{\{y\leq g(\tilde{\lambda}^C_{t-})\}},\rho(e)\mathbbm{1}_{\{y\leq f(\tilde{\lambda}^C_{t-})\}};\tilde{S}^C_{t-}\right)\right)\!-\!v\!\left(\bar{T}\!-\!t,\tilde{S}^C_{t-}\right)\right]\nu(de)\otimes dy\\
    &=\hspace{-10pt}\int\limits_{z(E\times \R_+)\setminus\{0\}}\hspace{-23pt}\int\limits_{\hspace{30pt}(E\times \R_+)\cap \{z(e,y)=\bar{z}\}}\hspace{-30pt}\Big[v\left(\bar{T}-t,\tilde{\mathfrak{s}}\left(\Gamma_t(\bar{z}),\eta(e)\mathbbm{1}_{\{y\leq g(\tilde{\lambda}^C_{t-})\}},\rho(e)\mathbbm{1}_{\{y\leq f(\tilde{\lambda}^C_{t-})\}};\tilde{S}^C_{t-}\right)\right)\\[-10pt]
    &\hspace{5.5cm}-v\left(\bar{T}-t,\tilde{S}^C_{t-}\right)\Big]K(\bar{z};de,  dy) \mu(d\bar{z})\\
    &\leq \hspace{-10pt}\int\limits_{z(E\times \R_+)\setminus\{0\}}\sup_{\gamma\in D(\tilde{\lambda}^C_{t-},\tilde{Q}^C_{t-})}\int\limits_{(E\times \R_+)\cap\{z(e,y)=\bar{z}\}}\hspace{-20pt}\Big[v\left(\bar{T}-t,\tilde{\mathfrak{s}}\left(\gamma,\eta(e)\mathbbm{1}_{\{y\leq g(\tilde{\lambda}^C_{t-})\}},\rho(e)\mathbbm{1}_{\{y\leq f(\tilde{\lambda}^C_{t-})\}};\tilde{S}^C_{t-}\right)\right)\\[-10pt]
    &\hspace{7cm}-v\left(\bar{T}-t,\tilde{S}^C_{t-}\right)\Big]K(\bar{z};de,  dy) \mu(d\bar{z}).
\end{align*}
Finally, we collect all the $dt$-terms from the dynamics of $dV^C_t$ and have the tight upper bound
\begin{align}\label{eq:collected_dtterms_HJB}
    &-\frac{\partial v}{\partial T}(\bar{T}-t,\tilde{S}^C_{t-})\\
    &+\!\!\!\!\int\limits_{(E\times \R_+)\cap \{z(e,y)=0\}}\hspace{-30pt}\left[v\left(\bar{T}-t,\tilde{\mathfrak{s}}\left(0,\eta(e)\mathbbm{1}_{\{y\leq g(\tilde{\lambda}^C_{t-})\}},\rho(e)\mathbbm{1}_{\{y\leq f(\tilde{\lambda}^C_{t-})\}};\tilde{S}^C_{t-}\right)\right)-v\left(\bar{T}-t,\tilde{S}^C_{t-}\right)\right](\nu(de)\otimes dy) \nonumber\\
    &+\!\!\!\!\int\limits_{z(E\times \R_+)\setminus\{0\}}\hspace{-5pt}\sup_{\gamma\in D(\tilde{\lambda}^C_{t-},\tilde{Q}^C_{t-})}\int\limits_{(E\times \R_+)\cap\{z(e,y)=\bar{z}\}}\hspace{-15pt}\Big[v\left(\bar{T}-t,\tilde{\mathfrak{s}}\left(\gamma,\eta(e)\mathbbm{1}_{\{y\leq g(\tilde{\lambda}^C_{t-})\}},\rho(e)\mathbbm{1}_{\{y\leq f(\tilde{\lambda}^C_{t-})\}};\tilde{S}^C_{t-}\right)\right)\\[-10pt]
    &\hspace{6cm}-v\left(\bar{T}-t,\tilde{S}^C_{t-}\right)\Big]K(\bar{z};de,  dy) \mu(d\bar{z}) \nonumber.
\end{align}
To satisfy the martingale optimality principle, the terms in \eqref{eq:collected_dtterms_HJB} must be smaller than or equal to zero for any admissible strategy $C\in\tilde{\mathcal{C}}(\lambda,q)$ and zero for some (optimal) strategy $C^*$. With \eqref{eq:impulse_part}, we thus obtain the HJBQVI
\begin{align}
&\max\bigg\lbrace-\frac{\partial v}{\partial T}(T,s)+\hspace{-15pt}\int\limits_{(E\times \R_+)\cap \{z(e,y)=0\}}\hspace{-22pt}\left[v\left(T,\tilde{\mathfrak{s}}\left(0,\eta(e)\mathbbm{1}_{\{y\leq g(\lambda)\}},\rho(e)\mathbbm{1}_{\{y\leq f(\lambda)\}};s\right)\right)-v\left(T,s\right)\right]\nu(de)\otimes dy\nonumber\\
    &\quad+\hspace{-18pt}\int\limits_{z(E\times \R_+)\setminus\{0\}}\hspace{-8pt}\sup_{\gamma\in  D(\lambda,q)}\hspace{-35pt}\int\limits_{\hspace{30pt}(E\times \R_+)\cap\{z(e,y)=\bar{z}\}}\hspace{-48pt}\left[v\left(T,\tilde{\mathfrak{s}}\left(\gamma,\eta(e)\mathbbm{1}_{\{y\leq g(\lambda)\}},\rho(e)\mathbbm{1}_{\{y\leq f(\lambda)\}};s\right)\right)\!-\!v\left(T,s\right)\right]K(\bar{z};de,  dy) \mu(d\bar{z});\nonumber\\
&\quad\sup\limits_{\Delta\in  D(\lambda,q)}\Big\lbrace v(T,\tilde{\mathfrak{s}}(\Delta,0,0;s))-v(T,s)\Big\rbrace\bigg\rbrace=0,\label{eq:HJB}
\end{align}
for $(T,s)=(T,\lambda,q,p,x)\in[0,\infty)\times [\underline{\lambda},\infty)\times[-\bar{q},\bar{q}]\times\R\times\R$. Due to trading being halted when liquidity falls below $\underline{\lambda}$, the value function evaluated in $\lambda<\underline{\lambda}$ is
\begin{align*}
   v(T,\lambda,q,p,x)=
       U\left(\w(x,q,p,\underline{\lambda})\right)\exp\left(\mathcal{L}_Y(\alpha q)\right),
\end{align*}
$\text{for }(T,q,p,x)\in[0,\infty)\times[-\bar{q},\bar{q}]\times\R\times\R$, and the initial condition for $s=(\lambda,q,p,x)\in [\underline{\lambda},\infty)\times [-\bar{q},\bar{q}]\times\mathbb{R}^2$ with zero time to go $T=0$ is
\begin{align}
   v(0,\lambda,q,p,x)=\begin{cases}
       U\left(\w(x,q,p,\lambda)\right)\exp\left(\mathcal{L}_Y(\alpha\rr(q,\lambda))\right),&\text{ if }\lambda\geq \underline{\lambda},\\
       U\left(\w(x,q,p,\underline{\lambda})\right)\exp\left(\mathcal{L}_Y(\alpha q)\right),&\text{ if }\lambda<\underline{\lambda}.
   \end{cases}
\end{align}

To reduce the dimensions of the HJB equation, we recall the multiplicative structure~\eqref{eq:multiplicativeValueFunction} of the value function.
For notational simplicity, we introduce 
\begin{align}
\gamma^\lambda&:=\Upsilon(\gamma,\lambda),\\
    \eta^{\lambda,\gamma}(e,y)&:=\Upsilon(\eta(e),\lambda-|\gamma^\lambda|)\mathbbm{1}_{\{y \leq g(\lambda)\}},\\
    \rho^{\lambda,\gamma}(e,y)&:=\Big(\rho^+(e)\mathbbm{1}_{\{\lambda-|\gamma^\lambda|\geq\underline{\lambda}\}}+\Upsilon(\rho^-(e),\lambda-|\gamma^\lambda|)\Big)\mathbbm{1}_{\{y \leq f(\lambda)\}},\\
       \lambda^{\lambda,\gamma}(e,y)&:=
    \hfill\lambda-|\gamma|-|\eta(e)|\mathbbm{1}_{\{y \leq g(\lambda)\}}+\rho(e)\mathbbm{1}_{\{y \leq f(\lambda)\}},\\
    \Delta^\lambda&:=\Upsilon(\Delta,\lambda),\\
    \w_0(q,\lambda)&:=\w(0,q,0,\lambda)=-\zeta|q|-\Xi(q,\lambda),
   \end{align}
and write the reduced HJB for $v_0(T,\lambda,q)$ with $(T,\lambda,q)\in[0,\infty)\times [\underline{\lambda},\infty)\times [-\bar{q},\bar{q}]$ when $\alpha>0$ as
\begin{equation}\label{eq:HJB_reduced}
\begin{aligned}
\max&\bigg\lbrace -\frac{\partial v_0}{\partial T}(T,\lambda,q)+\hspace{-20pt}\int\limits_{(E\times \R_+)\cap \{z(e,y)=0\}}\hspace{-30pt}\bigg[v_0\Big(T,\lambda^{\lambda,0}(e,y),q\Big)\Big|U\Big(I(\eta^{\lambda,0}(e,y),\lambda)q\Big)\Big|-v_0(T,\lambda,q)\bigg]\nu(de)\otimes dy\\&+\int\limits_{z(E\times \R_+)\setminus\{0\}}\hspace{-5pt}\sup_{\gamma\in  D(\lambda,q)}\bigg\lbrace\hspace{-50pt}\int\limits_{\hspace{50pt}(E\times \R_+)\cap\{z(e,y)=\bar{z}\}}\hspace{-40pt}\bigg[v_0\Big(T,\lambda^{\lambda,\gamma}(e,y),q+\gamma^\lambda\Big)\cdot\\&\hspace{3cm}\Big|U\Big(\w_0(\gamma^\lambda,\lambda)+\Big(I\left(\gamma^\lambda,\lambda\right)+I\big(\eta^{\lambda,\gamma}(e,y),\lambda-|\gamma^\lambda|\big)\Big)(q+\gamma^\lambda)\Big)\Big|\\[-5pt]&\hspace{4cm}-v_0(T,\lambda,q)\bigg] K(\bar{z};de,  dy)\bigg\rbrace\mu(d\bar{z});\\
&\sup\limits_{\Delta\in  D(\lambda,q)}\Big\lbrace  v_0(T,\lambda-|\Delta|,q+\Delta^\lambda)\Big|U\Big(\w_0(\Delta^\lambda,\lambda)+I\left(\Delta^\lambda,\lambda\right)(q+\Delta^\lambda)\Big)\Big|\\&\qquad\qquad-v_0(T,\lambda,q)\Big\rbrace\bigg\rbrace=0.
\end{aligned}
\end{equation}
For $\lambda<\underline{\lambda}$, the function $v_0(T,\lambda,q)$ takes the value
\begin{align}
   v_0(T,\lambda,q)=
       U\left(\w_0(q,\lambda)\right)\exp\left(\mathcal{L}_Y(\alpha q)\right),\quad (T,q)\in[0,\infty)\times [-\bar{q},\bar{q}],
\end{align}
and the initial condition for $(\lambda,q)\in \R\times [-\bar{q},\bar{q}]$ with zero time to go is
\begin{align}
   v_0(0,\lambda,q)=\begin{cases}
       U\left(\w_0(q,\lambda)\right)\exp\left(\mathcal{L}_Y(\alpha \rr(q,\lambda))^+)\right),&\text{ if }\lambda\geq \underline{\lambda},\\
       U\left(\w_0(q,\underline{\lambda})\right)\exp\left(\mathcal{L}_Y(\alpha q)\right),&\text{ if }\lambda<\underline{\lambda}.
   \end{cases}
\end{align}
Similarly, one obtains the reduced HJB for $w$ when $\alpha=0$.\\
The suggested optimal strategy can as usual be constructed in feedback form by determining the maximizer in the suprema of the HJB \eqref{eq:HJB_reduced}.

\subsection{The Numerical Scheme}\label{chapter:numericalscheme}

In the following, we present a numerical finite-difference scheme to approximate the solution of the dimension reduced HJB equation in \eqref{eq:HJB_reduced} for $\alpha>0$; a similar scheme can be derived for $\alpha=0$. 

For this, we first discretise the value function's domain $[0,T]\times[\underline{\lambda},\infty) \times [-\bar{q},\bar{q}]$. For the time dimension $T'\in[0,T]$, we take $\mathbb{T}_{\delta T}$ as the grid on $[0,T]$ with step size $\delta T>0$, where $T=N_1\delta T$ for some $N_1\in\mathbb{N}$.
For the liquidity dimension $\lambda\in[\underline{\lambda},\infty)$, we choose the grid $\mathbb{R}_{\delta \lambda}^{\overline{\lambda}}$ discretising $[\underline{\lambda}-\delta \lambda,\overline{\lambda}]$ with stepsize $\delta \lambda>0$ where $\overline{\lambda}=\underline{\lambda}+N_2\delta \lambda$ for some $N_2\in\mathbb{N}$ is a conveniently chosen maximum liquidity level. To remain within a bounded domain, we assume $\lambda$ equals $\underline{\lambda}-\delta \lambda$ after trading was halted and define the numerical scheme below accordingly. For the inventory levels $q\in\mathbb{R}$, we introduce the bounded grid $\mathbb{R}_{\delta q}^{\underline{Q},\overline{Q}}$ on $[\underline{Q},\overline{Q}]$ with step size $\delta q>0$ where $\underline{Q}=N_3\delta q$ and $\overline{Q}=N_4\delta q$ for some constants $N_3,N_4\in\mathbb{Z}$, $N_3<N_4$ and $[\underline{Q},\overline{Q}]\subseteq [-\overline{q},\overline{q}]$. 
For the inventory $q$ to remain within the grid $\mathbb{R}_{\delta q}^{\underline{Q},\overline{Q}}$ and for $\lambda$ to remain above $\underline{\lambda}-\delta \lambda$, we define the set of actions$$D^{\underline{Q},\overline{Q}}(q,\lambda):=\Big\{n\cdot\delta q\text{ with }n\in\mathbb{Z} :q+n\,\delta q\in\mathbb{R}_{\delta q}^{\underline{Q},\overline{Q}}\text{ and }\lambda-|n\delta|\geq\underline{\lambda}-\delta \lambda\Big\}, \qquad q\in\mathbb{R}_{\delta q}^{\underline{Q},\overline{Q}},\lambda\in\mathbb{R}_{\delta \lambda}^{\overline{\lambda}}.$$

To propagate the value function, it is convenient to introduce the following operator acting on functions $h^{T'}: \mathbb{R}_{\delta \lambda}^{\overline{\lambda}}\times \mathbb{R}_{\delta q}^{\underline{Q},\overline{Q}}\rightarrow\mathbb{R}$ with $T'\in\mathbb{T}_{\delta T}$:
\begin{align*}
&\mathcal{L}^{\delta T, \delta \lambda}(q,\lambda,h^{T'}):=h^{T'}(\lambda,q)+\delta T\bigg(\Delta_{z=0}h^{T'}(\lambda,q)+\hspace{-10pt}\int\limits_{z(E\times \mathbb{R}_+)\setminus \{0\}}\sup\limits_{\gamma\in D^{\underline{Q},\overline{Q}}(q,\lambda)}\hspace{-5pt}\Delta_{\text{z=}\bar{z}}h^{T'}(\lambda,q,\gamma,\psi)\;\;\mu(d\bar{z})\bigg),
\end{align*}
where  
\begin{align*}
&\Delta_{\text{z=0}}h^{T'}(\lambda,q):=\hspace{-20pt}\int\limits_{(E\times \mathbb{R}_+)\cap\{z(e,y)=0\}}\hspace{-26pt}\Big[\bar{h}^{T'}\big(\lambda^{\lambda,0}(e,y)\vee (\underline{\lambda}-\delta \lambda),q\big)\big|U\big(I(\eta^{\lambda,0}(e,y),\lambda)q\big)\big|-h^{T'}(\lambda,q)\Big]\nu(de)\otimes dy
\end{align*}
describes the propagation while there is no signal to act upon, and
\begin{align*}
 & \Delta_{\text{z=}\bar{z}}h^{T'}(\lambda,q,\gamma):=\\
 &\hspace{-40pt}
\int\limits_{\hspace{30pt}(E\times \mathbb{R}_+)\cap\{z(e,y)=\bar{z}\}}\hspace{-25pt}\bigg[\bar{h}^{T'}\big(\lambda^{\lambda,\gamma}(e,y)\vee (\underline{\lambda}-\delta \lambda),q+\gamma^\lambda\big)\cdot\;\\[-5pt]&\hspace{0.6cm}\Big|U\Big(\w_0(\gamma^\lambda,\lambda)-\big(I(\gamma^\lambda,\lambda)+I\big(\eta^{\lambda,\gamma}(e,y),\lambda-|\gamma^\lambda|\big)\big)(q+\gamma^\lambda)\Big)\Big|-h^{T'}(\lambda,q)\bigg] K(\bar{z};de,dy)
\end{align*}
to describe the propagation when an order $\gamma$ is placed upon observing a signal $\bar{z}$. For  values $\lambda\notin \mathbb{R}_{\delta \lambda}^{\overline{\lambda}}$ we linearly interpolate to obtain
\begin{align*}
\bar{h}^{T'}(q,\lambda):=&h^{T'}\left(q,\left\lceil\frac{\lambda}{\delta \lambda }\right\rceil\delta \lambda\right)-\left(\left\lceil\frac{\lambda}{\delta \lambda }\right\rceil\delta \lambda-\lambda\right)\left(h^{T'}\left(q,\left\lceil\frac{\lambda}{\delta \lambda }\right\rceil\delta \lambda\right)-h^{T'}\left(q,\left\lfloor\frac{\lambda}{\delta \lambda }\right\rfloor\delta \lambda\right)\right).
\end{align*}
Here, $\lfloor\cdot\rfloor$ and $\lceil\cdot\rceil$ denotes rounding to the closest integer below and above.
Similarly, we define 
\begin{align*}
&\mathcal{M}^{\delta T, \delta \lambda}(\lambda,q,h^{T'})
:=\hspace{-15pt}\sup\limits_{\Delta\in D^{\underline{Q},\overline{Q}}(q,\lambda)} \!\bigg\lbrace \bar{h}^{T'}\big((\lambda-|\Delta|)\vee(\lambda-\delta\lambda),q+\Delta\big)\cdot\\[-10pt]&\hspace{5cm}\big|U\big(\w_0(\Delta^\lambda,\lambda)-I(\Delta^\lambda,\lambda)(q+\Delta^\lambda)\big)\big|\bigg\rbrace
\end{align*}
as an operator that ensures state based trades will be executed optimally.

With this notation in place, we specify the following numerical scheme to compute an approximate solution $w$ to \eqref{eq:HJB_reduced}. First, initialize with the boundary values with time $T=0$ to go:\begin{align*}
v_0^0(q,\lambda)&:=\begin{cases}
    U\left(\w_0(q,\lambda)\right)\exp\left(\mathcal{L}_Y(\alpha \rr(q,\lambda))\right),&\text{ if }\lambda\geq \underline{\lambda}, \\\hfill U\left(\w_0(q,\underline{\lambda})\right)\exp\left(\mathcal{L}_Y(\alpha q)\right),&\text{ if }\lambda=\underline{\lambda}-\delta \lambda.
\end{cases}
\end{align*}
Next, after computing the value function up to a time horizon $T'$, we propagate it a step further to $T'+\delta T$ by calculating the value when one can play for an extra period to obtain
\begin{align*}    
\tilde{v_0}^{T'+\delta T}(q,\lambda)&:=\mathcal{L}^{\delta T, \delta \lambda}(q,\lambda,v_0^{T'}),\qquad\lambda\geq \underline{\lambda}.
\end{align*}
This value is then updated to account for possible state-based trades to obtain the next instance
\begin{align*}
v_0^{T'+\delta T}(q,\lambda)&:=\max(\tilde{v_0}^{T'+\delta T}(q,\lambda),\mathcal{M}^{\delta t, \delta \lambda}(q,\lambda,\tilde{v_0}^{T'+\delta T})),\qquad\lambda\geq \underline{\lambda},
\end{align*}
while maintaining the boundary condition
\begin{align*}
v_0^{T'+\delta T}(q,\underline{\lambda}-\delta\lambda)&:=
     U\left(\w_0(q,\underline{\lambda})\right)\exp\left(\mathcal{L}_Y(\alpha q)\right)\quad \text{for}\; T'\in \{0,\delta t,...,T-2\delta t, T-\delta t\}
\end{align*}
 when trading halts.
\appendix 

\section{Proofs for Section~\ref{chapter:model}}\label{chapter:proofs}

\paragraph{Model dynamics admit a unique solution.}\label{sec:uniquesol_model}
\begin{lemma}\label{lemma:uniquesol_model}
 Under condition~\eqref{asp:fandgLipschitzMonotone}, the system dynamics for $(\lambda,L,M)$ given by~\eqref{eq:dynamics_liquidity}--\eqref{eq:dynamics_MO} allow for a unique solution.
\end{lemma}
\begin{proof}
It suffices to rule out that the mutually-exciting dynamics lead to blow-ups. 
To this end, we prove that the expectation of the total variation $V_{[0,t]}(\lambda)$ of the liquidity process $(\lambda_s)$ over $[0,t]$ is bounded for each $t\in[0,T]$. For this we use that $f$ and $g$ have at most linear growth to estimate
\begin{align*}
    &\mathbb{E}[V_{[0,t]}(\lambda)]\\&= \mathbb{E}\left[\int_{ [0,t]\times E \times \mathbb{R}_+}\!\!\mathbbm{1}_{\{y\leq f(\lambda_{s-})\}}|\rho(e)|N(ds,de,dy)+\int_{ [0,t]\times E \times \mathbb{R}_+}\!\!\mathbbm{1}_{\{y\leq g(\lambda_{s-})\}}|\eta(e)|N(ds,de,dy)\right]\\
    &=\int_E|\rho(e)|\nu(de)\mathbb{E}\left[\int_0^tf(\lambda_s)ds\right]+\int_E|\eta(e)|\nu(de)\mathbb{E}\left[\int_0^t g(\lambda_s)ds\right]\\
    &\leq \left(\int_E|\rho(e)|\nu(de)+\int_E|\eta(e)|\nu(de)\right)cT\left(1+|\lambda_{0-}|+\mathbb{E}\left[\int_0^t V_{[0,s]}(\lambda)ds\right]\right).
\end{align*}
The claim follows by Gronwall's inequality.
\end{proof}
\paragraph{Link between price volatility and liquidity of the market --- Proof of Lemma \ref{lemma:link_volaliquidity}.}
The quadratic variation of the price process has dynamics
\begin{align}
d[P]_t=&\int_E\mathbbm{1}_{\{y\leq g(\lambda_{t-})\}}I(\eta(e),\lambda_{t-})^2N(dt,de,dy),\quad [P]_0=0.
\end{align}
Hence, the dynamics of its predictable compensator coincide with \eqref{eq:pred_quadvar_price}.
The derivative of the expression in \eqref{eq:function_price_vola} with respect to $\lambda$ is negative if condition \eqref{eq:elasticitycond} holds; this proves the monotonicity claim.

\paragraph{Decomposition of $C$ --- Proof of Lemma \ref{Lemma:decomposition_q}.}\label{sec:DecompositionStrategies}
It is easy to see that $C$ of the form \eqref{eq:representation_Q} is $\Lambda$-measurable and has bounded total variation. 

To prove the converse, we have to show that the left-jumps of $C$ are of the form~\eqref{eq:representation_deltaleft} with $\Gamma$ and $C^{l,p}$ as described in the lemma. Without loss of generality, we can assume that they are all negative and that their sum $C^l$ is bounded. Then $C^l$ is a $\Lambda$-supermartingale and, thus, can be decomposed into $C^l=M-A-B_-$ for some local $\Lambda$-martingale, some predictable, nondecreasing, rightcontinuous $A$ with $A_0=0$ and some $\Lambda$-measurable nondecreasing, right-continuous pure-jump process $B$ with $B_0=0$; cf.~\cite{Lenglart:80}, Theorem 4, p.527. By localization, $M$ is even a true martingale and, by the martingale representation theorem for Poisson point processes, there is a $\mathcal{P}\otimes \mathcal{E} \otimes \mathcal{B}([0,T])$-measurable $\xi \in L^1(\PP \otimes\, ds \otimes \nu \otimes dy)$ such that
\begin{align}
    M_T = \int_{[0,T] \times E \times [0,\infty)} \xi_s(e,y) \bar{N}(ds,de,dy),
\end{align}
where we adopt the notation from Section~\ref{section:HJB}.
Taking the $\Lambda$-projection gives us
\begin{align}
    M_t =(^{\Lambda}M)_t=\int_{[0,t) \times E \times [0,\infty)} \xi_s(e,y) \bar{N}(ds,de,dy)+\hat{\xi}_t(Z_t), \quad t \in [0,T],
\end{align}
where $\hat{\xi}_t(0):=0$ and $\hat{\xi}_t(\bar{z}):=\int_{\{z=\bar{z}\}} \xi_t(e,y) K(\bar{z};de, dy)$ for $\bar{z} \in z(E\times[0,\infty))\setminus\{0\}$, with $K$ as in~\eqref{eq:disintegration}. Since $A$ is predictable and $N$ places marks only at $\mathcal{P}$-totally inaccessible stopping times, the jumps $\Delta A=:-\Delta C^{l,p}$ almost surely do not coincide with mark times: $\{t \in [0,T]\,:\, \Delta_t A\not=0\} \subset \{t \in [0,T]\,:\, N(\{t\} \times E \times [0,\infty))=0\}$. It follows that $\Delta^l_t C=\Delta^l_t M-\Delta^l_t A-\Delta^l_t B_-=\hat{\xi}_t(Z_t)+\Delta_t C^{l,p}+0$. This is~\eqref{eq:representation_deltaleft} with $\Gamma :=\hat{\xi}$.

\paragraph{Dynamics of realizable portfolio value --- Proof of Proposition~\ref{prop:wealthDynamics}\label{proof:propWealthDynamics}.}
We use that Marcus-style dynamics~\eqref{eq:MarcusDynamicsLiquidityControlled}, \eqref{eq:MarcusDynamicsCashControlled}, \eqref{eq:MarcusDynamicsPositionControlled}, \eqref{eq:MarcusDynamicsPriceControlled} of $\lambda^C$, $X^C$, $Q^C$, and $P^C$ are governed by classical calculus to compute the dynamics of ${W}^C=X^C+Q^CP^C-\left(\zeta|Q^C|+\Xi(Q^C,\lambda^C)\right)$ as
\begin{align}
    {W}^C_t-{W}^C_{0-} 
    &=\;  {X}^C_t+{P}^C_t{Q}^C_t-(x+pq)-\zeta\left(|{Q}^C_t|-|q|\right)-\left(\Xi({Q}^C_t,{\lambda}^C_t)-\Xi(q,\lambda)\right)\\
    &= \int_0^t {Q}^C \diamond d{P}^C-\zeta \int_0^t \diamond\; |d{Q}^C|-\zeta \int_0^t \sgn({Q}^C) \diamond d{Q}^C\\
    & \qquad - \left(\int_0^t \partial_\Delta \Xi({Q}^C,{\lambda}^C) \diamond d{Q}^C+\int_0^t \partial_\lambda \Xi(-{Q}^C,{\lambda}^C)\diamond d{\lambda}^C\right).
\end{align}
With
\begin{align}
    \partial_\Delta \Xi(\Delta,\lambda) &= \sgn(\Delta) I(|\Delta|,\lambda), \\
    \partial_\lambda \Xi(\Delta,\lambda) &= \int_0^{|\Delta|}\partial_\lambda I(z,\lambda)dz = \int_0^{|\Delta|}\int_0^{z} \iota'(\lambda-{\tilde{z}})d{\tilde{z}} dz= \int_0^{|\Delta|} (\iota(\lambda)-\iota(\lambda-z))dz\\ &= |\Delta|\iota(\lambda)-I(|\Delta|,\lambda)
\end{align}
we find
\begin{align}
    {W}^C_t-{W}^C_{0-} 
    =\; & \int_0^t {Q}^C \iota({\lambda}^C) \diamond \left(d{Q}^C+(0\vec{+}d{M}^C\vec{+}0)\right)\\
    &-\zeta \int_0^t \left(\diamond \;|d{Q}^C|+ \sgn({Q}^C) \diamond d{Q}^C\right)\\
    & - \left(\int_0^t \sgn({Q}^C)I(|{Q}^C|,\lambda^C) \diamond d{Q}^C+\int_0^t \left(|{Q}^C|\iota({\lambda}^C)-I(|{Q}^C|,{\lambda}^C)\right)\diamond d{\lambda}^C
    \right)\\
    =\;& \int_0^t \left(|{Q}^C| \iota({\lambda}^C)-I(|{Q}^C|,{\lambda}^C)-\zeta\right)\left(\sgn({Q}^C) \diamond d{Q}^C+\diamond \;|d{Q}^C|\right)\\
    & + \int_0^t |{Q}^C|\iota({\lambda}^C)\sgn({Q}^C) \diamond \left(0\vec{+}d{M}^C\vec{+}0\right)\\&
    -\int_0^t \left(|{Q}^C|\iota({\lambda}^C)-I(|{Q}^C|,{\lambda}^C)\right)\diamond\left(0\vec{+}(dL^C-|d{M}^C|)\vec{+}0\right).
\end{align}
Decomposing 
\begin{align}
    d{Q}^C &= d{Q}^{C,+}-d{Q}^{C,-}=(dC^{l,+}-dC^{l,-})\vec{+}0\vec{+}(dC^{r,c,+}-dC^{r,c,-}),\\
    |d{Q}^C|&=d{Q}^{C,+}+d{Q}^{C,-}= (dC^{l,+}+dC^{l,-})\vec{+}0\vec{+}(dC^{r,c,+}+dC^{r,c,-}),
\end{align}
and decomposing similarly the $d{M}^C$-terms leads to 
\begin{align}
      {W}^C_t-{W}^C_{0-} 
    =\; & \int_0^t \left(|{Q}^C| \iota({\lambda}^C)-I(|{Q}^C|,{\lambda}^C)-\zeta\right)2\left(\mathbbm{1}_{\{{Q}^C>0\}} \diamond d{Q}^{C,+}+\mathbbm{1}_{\{{Q}^C<0\}}d{Q}^{C,-}) \right)\\
    & + \int_0^t \left(2\mathbbm{1}_{\{{Q}^C>0\}}|{Q}^C|\iota({\lambda}^C)-I(|{Q}^C|,{\lambda}^C)\right) \diamond \left(0\vec{+}d{M}^{C,+}\vec{+}0\right)\\&+\int_0^t\left(2\mathbbm{1}_{\{{Q}^C<0\}}|{Q}^C|\iota({\lambda}^C)-I(|{Q}^C|,{\lambda}^C)\right) \diamond \left(0\vec{+}d{M}^{C,-}\vec{+}0\right)\\
    &+
    \int_0^t \left(I(|{Q}^C|,{\lambda}^C)-|{Q}^C|\iota({\lambda}^C)\right) \diamond \left(0\vec{+}d{L}^C\vec{+}0\right).
\end{align}
Rearranging terms gives~\eqref{eq:WealthDynamics}.

\section{Regularity of the value function}\label{sec:regularity}

This section gives the proof of our regularity Theorem~\ref{thm:continuity_all} for the value function.

\paragraph{Total order activity bounded by provided liquidity.}

We can use that market orders and cancellations are executed only as long as liquidity remains greater than $\underline{\lambda}$ to estimate 
\begin{align}\label{eq:proofMQ_lowerbound}
\underline{\lambda}\leq \tilde{\lambda}^C_{t}=\lambda-V_{[0,t]}(\tilde{Q}^C)-V_{[0,t]}(\tilde{M}^C)-\tilde{L}^{C,-}_{t}+\tilde{L}^{C,+}_{t}, \quad t \in [0,T],
\end{align}
with $\tilde{M}^C$, $\tilde{L}^{C,+}$, $\tilde{L}^{C,-}$ as in \eqref{eq:MO_with_circuitbreaker}, \eqref{eq:LO_with_circuitbreaker} and \eqref{eq:cancel_with_circuitbreaker}.
By rearranging~\eqref{eq:proofMQ_lowerbound}, we find 
\begin{align*}
V_{[0, t]}(\tilde{Q}^C)+V_{[0, t]}(\tilde{M}^C)+\tilde{L}^{C,-}_t
\leq \lambda-\underline{\lambda}+\tilde{L}^{C,+}_t\quad t \in [0,T].
\end{align*}
Since the overall liquidity supplied in a market with minimum liquidity $\underline{\lambda}$ can be bounded by
\begin{align}\label{eq:defLiquidityBound}
\tilde{L}^{C,+}_t \leq \int_{[0,t]\times E \times [0,f(\underline{\lambda})]} \rho^+(e) N(ds,de,dy) =: \bar{L}^+_t,   
\end{align}
we find the total amount of market orders and cancellations controlled by
\begin{align}\label{eq:liquidityBound}
    V_{[0, t]}(\tilde{Q}^C)+V_{[0, t]}(\tilde{M}^C)+\tilde{L}^{C,-}_t \leq \lambda-\underline{\lambda}+\bar{L}^+_t, \quad t \in [0,T] .
\end{align}
\paragraph{Multiplicative structure and nondegeneracy of the value function.}
By Proposition~\ref{prop:wealthDynamics}, 
    \begin{align}
       \tilde{W}^C_{T+}  = \; & \w(x,q,p,\lambda)\\
    & + \int_0^{T+} \left(I(|\tilde{Q}^C|,\tilde{\lambda}^C)-|\tilde{Q}^C|\iota(\tilde{\lambda}^C)\right) \diamond \left(0\vec{+}d\tilde{L}^C\vec{+}0\right)\\
     & - \int_0^{T+} \left(I(|\tilde{Q}^C|,\tilde{\lambda}^C)-2\mathbbm{1}_{\{\tilde{Q}^C>0\}}|\tilde{Q}^C|\iota(\tilde{\lambda}^C)\right) \diamond \left(0\vec{+}d\tilde{M}^{C,+}\vec{+}0\right)\label{eq:TerminalWealthRepresentation}\\&-\int_0^{T+}\left(I(|\tilde{Q}^C|,\tilde{\lambda}^C)-2\mathbbm{1}_{\{\tilde{Q}^C<0\}}|\tilde{Q}^C|\iota(\tilde{\lambda}^C)\right) \diamond \left(0\vec{+}d\tilde{M}^{C,-}\vec{+}0\right)\\
    & - \int_0^{T+} 2\left(I(|\tilde{Q}^C|,\tilde{\lambda}^C)-|\tilde{Q}^C| \iota(\tilde{\lambda}^C)+\zeta\right) \diamond \left(\mathbbm{1}_{\{\tilde{Q}^C>0\}}d\tilde{Q}^{C,+}+\mathbbm{1}_{\{\tilde{Q}^C<0\}} d\tilde{Q}^{C,-}\right)\\
    &+Y\rr(\tilde{Q}^C_{T+},\tilde{\lambda}^C_{T+}).
    \end{align}
We notice that neither $x$ nor $p$ plays a role in the above integrals or the $Y$-term. As $\w(x,q,p,\lambda)=x+pq+\w(0,q,0,\lambda)$ the exponential utility structure allows us to separate $x+pq$ and we conclude that
\begin{align}
    v(T,\lambda,q,p,x)  &= -\exp\left(-\alpha(x+pq)\right) v(T,\lambda,q,0,0),
\end{align}
which is the claimed multiplicative form~\eqref{eq:multiplicativeValueFunction} of the value function.

As a further consequence of the above wealth representation, we can use the position limit $\bar{q}$ of~\eqref{eq:inventorybound} and the liquidity bound $\underline{\lambda}$ to estimate
\begin{align}\label{eq:wealthBoundAbove}
    \tilde{W}^C_{T+} \leq \; & \w(x,q,p,\lambda)
        +  \bar{q}\iota(\underline{\lambda})\left(\tilde{L}^{C,+}_T+2V_{[0,T]}(\tilde{M}^C)\right)
    +Y\rr(\tilde{Q}^C_{T+},\tilde{\lambda}^C_{T+}).
\end{align}
Therefore, we find with~\eqref{eq:liquidityBound} that
\begin{align}
    \E&[U_{\alpha}(\tilde{W}^C_{T+})] \\ &\leq -\exp\left(-\alpha \w(x,q,p,\lambda)\right)
\E\left[\exp\left(-3\alpha\bar{q}\iota(\underline{\lambda})(\lambda-\underline{\lambda}+\bar{L}^+_{T})\right)\left.\E\left[\exp(a Y)|\mathcal{F}_T\right]\right|_{a=-\alpha\rr(\tilde{Q}^C_{T+},\tilde{\lambda}^C_{T+})} \right].
\end{align}
Next, due to~\eqref{eq:YProperties}, $Y$ is centered and independent of $\mathcal{F}_T$, the conditional expectation is at least~$1$ by Jensen's inequality, and so
\begin{align}
    \E[U_{\alpha}(\tilde{W}^C_{T+})] & \leq -\exp\left(-\alpha\w(x,q,p,\lambda)\right)\E\left[\exp\left(-3\alpha\bar{q}\iota(\underline{\lambda})(\lambda-\underline{\lambda}+\bar{L}^+_{T})\right) \right]<0.
\end{align}
As this bound does not depend on the choice of $C \in \tilde{\mathcal{C}}(\lambda,q)$, the value function $v$ of~\eqref{eq:value_fct} is non-degenerate.

\paragraph{Continuity of the value function with respect to the time horizon.}

Let us first show that our value functions depends continuously and monotonously on the time horizon:
\begin{lemma}\label{lemma:continuity_time}
    For any $T'\leq T$, we have
    \begin{align}\label{eq:timeRegularityEstimate}
        v(T',q,p,x) \leq v(T,\lambda,q,p,x) \leq v(T',\lambda,q,p,x) \tau(T-T')<0
    \end{align}
 for some nonincreasing continuous function $\tau:[0,\infty) \to (0,1]$ with $\tau(0)=1$.
\end{lemma}
\begin{proof}
 As the investor can liquidate and wait until the time available for investment has elapsed, it is obvious that the value function is monotone in $T$.
     
 For the second inequality, let us consider a control $C \in \tilde{C}$ for time horizon $T$ and use the properties of $Y$ from~\eqref{eq:YProperties} to write its expected utility as
 \begin{align}\label{eq:expectedUtilityTime}
     \E[U(\tilde{W}^C_{T+})] = 
     \E\left[-\exp\left(-\alpha \w\left(\tilde{X}^C_{T+},\tilde{Q}^C_{T+},\tilde{P}^C_{T+},\tilde{\lambda}^C_{T+}\right)
     +\mathcal{L}_Y\left(\alpha \rr(\tilde{Q}^C_{T+},\tilde{\lambda}^C_{T+})\right)\right)
     \right].
 \end{align}
 Analogously to~\eqref{eq:wealthBoundAbove}, we conclude that
 \begin{align}
     \w&\left(\tilde{X}^C_{T+},\tilde{Q}^C_{T+},\tilde{P}^C_{+},\tilde{\lambda}^C_{T+}\right)-\w\left(\tilde{X}^C_{T'+},\tilde{Q}^C_{T'+},\tilde{P}^C_{T'+},\tilde{\lambda}^C_{T'+}\right)\\ &\leq \bar{q}\iota(\underline{\lambda})\left(\tilde{L}^{C,+}_T-\tilde{L}^{C,+}_{T'}+2V_{(T',T]}(\tilde{M}^C)\right)\\
     & \leq \bar{q}\iota(\underline{\lambda})\int_{(T',T]\times E \times [0,\infty)]} \left(\mathbbm{1}_{[0,f(\underline{\lambda})]}(y)\rho^+(e)+2\mathbbm{1}_{[0,g(\infty)]}(y)|\eta(e)|\right)N(ds,de,dy).
 \end{align}
 Similarly,
 \begin{align}
    \mathcal{L}_Y&\left(\alpha \rr(\tilde{Q}^C_{T+},\tilde{\lambda}^C_{T+})\right)-\mathcal{L}_Y\left(\alpha \rr(\tilde{Q}^C_{T'+},\tilde{\lambda}^C_{T'+})\right)\\
     &= \int_{T'}^{T+} \mathcal{L}_Y' (\alpha \rr(\tilde{Q}^C,\tilde{\lambda}^C))\alpha  \mathbbm{1}_{\{\rr(\tilde{Q}^C,\tilde{\lambda}^C)^C>0\}}\Big(\sgn(\tilde{Q}^C) \diamond d\tilde{Q}^C\\&
     \qquad\qquad\qquad\qquad\qquad\qquad\qquad\qquad-\diamond \left(-|d\tilde{Q}^C|+\left(0\vec{+}(d\tilde{L}^{C}-|d\tilde{M}^C|)\vec{+}0\right)\right)\Big)\\&= \int_{T'}^{T+} \mathcal{L}_Y' (\alpha \rr(\tilde{Q}^C,\tilde{\lambda}^C))\alpha  \mathbbm{1}_{\{\rr(\tilde{Q}^C,\tilde{\lambda}^C)>0\}}\Big(2\mathbbm{1}_{\{\tilde{Q}^C>0\}} \diamond d\tilde{Q}^{C,+}+2\mathbbm{1}_{\{\tilde{Q}^C<0\}} \diamond d\tilde{Q}^{C,-}\\&
     \qquad\qquad\qquad\qquad\qquad\qquad\qquad\qquad\qquad\qquad -\diamond \left(0\vec{+}(d\tilde{L}^{C}-|d\tilde{M}^C|)\vec{+}0\right)\Big)\\
     & \geq -\mathcal{L}_Y' (\alpha\bar{q})\alpha (\tilde{L}^{C,+}_{T}-\tilde{L}^{C,+}_{T'}) \\
     &=  
     - \mathcal{L}_Y' (\alpha\bar{q})\alpha\int_{(T',T]\times E \times [0,\infty)]} \mathbbm{1}_{[0,f(\underline{\lambda})]}(y)\rho^+(e)N(ds,de,dy),
 \end{align}
 where for the last estimate we used that $\mathcal{L}_Y'$ is increasing from $\mathcal{L}_Y'(0)=0$ because $Y$ is centered with symmetric, convex cumulant generating function $\mathcal{L}_Y$; cf.~\eqref{eq:YProperties}.
 
 With these estimates we get from~\eqref{eq:expectedUtilityTime}:
 \begin{align}
     \E&[U(\tilde{W}^C_{T+})]\\ &= 
     \E\Bigg[-\exp\left(-\alpha\w\left(\tilde{X}^C_{T'+},\tilde{Q}^C_{T'+},\tilde{P}^C_{+},\tilde{\lambda}^C_{T'+}\right)
     +\mathcal{L}_Y\left(\alpha \rr(\tilde{Q}^C_{T'+},\tilde{\lambda}^C_{T'+})\right)\right)\\
     &\qquad\qquad\cdot 
     \exp\left(-\alpha\left(\w\left(\tilde{X}^C_{T+},\tilde{Q}^C_{T+},\tilde{P}^C_{+},\tilde{\lambda}^C_{T+}\right)-\w\left(\tilde{X}^C_{T'+},\tilde{Q}^C_{T'+},\tilde{P}^C_{T'+},\tilde{\lambda}^C_{T'+}\right)\right)\right)\\
     &\qquad\qquad\cdot\exp\left(\mathcal{L}_Y\left(\alpha \rr(\tilde{Q}^C_{T+},\tilde{\lambda}^C_{T+})\right)-\mathcal{L}_Y\left(\alpha \rr(\tilde{Q}^C_{T'+},\tilde{\lambda}^C_{T'+})\right)\right)\Bigg]\\
     &\leq 
     \E\Bigg[-\exp\left(-\alpha \w\left(\tilde{X}^C_{T'+},\tilde{Q}^C_{T'+},\tilde{P}^C_{T'+},\tilde{\lambda}^C_{T'+}\right)
     +\mathcal{L}_Y\left(\alpha \rr(\tilde{Q}^C_{T'+},\tilde{\lambda}^C_{T'+})\right)\right)\\
     &\qquad\qquad\cdot 
     \exp\left(-\alpha\bar{q}\iota(\underline{\lambda})\int_{(T',T]\times E \times [0,\infty)]} \left(\mathbbm{1}_{[0,f(\underline{\lambda})]}(y)\rho^+(e)+2\mathbbm{1}_{[0,g(\infty)]}(y)|\eta(e)\right)N(ds,de,dy)\right)
     \\&\qquad\qquad\cdot\exp\left(-\mathcal{L}_Y' (\alpha\bar{q})\alpha \int_{(T',T]\times E \times [0,\infty)]} \mathbbm{1}_{[0,f(\underline{\lambda})]}(y)\rho^+(e)N(ds,de,dy)\right)\Bigg].
 \end{align}
  The above integrals with respect to $N$ are independent of $\mathcal{F}_{T'}$. Hence, their contribution can be worked out separately and, by the L\'evy-Khintchine formula, there is a constant $\theta>0$ such that it is of the form $\tau(T-T')=\exp(-\theta(T-T'))$. With another application of~\eqref{eq:expectedUtilityTime}, but now for $T'$ instead of $T$, it follows that
  \begin{align}
      \E[U(\tilde{W}^C_{T+})] \leq 
     \E\left[U(\tilde{W}^C_{T'+})\right]\tau(T-T').
  \end{align}
  With $C \in \tilde{C}$ chosen arbitrarily this yields the second estimate in~\eqref{eq:timeRegularityEstimate}.
\end{proof}

\paragraph{Continuity of the value function with respect to the initial liquidity level and inventory.}

Let us next investigate continuity of $v$ with respect to $(\lambda,q) \in [\underline{\lambda},\infty) \times [-\bar{q},\bar{q}]$. By the multiplicative structure of the value function~\eqref{eq:multiplicativeValueFunction}, it suffices to do this for $v_0$, i.e., for $x=0$, $p=0$. 

Our task is to construct from a given strategy $C \in \tilde{\mathcal{C}}(\lambda,q)$ a strategy $C'$ which is admissible when starting from $(\lambda',q') \in [\underline{\lambda},\infty) \times [-\bar{q},\bar{q}]$ that performs at least almost as well as $C$ does. 
As we will show, such a $C'$ can be obtained by executing the same orders as $C$, at least to the extent possible, i.e., without violating the inventory constraint and only using the available liquidity. Before $C$ is halted at time $\tau^C$ (cf.~\eqref{eq:tauC}), this strategy can formally be described via the solution to the SDE-system
\begin{align}
    \tilde{Q}'_{0-}&=q', & \diamond\; d\tilde{Q}'&= \mathbbm{1}_{\{\tilde{Q}'\in [-\bar{q},\bar{q}), \; \lambda^{C'}>\underline{\lambda}\}} \diamond dC^+-\mathbbm{1}_{\{\tilde{Q}'\in (-\bar{q},\bar{q}], \; \lambda^{C'}>\underline{\lambda}\}} \diamond dC^-, \\
    \lambda^{C'}_{0-}&=\lambda', & \diamond \;d\lambda^{C'} &= 
    \diamond\;  \left(-|d\tilde{Q}'|+\left(0\vec{+}(d\tilde{L}'-|d\tilde{M}'|)\vec{+}0\right)\right),\\
    \tilde{L}'_{0-}&=0, & d\tilde{L}'&=\mathbbm{1}_{\{y\leq f(\lambda^{C'}_{-}), \; \lambda^{C'}_- \geq \underline{\lambda}\}}\rho(e)N(dt,de,dy),\\
    \tilde{M}'_{0-}&=0, & d\tilde{M}'&=\mathbbm{1}_{\{y\leq g(\lambda^{C'}_{-}), \; \lambda^{C'}_-\geq\underline{\lambda}\}}\eta(e)N(dt,de,dy).
\end{align}
Indeed, we have 
\begin{align}\label{eq:DefC'}
    \diamond \; dC':=\diamond \left(\;dC'^l\vec{+}0\vec{+}dC'^{r,c}\right):= \diamond \; d\tilde{Q'} \text{ on } [0,\tau^C) \cap [0,T].
\end{align}
We note that, over $[0,\tau^C) \cap [0,T]$, $\lambda^{C'}$ can fall below $\underline{\lambda}$ only due to an exogenous order. If, first, $\tau^C \leq T$, the halt for $C$ can be triggered for three reasons: If it is due to a left-jump trade ($\Delta^l C>0$), $C'$ liquidates its present position in this moment, possibly triggering a halt in the process, and stops operating afterwards. If, second, $C$'s halt is triggered by an external order that has not halted $C'$, the latter still copies, to the extent allowed by liquidity- and position-constraints, any left-jump trade preceding it, lets any external orders take effect, but after that closes shop by liquidating its position in a state-based trade, again possibly triggering a halt itself. Finally, if, third, $C$'s state-based trade triggered the halt, $C'$ seizes operations after a final liquidating state-based trade, which again is allowed to trigger a halt. 

With $C'$ constructed this way, notice that the dynamics of $\lambda^{C'}$ can be written as $$\diamond \;d\lambda^{C'}=\diamond\;\left((-|dC'^{l}|)\vec{+}(d\tilde{L}'-|d\tilde{M}'|)\vec{+}(-|dC'^{r,c}|)\right),$$ in line with the dynamics~\eqref{eq:MarcusDynamicsLiquidityControlled} of $\lambda^C$. In particular, just like $\lambda^C$, also $\lambda^{C'}$ can drop below $\underline{\lambda}$, but, with the exception of a possible final state-based trade by $C'$ at time $\tau^C$, this halt can only be due to external orders from $\tilde{L}'$ and $\tilde{M}'$. These two and also $\tilde{Q}'$ will in turn not move when $\lambda^{C'}_-<\underline{\lambda}$, i.e., right after $C'$ has triggered this trading halt signal. The dynamics of $\tilde{Q}'$, $\tilde{L}'$, and $\tilde{M}'$ thus come to a halt as well. As a result, all these dynamics are in accordance with those of Section~\ref{section:circuitbreaker} in the sense that $\tilde{L}'=\tilde{L}^{C'}$, $\tilde{M}'=\tilde{M}^{C'}$, $\tilde{Q}'=\tilde{Q}^{C'}$.

\begin{lemma}\label{lem:strategyComparison}
    Assume that $\rho,\eta \in L^2(\nu)$ and $g(\infty)<\infty$. There is a constant $c$ with the properties~(i), (ii), and~(iii) below which does neither depend
    on $(\lambda,q), \, (\lambda',q') \in [\underline{\lambda},\bar{\lambda}] \times [-\bar{q},\bar{q}]$ nor on the choice of $C \in \tilde{\mathcal{C}}(\lambda,q)$ provided that 
\begin{align}
    \label{eq:TriggerOnlyWhenUnwinding}
    \text{$C$ only actively triggers a halt when reducing its position $|Q^C|$, not while expanding it.}
\end{align}
    \begin{itemize}
        \item[(i)] Until it runs into a trading halt, the realizable wealth from strategy $C'$ remains close to that of $C$ in the sense that 
        \begin{align}
         \label{eq:wealthClose}
        |\tilde{W}^C&-\tilde{W}^{C'}|\\&\leq c(|\lambda-\lambda'|+|q-q'|+V^{C,C'})(1+\lambda\vee\lambda'-\underline{\lambda}+\bar{L}^+ +|Y|)\text{ on }   [0,\tau^{C'} \wedge T],
        \end{align}
        where 
        \begin{align}
            V^{C,C'}_. :=V_{[0,.]}(\tilde{L}^C-\tilde{L}^{C'})+V_{[0,.]}(\tilde{M}^C-\tilde{M}^{C'})
        \end{align}
        denotes the cumulative external order volume that exclusively affects the trading environment of just one of the strategies.
        
        \item[(ii)] In expectation the exclusive order volume $V^{C,C'}$ is controlled in the sense that
        \begin{align}\label{eq:controlV}
        \E[V^{C,C'}_{T \wedge \tau^C \wedge \tau^{C'}}]\leq c(|\lambda-\lambda'|+|q-q'|)
        \end{align}
        and
        \begin{align}\label{eq:controlV2}
        \E[(V^{C,C'}_{T \wedge \tau^C \wedge \tau^{C'}})^2]^{1/2}\leq c(|\lambda-\lambda'|+|q-q'|)^{1/2}.
        \end{align}
        
        \item[(iii)] The probability of an early halt for $C'$ is \begin{align}\label{eq:controlTriggerProbability}
            \PP[\tau^{C'}<\tau^C \wedge T] \leq c\left(w_{\rho^-+|\eta|}\left(c(|\lambda-\lambda'|+|q-q'|\right)+ |\lambda-\lambda'|+|q-q'|)\right)
        \end{align}
        where $w_{\rho^-+|\eta|}$ is the concave envelope of 
        the modulus of continuity for $\nu(\rho^-+|\eta|>.)$ on~$(0,\infty)$.
    \end{itemize}
\end{lemma}
\begin{remark}
    The preceding lemma establishes a precise control about how close $C'$ constructed above manages to track $C$. Property~\eqref{eq:TriggerOnlyWhenUnwinding} will only be violated by clearly suboptimal strategies which are of no interest for the the value function; cf.~also the proof of Corollary~\ref{cor:continuityLiquidity}.
\end{remark}
\begin{proof}
    For ease of notation, let us introduce $\tau:=\tau^C$, $\tau':=\tau^{C'}$, $\tilde{\lambda}:=\tilde{\lambda}^{C}$, $\tilde{\lambda}':=\tilde{\lambda}^{C'}$,  $\tilde{W}:=\tilde{W}^{C}_+$, $\tilde{W}':=\tilde{W}^{C'}_+$ as well as $\tilde{Q}:=\tilde{Q}^C$, $\tilde{L}:=\tilde{L}^C$, $\tilde{M}:=\tilde{M}^C$.
    
    The key observation here is that on $[0,\tau'\wedge T]$, the distance $|\tilde{Q}-\tilde{Q}'|+|\tilde{\lambda}-\tilde{\lambda}'|$ can only grow when an external order affects just one of the strategies in the sense that $V:=V^{C,C'}$ grows. Indeed, with $d\hat{Q}:=d\tilde{Q}-d\tilde{Q}'$ denoting the extra trades of $C$ over $C'$, and similarly $d\hat{L}$, $d\hat{M}$ the differences in order flow for $C$ and $C'$, we compute
    \begin{align}
        \diamond \; d(|\tilde{Q}-\tilde{Q}'|+|\tilde{\lambda}-\tilde{\lambda}'|)=&\;\sgn(\tilde{Q}-\tilde{Q}')\diamond d\hat{Q}+\sgn(\tilde{\lambda}-\tilde{\lambda}') \diamond \left(-|d\hat{Q}|+\left(0\vec{+}(d\hat{L}-|d\hat{M}|)\vec{+}0\right)\right)\\
        =&\; (\sgn(\tilde{Q}-\tilde{Q}')-\sgn(\tilde{\lambda}-\tilde{\lambda}')) \diamond d\hat{Q}^+
        \\&+(-\sgn(\tilde{Q}-\tilde{Q}')-\sgn(\tilde{\lambda}-\tilde{\lambda}')) \diamond d\hat{Q}^-\\
        &+\sgn(\tilde{\lambda}-\tilde{\lambda}')\diamond\left(0\vec{+}(d\hat{L}-|d\hat{M}|)\vec{+}0\right).
    \end{align}
    Before time $\tau \wedge \tau'$, an extra buy order by $C$ where $C'$ cannot follow suit (``$d\hat{Q}^+>0$'') can only happen when $\tilde{Q}'=\bar{q}$ or when $\tilde{\lambda}'=\underline{\lambda}$. Therefore, the $d\hat{Q}^+$-integrand in the preceding expression will be $-1$ or even $-2$  in such a moment. Similarly, a sell order exclusive to $C$ (``$d\hat{Q}^->0$'') can only happen when $\tilde{Q}'=-\bar{q}$ or when $\tilde{\lambda}'=\underline{\lambda}$ and, then, the above $d\hat{Q}^-$-integrand will be $-1$ or $-2$. As a result, orders from $C$ or $C'$ can only decrease the distance $|\tilde{Q}-\tilde{Q}'|+|\tilde{\lambda}-\tilde{\lambda}'|$  and this distance increases by at most $dV=|d\hat{L}|+|d\hat{M}|$. It follows that
    \begin{align}\label{eq:distanceEvolution}
        |\tilde{Q}-\tilde{Q}'|+|\tilde{\lambda}-\tilde{\lambda}'|\leq|q-q'|+|\lambda-\lambda'|+V \text{ on } [0,\tau \wedge \tau' \wedge T].
    \end{align}
    Note that we can include $\tau \wedge \tau' \wedge T$ in~\eqref{eq:distanceEvolution} as, by~\eqref{eq:TriggerOnlyWhenUnwinding} and by construction of $C'$, $C$ and $C'$ will both be unwinding their position in this moment if $\tau \leq \tau' \wedge T$. In case $\tau'<\tau \wedge T$, we note that this can only happen due to external orders and so~\eqref{eq:distanceEvolution} holds in this moment because its left-hand side remains unchanged while the right-hand side can only increase.
    Moreover, we can use the above $\sgn$-observation to estimate the excess volume of trades by $C$ over those of $C'$ by 
    \begin{align}
        \int_0^. \diamond\; |d\hat{Q}|&\leq -\int_0^. (\sgn(\tilde{Q}-\tilde{Q}')-\sgn(\tilde{\lambda}-\tilde{\lambda}')) \diamond d\hat{Q}^+
        +(-\sgn(\tilde{Q}-\tilde{Q}')-\sgn(\tilde{\lambda}-\tilde{\lambda}')) \diamond d\hat{Q}^-
        \\&= -\int_0^. \diamond \;d|\tilde{Q}-\tilde{Q}'|+d|\tilde{\lambda}-\tilde{\lambda}'|-\int_0^.\sgn(\tilde{\lambda}-\tilde{\lambda}')\diamond\left(0\vec{+}(d\hat{L}-|d\hat{M}|)\vec{+}0\right)\\
        &\leq|q-q'|+|\lambda-\lambda'|+V \text{ on } [0,\tau \wedge \tau' \wedge T].\label{eq:excessOrderVolumeEstimate}
    \end{align}
    
    \begin{itemize}
    \item[(i)] 
    In the following, we denote by $c$ a generic constant which may grow from line to line, but never depends on $\lambda$, $\lambda'$, $q$, $q'$, $C$, $C'$. Let $c$ be a joint Lipschitz-constant for $\w(T,0,.,0,.)$, $\rr$, and for the functions
    \begin{align}
        a(q,\lambda):=I(|q|,\lambda)-|q|\iota(\lambda) \text{ and }  b^\pm(q,\lambda):=I(|q|,\lambda)-2q^\pm\iota(\lambda). 
    \end{align}
    Representing both $\tilde{W}$ and $\tilde{W}'$ as in~\eqref{eq:TerminalWealthRepresentation} allows us to write 
    \begin{align}
        \tilde{W}-\tilde{W}'
        =&\;\w(T,0,q,0,\lambda)-\w(T,0,q',0,\lambda')\\&+\int_0^.a(\tilde{Q},\tilde{\lambda}) \diamond \left(0\vec{+}d\tilde{L}^C\vec{+}0\right)
        - \int_0^.a(\tilde{Q}',\tilde{\lambda}') \diamond \left(0\vec{+}d\tilde{L}^{C'}\vec{+}0\right)\\
    &-\left(\int_0^{.}  b^+(\tilde{Q},\tilde{\lambda}) \diamond \left(0\vec{+}d\tilde{M}^{C,+}\vec{+}0\right)
     - \int_0^{.} b^+(\tilde{Q'},\tilde{\lambda}') \diamond \left(0\vec{+}d\tilde{M}^{C',+}\vec{+}0\right)\right)\\&
    -\left(\int_0^{.}b^-(\tilde{Q},\tilde{\lambda}) \diamond \left(0\vec{+}d\tilde{M}^{C,-}\vec{+}0\right)
    -\int_0^{.}b^-(\tilde{Q'},\tilde{\lambda}') \diamond \left(0\vec{+}d\tilde{M}^{C',-}\vec{+}0\right)\right)\\
    &-\left(\int_0^. 2b^+(\tilde{Q},\tilde{\lambda})\diamond d\tilde{Q}^+-\int_0^.2b^+(\tilde{Q}',\tilde{\lambda}')\diamond d\tilde{Q}'^+\right)\\
        &-\left(\int_0^. 2b^-(\tilde{Q},\tilde{\lambda})\diamond d\tilde{Q}^--\int_0^. 2b^-(\tilde{Q}',\tilde{\lambda}')\diamond d\tilde{Q}'^-\right)\\
        &-\zeta \left((V_{[0,.]}(C)+|\tilde{Q}|)-(V_{[0,.]}(C')+|\tilde{Q'}|)\right)
        \\&+Y\left(\rr(\tilde{Q},\tilde{\lambda})-\rr(\tilde{Q'},\tilde{\lambda}')\right).
    \end{align}
    Note that, contrary to~\eqref{eq:TerminalWealthRepresentation}, we opt here to express the difference in spread costs via the difference in position variation and the difference of terminal positions that need to be liquidated. The first amounts to $\int_0^. \diamond\;|d\hat{Q}|$ and is thus controlled by~\eqref{eq:excessOrderVolumeEstimate}; the second is controlled by~\eqref{eq:distanceEvolution}. Both controls are of the desired form. 
    
    Since $c$ is a Lipschitz-constant for the function $\w(T,0,.,0,.)$, the first term contributes at most
    \begin{align}
        \label{eq:contr5}
        |\w(T,0,q,0,\lambda)-\w(T,0,q',0,\lambda')| \leq c(|q-q'|+|\lambda-\lambda'|).
    \end{align}
    Also $\rr$ is Lipschitz-continuous with this constant, and so the last term contributes at most $|Y|c(|\tilde{Q}-\tilde{Q'}|+|\tilde{\lambda}-\tilde{\lambda}'|)$, which, by~\eqref{eq:distanceEvolution}, is not more than \begin{align}\label{eq:YContribution}
    |Y|c(|q-q'|+|\lambda-\lambda'|+V)\text{ on } [0,\tau \wedge \tau' \wedge T].
    \end{align}
    
    In fact, we can proceed similarly for the other terms. When a limit order or cancellation affects the environment of both $C$ and $C'$, the integrators for the $a$-terms coincide and the absolute value of the net contribution is found to be no more than
    \begin{align}
        |a(\tilde{Q},\tilde{\lambda}) - a(\tilde{Q'},\tilde{\lambda}')| \diamond \left(0\vec{+}|d\tilde{L}^{C}|\vec{+}0\right)&\leq c(|\tilde{Q}-\tilde{Q'}|+|\tilde{\lambda}-\tilde{\lambda}'|) \diamond \left(0\vec{+}d\tilde{L}^{C,+}\vec{+}0\right).
    \end{align}
    Using~\eqref{eq:distanceEvolution}  in the same way as we did for the $Y$-term, this gives a contribution
    \begin{align}
        |a(\tilde{Q},\tilde{\lambda}) - a(\tilde{Q'},\tilde{\lambda}')| \diamond \left(0\vec{+}|d\tilde{L}^{C}|\vec{+}0\right)&\leq c(|q-q'|+|\lambda-\lambda'|+V) \diamond \left(0\vec{+}|d\tilde{L}^{C}|\vec{+}0\right)
    \end{align}
    for limit orders and cancellations $d\tilde{L}^{C}$ common to $C$ and $C'$. We similarly find 
    \begin{align}
        |b^\pm(\tilde{Q},\tilde{\lambda}) - b^\pm(\tilde{Q'},\tilde{\lambda}')| \diamond \left(0\vec{+}d\tilde{M}^{C,\pm}\vec{+}0\right)&\leq c(|q-q'|+|\lambda-\lambda'|+V) \diamond \left(0\vec{+}d\tilde{M}^{C,\pm}\vec{+}0\right)
    \end{align}
    for common market buy and sell orders. Upon aggregation, we obtain a contribution from non-exclusive external orders which is no more than
    \begin{align}\label{eq:contr4}
        c(|q-q'|+|\lambda-\lambda'|+V)(V_{[0,.]}(2\tilde{M}^C)+V_{[0,.]}(\tilde{L}^C))\text{ on } [0,\tau \wedge \tau' \wedge T].
    \end{align}
    When a market buy or a sell order 
    affects only one of the environments, but not the other, the resulting absolute change in realizable wealth due to the $b^{\pm}$-integrals is no more than its volume times $\max |b^{\pm}| \leq 2\bar{q}\iota(\underline{\lambda})$. Similarly, from the $a$-integrals we collect no more than the exclusive volume times $\max |a| \leq \bar{q}\iota(\underline(\lambda))$. All of these contributions aggregate to no more than a multiple of \begin{align}\label{eq:contr3}
    \bar{q}\iota(\underline(\lambda))V\text{ on } [0,\tau \wedge \tau' \wedge T].
    \end{align}
    
     We can treat the $dQ^{\pm}$-integrals similarly as well: For common trades of $C$ and $C'$ we can net the integrands collecting a contribution of no more than
     \begin{align}\label{eq:contr2}
        c(|q-q'|+|\lambda-\lambda'|+V)V_{[0,.]}(\tilde{Q})\text{ on } [0,\tau \wedge \tau' \wedge T].
    \end{align}
    The contributions when only $C$ trades while $C'$ has to pause
    are bounded by $2\bar{q}\iota(\underline{\lambda})$ times the aggregate extra buy or sell volume $\int_0^. \diamond \; |d\hat{Q}|$ that $C$ does on top of $C'$. Due to~\eqref{eq:excessOrderVolumeEstimate}, we thus pick up a contribution of at most
    \begin{align}\label{eq:contr1}
     2\bar{q}\iota(\underline{\lambda}) ( |q-q'|+|\lambda-\lambda'|+V) \text{ on } [0,\tau \wedge \tau' \wedge T] 
    \end{align}
     from these $C$-exclusive orders.

     Adding the contribution estimates~\eqref{eq:contr5},~\eqref{eq:YContribution}, \eqref{eq:contr4}, \eqref{eq:contr3}, \eqref{eq:contr2}, \eqref{eq:contr1} to the spread costs controlled by~\eqref{eq:excessOrderVolumeEstimate} and~\eqref{eq:distanceEvolution}, we find that
     \begin{align}
         |\tilde{W}-\tilde{W'}|\leq&\; c(|q-q'|+|\lambda-\lambda'|+V)(1+|Y|+2V_{[0,.]}(\tilde{M}^C)+V_{[0,.]}(\tilde{L}^C)+2V_{[0,.]}(\tilde{Q}))\\
         &+2c\bar{q}\iota(\underline(\lambda))( |q-q'|+|\lambda-\lambda'|+V)+2\zeta( |q-q'|+|\lambda-\lambda'|+V)
     \end{align}
     on $[0,\tau \wedge \tau' \wedge T]$. 
     Due to~\eqref{eq:liquidityBound}, the above total variation terms are jointly dominated by~$c(\lambda\vee\lambda'-\underline{\lambda}+\tilde{L}^+$). Finally,
     yet again exchanging $c$ for a higher constant again denoted by $c$ and still independent of the initial data and $C$ and $C'$, we arrive at~\eqref{eq:wealthClose} over $[0, \tau \wedge \tau' \wedge T]$. Since in case $\tau \leq \tau' \wedge T$ both $C$ and $C'$ are unwinding their position in this moment, their realizable portfolio values do not change (cf.\ also~\eqref{eq:WealthDynamics}) and thus remain as close as before. Since $C'$ will also stop trading after $\tau$,~\eqref{eq:wealthClose} extends to $[0, \tau' \wedge T]$ in any case.

    \item[(ii)] 
    We can express $V$ in terms of the point process $N$ as
    \begin{align}\label{eq:VasNIntegral}
        V_{t \wedge \tau \wedge \tau'} = \int_{[0,t \wedge \tau\wedge \tau'] \times E \times [0,\infty)} &\left(\mathbbm{1}_{\{g(\tilde{\lambda}_{s-}) \wedge g(\tilde{\lambda}'_{s-}) \leq y<g(\tilde{\lambda}_{s-}) \vee g(\tilde{\lambda}'_{s-})\}}|\rho(e)|\right.\\
        &\quad\left. + \mathbbm{1}_{\{f(\tilde{\lambda}_{s-}) \wedge f(\tilde{\lambda}'_{s-}) \leq y<f(\tilde{\lambda}_{s-}) \vee f(\tilde{\lambda}'_{s-})\}}|\eta(e)|\right)
        N(ds,de,dy).
    \end{align}
    Upon taking expectation, we can pass to the compensator $dt \otimes \nu(de) \otimes dy$ to obtain
    \begin{align}
        \E[V_{t \wedge \tau \wedge \tau'}] &= \E\Bigg[ \int_{0}^{t \wedge \tau\wedge \tau'} \left(|g(\tilde{\lambda}_{s-})- g(\tilde{\lambda}'_{s-})|\int_E |\rho(e)|\nu(de)\right. \\
        & \qquad\qquad\qquad\left. + |f(\tilde{\lambda}_{s-})- f(\tilde{\lambda}'_{s-})|\int_E |\eta(e)|\nu(de)\right) ds\Bigg].
    \end{align}
    Taking $c$ to be a Lipschitz constant for $f$ and $g$ and using~\eqref{eq:distanceEvolution}, we get 
    \begin{align}
        \E[V_{t \wedge \tau \wedge \tau'}] & \leq \E\Bigg[ \int_{0}^{t \wedge \tau\wedge \tau'} c(|\lambda-\lambda'|+|q-q'|+V_{s})\left(\int_E |\rho(e)|\nu(de)+\int_E |\eta(e)|\nu(de)\right)  ds\Bigg]\\
        & \leq \int_E (|\rho(e)|+ |\eta(e)|)\nu(de) \, \left(c(|\lambda-\lambda'|+|q-q'|)t+c\int_{0}^{t} \E[V_{s \wedge \tau\wedge \tau'}] ds\right).
    \end{align}
    Our estimate for $\E[V_{t \wedge \tau \wedge \tau'}]$ now follows (with a modified $c$ taking into account the above $\nu$-integral) from Gronwall's lemma.

    To get the bound for $\E[V^2_{t \wedge \tau \wedge \tau'}]^{1/2}$, we consider again~\eqref{eq:VasNIntegral} and use Lemma~\ref{lemma:PoissonIntegralMomentEstimate} to conclude that for some constant $c$ depending only on $g(\infty)\vee f(\underline{\lambda})$
    \begin{align}
        \E[V^2_{T  \wedge \tau\wedge \tau'}]^{1/2} & \leq c\Bigg(\int_0^T \int_E \int_0^{g(\infty)\vee f(\underline{\lambda})} \E\Big[\left(\mathbbm{1}_{\{g(\tilde{\lambda}_{s-}) \wedge g(\tilde{\lambda}'_{s-}) \leq y<g(\tilde{\lambda}_{s-}) \vee g(\tilde{\lambda}'_{s-})\}}|\rho(e)|^2\right.\\
        &\qquad\left. + \mathbbm{1}_{\{f(\tilde{\lambda}_{s-}) \wedge f(\tilde{\lambda}'_{s-}) \leq y<f(\tilde{\lambda}_{s-}) \vee f(\tilde{\lambda}'_{s-})\}}|\eta(e)|^2\right)\mathbbm{1}_{\{s \leq T \wedge \tau \wedge \tau'\}}\Big] dy \nu(de) ds\Bigg)^{1/2}\\
        &=c\Bigg(\int_0^T \E\Big[\Big(|g(\tilde{\lambda}_{s-})-g(\tilde{\lambda}'_{s-})|\int_E |\rho(e)|^2\nu(de)\\
        &\qquad\qquad\qquad + |f(\tilde{\lambda}_{s-}) -f(\tilde{\lambda}'_{s-})|\int_E|\eta(e)|^2\nu(de)\Big)\mathbbm{1}_{\{s \leq T \wedge \tau \wedge \tau'\}}\Big]   ds\Bigg)^{1/2}.
    \end{align}
    From here we can again proceed by a Lipschitz-estimate and use~\eqref{eq:distanceEvolution} to arrive at
    \begin{align}
        \E[V^2_{t \wedge \tau \wedge \tau'}]^{1/2} & \leq c\Bigg(Tc(|\lambda-\lambda'|+|q-q'|+\E[V_{T  \wedge \tau \wedge \tau'}])\int_E (|\rho(e)|^2+ |\eta(e)|^2)\nu(de)\Big]   ds\Bigg)^{1/2}.
    \end{align}
    An application of~\eqref{eq:controlV} now allows us to conclude~\eqref{eq:controlV2} by passing to a higher constant $c$.
    
    \item[(iii)] Let us denote by $\Gamma_s(z(e,y))$, $\Gamma'_s(z(e,y))$ the signal-based trades of, respectively, $C$ and $C'$ from Lemma~\ref{Lemma:decomposition_q}. Observe that $\tau'<\tau \wedge T$ can only happen if at time $\tau'=s$ the point process $N(ds,de,dy)$ puts a mark $(e,y)$ which results in too high a market order or cancellation to be covered by the  liquidity $R'_s(e,y):=\tilde{\lambda}'_{s-}-|\Gamma'_s(z(e,y))|-\underline{\lambda}$ remaining after a possible signal-based trade by $C'$, while $C$ is either not exposed to this order or its remaining liquidity $R_s(e,y):=\tilde{\lambda}_{s-}-|\Gamma_s(z(e,y))|-\underline{\lambda}$ suffices to cover that order without triggering a halt. Recalling that $\rho(e) \eta(e) \equiv 0$, this means in particular that a mark $(e,y)$ has to be placed in one of the mark sets
    \begin{align}
        M^1_s&:=\{(e,y) \in E \times [0,g(\infty) \vee f(\underline{\lambda})]\;:\;
        R'_s(e,y)< \rho^-(e)+|\eta(e)| \leq R_s(e,y)\}\\
        M^2_s&:=\{(e,y)\in E \times [0,g(\infty) \vee f(\underline{\lambda})]\;:\;y \in (g(r),g(r')] \cup (f(r),f(r')] \\&\qquad\qquad\qquad\qquad\qquad\qquad\qquad\qquad\text{ where } r=R_s(e,y)+\underline{\lambda}, \;r'=R'_s(e,y)+\underline{\lambda}\}.
    \end{align}
    Notice that with $w$ denoting a modulus of continuity for $\nu(\rho^-+|\eta|>.)$ on $(0,\infty)$ we can estimate for any $y$-section $M^{1,y}_s$ of $M^1_s$ its $\nu$-measure
    \begin{align}
        \nu(M^{1,y}_s) \leq w(R_s(e,y)-R'_s(e,y)) \leq w(\tilde{\lambda}_{s-}-\tilde{\lambda}'_{s-})
    \end{align}
    where for the second estimate we used that $|\Gamma'(z(e,y))| \leq |\Gamma(z(e,y)|$ (recalling that $C'$ always trades at most as much as $C$). For the same reason, we can estimate the $dy$-measure of the $e$-sections $M^{2,e}_s$ of $M^2_s$ by
    \begin{align}
        dy(M^{2,e}_s) \leq 2\max_{y \in [0,g(\infty)\vee f(\underline{\lambda})]}|g(r')-g(r)|\vee|f(r')-f(r)|\leq c|\tilde{\lambda}_{s-}-\tilde{\lambda}'_{s-}|,
    \end{align}
    where as before $r=R_s(e,y)+\underline{\lambda}$, $r'=R'_s(e,y)+\underline{\lambda}$.
    It follows that
    \begin{align}
        \PP[\tau'<\tau \wedge T] & \leq \E\left[\int_{[0,\tau' \wedge T]\times E \times [0,g(\infty) \vee f(\underline{\lambda})]} (\mathbbm{1}_{M^1_s}(e,y)+\mathbbm{1}_{M^2_s}(e,y))
        N(ds,de,dy)\right]\\
        &\leq \E\left[\int_{0}^{\tau' \wedge T} 
        (g(\infty) \vee f(\underline{\lambda}))w(\tilde{\lambda}_{s-}-\tilde{\lambda}'_{s-})
        +\nu(E)c|\tilde{\lambda}_{s-}-\tilde{\lambda}'_{s-}|) ds\right]\\
        &\leq c\int_{0}^{T}\E\left[ 
        (w(\delta)+ \delta)|_{\delta=|\lambda-\lambda'|+|q-q'|+V_{\tau'\wedge T}}
)\right]  ds\\
&\leq cT (w_{\rho^-+|\eta|}(\delta)+ \delta)|_{\delta=|\lambda-\lambda'|+|q-q'|+\E[V_{\tau'\wedge T}]},
    \end{align}
    where we used~\eqref{eq:distanceEvolution} in the last but one step and Jensen's inequality for the concave envelope $w_{\rho^-+|\eta|}$ of $w$ in the last one. Assertion~(iii) now follows from the estimate in~(ii). 
\end{itemize}
\end{proof}

With this lemma at hand we can now establish the main result of this paragraph:

\begin{corollary}\label{cor:continuityLiquidity}
    Assume that $\nu(\rho^-+\eta>.)$ is continuous on $(0,\infty)$ and suppose that $g(\infty)<\infty$ and $\eta,\rho \in L^2(\nu)$ with $\rho^+$ even satisfying the exponential integrability condition~\eqref{eq:rhoExponentialIntegrable}. Then $v_0(T,.,.)=v(T,.,.,0,0)$ is continuous on $[\underline{\lambda},\infty)\times [-\bar{q},\bar{q}]$.
\end{corollary}
\begin{proof}
    Choose $\bar{\lambda}>\underline{\lambda}$. Let us start by proving the exponential moment estimate
    \begin{align}\label{eq:expMomentEstimateWealth} 
        \sup_{(\lambda,q) \in [\underline{\lambda},\bar{\lambda}]\times[-\bar{q},\bar{q}]} \sup_{C \in  \tilde{\mathcal{C}}(\lambda,q)}\E[\exp(-c W^C_{T+})]  <\infty \text{ for any } c>0.
    \end{align}
    For this, consider again the terminal wealth representation~\eqref{eq:TerminalWealthRepresentation} and observe that it yields the lower bound
    \begin{align}
    \tilde{W}^C_{T+} &\geq \w(0,q,0,\lambda)+Y\rr(\tilde{Q}^C_{T+},\tilde{\lambda}^C_{T+})  -  \int_0^{T+} I(|\tilde{Q}^C|,\tilde{\lambda}^C) \diamond \left(0\vec{+}(d\tilde{L}^{C,-}+|d\tilde{M}^C|)\vec{+}0\right)\\
    &\qquad-2\int_0^{T+} \left( I(|\tilde{Q}^C|,\tilde{\lambda}^C)-|\tilde{Q}^C|\iota(\tilde{\lambda}^C)+\zeta\right)\left(\mathbbm{1}_{\{\tilde{Q}^C>0\}} \diamond d\tilde{Q}^{C,+}+\mathbbm{1}_{\{\tilde{Q}^C<0\}} \diamond d\tilde{Q}^{C,-}) \right)\\
    &\geq \w(x,q,p,\lambda)+Y\rr(\tilde{Q}^C_{T+},\tilde{\lambda}^C_{T+})-\bar{q}\iota(\underline{\lambda})(\tilde{L}^{C,-}_T+V_{[0,T]}(\tilde{M}^C))-2(\bar{q}\iota(\underline{\lambda})+\zeta)V_{[0,T]}(\tilde{Q}^C)
\end{align}
With~\eqref{eq:liquidityBound} and using the independence of $Y$ from $\mathcal{F}_T$ and its symmetry (see~\eqref{eq:YProperties}), we thus get with the inventory bound from~\eqref{eq:inventorybound} for any order $c>0$ that
\begin{align}
    \E[\exp(-c\tilde{W}^C_T)] \leq \exp&\left(-c(\w(0,q,0,\lambda)-\bar{q}\iota(\underline{\lambda})(\lambda-\underline{\lambda}))+\mathcal{L}_Y(c\bar{q})\right)
    \E\left[\exp\left(c(3\bar{q}\iota(\underline{\lambda})+2\zeta)\bar{L}^+_T\right)\right].
\end{align}
By the L\'evy-Khintchine formula, the right-hand side is finite under condition~\eqref{eq:rhoExponentialIntegrable}. Since the above estimate is uniform in $C \in \tilde{\mathcal{C}}(\lambda,q)$ and $\w(0,\lambda,0,q)$ is bounded from below uniformly in $(\lambda,q) \in [\underline{\lambda},\bar{\lambda}]\times[-\bar{q},\bar{q}]$, this estimate yields~\eqref{eq:expMomentEstimateWealth} .
    
    Take now $(\lambda,q),(\lambda',q') \in [\underline{\lambda},\bar{\lambda}]\times[-\bar{q},\bar{q}]$.
    For any $\epsilon>0$, we can take an $\epsilon$-optimal $C \in \tilde{\mathcal{C}}(\lambda,q)$ and assume that it triggers a halt only when reducing its exposure $|Q^C|$, i.e.,  it satisfies~\eqref{eq:TriggerOnlyWhenUnwinding}. Indeed, triggering a halt while expanding the exposure means executing a loss-leading round trip because of the forced unwinding of any remaining position after a halt has been triggered.
    
    Now consider the associated $C'\in\tilde{\mathcal{C}}(\lambda',q')$ from the preceding lemma for the estimate
    \begin{align}
        &v_0(T,\lambda,q)-v_0(T,\lambda',q')\\&\leq \epsilon + \E[U(W^C_{T+})]-\E[U(\tilde{W}^{C'}_{T+}]\\
        &\leq \epsilon+\E[U'(\tilde{W}^C_{T+})(\tilde{W}^C_{T+}-\tilde{W}^{C'}_{T+})\mathbbm{1}_{\{\tau^{C'} \geq \tau^{C} \wedge T\}}]
        +\E[(U(\tilde{W}^C_{T+})-U(\tilde{W}^{C'}_{T+}) )\mathbbm{1}_{\{\tau^{C'} < \tau^{C} \wedge T\}}], \label{eq:UtilityEstimate}
    \end{align}
    where we used the concavity of the utility function $U$. 
    
    Let us denote in the following by $c$ again a generic constant that does not depend on the initial data and neither on $C$ or $C'$. 

    Note that on $\{\tau^C \leq \tau^{C'}<T\}$, both $C$ and $C'$ liquidate their positions at the same time and then stop operating. So $\tilde{W}^C_{T+}=\tilde{W}^C_{\tau^{C'}+}=\tilde{W}^C_{\tau^{C'}}$ and $\tilde{W}^{C'}_{T+}=\tilde{W}^{C'}_{\tau^{C'}+}=\tilde{W}^{C'}_{\tau^{C'}}$ on this set. 
    Using this, applying then Lemma~\ref{lem:strategyComparison}~(i) followed by H\"older's inequality,  we find that the first expectation in~\eqref{eq:UtilityEstimate} is
    \begin{align}
        \E[&U'(\tilde{W}^C_{T+})(\tilde{W}^C_{T+}-\tilde{W}^{C'}_{T+})\mathbbm{1}_{\{\tau^{C'} \geq \tau^{C} \wedge T\}}]\\
        & = \E[U'(\tilde{W}^C_{T+})(\tilde{W}^C_{\tau^{C'} \wedge T+}-\tilde{W}^{C'}_{\tau^{C'}\wedge T+})\mathbbm{1}_{\{\tau^{C'} \geq \tau^{C} \wedge T\}}]\\
        &\leq
        \|U'(\tilde{W}^C_{T+})\|_{L^4(\PP)}
        \|1+\bar{\lambda}-\underline{\lambda}+\bar{L}^+_T+|Y|\|_{L^4(\PP)} \|c(|\lambda-\lambda'|+|q-q'|+V^{C,C'}_{\tau^{C} \wedge \tau^{C'} \wedge T})\|_{L^2(\PP)} \\
        &\leq c(|\lambda-\lambda'|+|q-q'|+(|\lambda-\lambda'|+|q-q'|)^{1/2}),
    \end{align}
    where the final inequality is due to~\eqref{eq:expMomentEstimateWealth} and Lemma~\ref{lem:strategyComparison}~(ii). Clearly, there is $\delta(\epsilon)>0$ such that for initial data from $[\underline{\lambda},\bar{\lambda}]\times[-\bar{q},\bar{q}]$ with $|\lambda-\lambda'|+|q-q'|=:\delta<\delta(\epsilon)$ the above expression is less than $\epsilon$.

    Again by~\eqref{eq:expMomentEstimateWealth}, the utilities in the second expectation in~\eqref{eq:UtilityEstimate} are from a uniformly integrable family of random variables. They are integrated over the set $\{\tau^{C'}<\tau^{C} \wedge T\}$ whose probability, due to Lemma~\ref{lem:strategyComparison}~(iii), is uniformly controlled by $ c(w_{\rho^-+|\eta|}\left(\delta\right)+\delta+\delta^{1/2})$ where again $\delta=|\lambda-\lambda'|+|q-q'|$. So, we can decrease $\delta(\epsilon)>0$ even further and still in a way which does not dependent on $(\lambda,q)$, $(\lambda',q')$ or $C$, $C'$ such that for all initial data satisfying $|\lambda-\lambda'|+|q-q'|=\delta<\delta(\epsilon)$ the second expectation in~\eqref{eq:UtilityEstimate} is less than $\epsilon$.

    It follows that if $|\lambda-\lambda'|+|q-q'|<\delta(\epsilon)$ then
    \begin{align}
        v_0(T,\lambda,q)&-v_0(T,\lambda',q')
        \leq 
        3\epsilon.
    \end{align}
    Exchanging the roles of $(\lambda,q)$ and $(\lambda',q')$ in the above argument yields $\delta(\epsilon)>0$ such that
    \begin{align}
        |v_0(T,\lambda,q)&-v_0(T,\lambda',q')|
        \leq 3\epsilon
    \end{align}
    for all $(\lambda,q),(\lambda',q') \in [\underline{\lambda},\bar{\lambda}] \times [-\bar{q},\bar{q}]$ with $|\lambda-\lambda'|+|q-q'|<\delta(\epsilon)$. This proves continuity of $v_0(T,.,.)$ on this domain.
\end{proof}

\paragraph{Continuity of the value function.}

Continuity of the value function on $[\underline{\lambda},\infty) \times [-\bar{q},\bar{q}]$ follows by recalling its multiplicative structure~\eqref{eq:multiplicativeValueFunction} and by combining Lemma~\ref{lemma:continuity_time} with~Corollary~\ref{cor:continuityLiquidity}.

\section{A moment estimate for Poisson-integrals}

For the proof of continuity of the value function, we used the following moment estimate for the integral of a predictable mark-dependent process against a Poisson point process. We record it here for the sake of completeness as we could not find a result to this effect in the literature:

\begin{lemma}\label{lemma:PoissonIntegralMomentEstimate}
    On $(\Omega,\mathcal{F},\PP)$, let $N=N(ds,de)$ be a marked Poisson point process with finite intensity measure $ds \otimes \nu(de)$ and marks in $(E,\mathcal{E})$. Suppose that $(\mathcal{F}_t)$ is a filtration to which $N$ is adapted in the sense that for any $E' \in \mathcal{E}$ the process $N([0,t]\times E') \in \mathcal{F}_t$, $t \geq 0$.
    
    Then, for $n=0,1,\dots$, there is a constant $c_n$ depending only on $T\nu(E)$ such that for any $\mathcal{P}\otimes \mathcal{E}$-measurable $\alpha$ we have
    \begin{align}\label{eq:PoissonEstimate}
        \left\|\int_{[0,T] \times E}|\alpha_s(e)| N(ds,de)\right\|_{L^{2^n}(\PP)} &\leq c_n \|\alpha\|_{L^{2^n}(\PP\otimes ds|_{[0,T]} \otimes \nu)}
        .
    \end{align}
\end{lemma}
\begin{proof}
    As the claim is clearly true for $n=0$ with $c_0:=1$, we can proceed by induction. 
    
    It suffices to consider $\alpha \in L^1(\PP \otimes ds \otimes \nu(de))$. For such $\alpha$, we have that the compensated integral $\int_{[0,.] \times E}|\alpha_s(e)| \bar{N}(ds,de)$ is a martingale. This allows us to apply the Burkholder--Davis--Gundy inequality to obtain a universal constant $b_n$ such that
    \begin{align}
        &\left\|\int_{[0,T] \times E}|\alpha_s(e)| \bar{N}(ds,de)\right\|_{L^{2^n}(\PP)}\\ &\leq b_n\left\|\left[\int_{[0,.] \times E}|\alpha_s(e)| \bar{N}(ds,de)\right]_T^{1/2}\right\|_{L^{2^{n}}(\PP)}
        = b_n\left\|\int_{[0,T] \times E}|\alpha_s(e)|^2 N(ds,de)\right\|^{1/2}_{L^{2^{n-1}}(\PP)}\\
        &\leq b_n \left(c_{n-1}\||\alpha|^2\|_{L^{2^{n-1}}(\PP\otimes ds|_{[0,T]} \otimes \nu)}
        \right)^{1/2}
        =b_n c_{n-1}^{1/2}\|\alpha\|_{L^{2^n}(\PP\otimes ds|_{[0,T]} \otimes \nu)}
        ,
    \end{align}
    where the last estimate follows by applying the claimed result for $n':=n-1$ to $\alpha':=|\alpha|^2$.
    
    Moreover, the compensating integral satisfies
    \begin{align}
        &\left\|\int_{[0,T] \times E}|\alpha_s(e)| ds \otimes \nu(de)\right\|_{L^{2^n}(\PP)} =\left(\int_{([0,T] \times E)^{2n}}\E\left[\prod_{i=1}^{2^n}|\alpha_{s_i}(e_i)|\right] \otimes_{i=1}^{2^n}(ds_i \otimes \nu(de_i))\right)^{1/2^n}\\
        & \leq \left(\int_{([0,T] \times E)^{2n}}\prod_{i=1}^{2^n}\|\alpha_{s_i}(e)\|_{L^{2^n}(\PP)} \otimes_{i=1}^{2^n}(ds_i \otimes \nu(de_i))\right)^{1/2^n}
        =\int_{[0,T] \times E}\|\alpha_{s}(e)\|_{L^{2^n}(\PP)} ds \otimes \nu(de)
        \\
        &\leq (T\nu(E))^{1-1/2^n}\left(\int_{[0,T] \times E}\E[|\alpha_{s}(e)|^{2^n}] ds \otimes \nu(de)\right)^{1/2^n}
        =(T\nu(E))^{1-1/2^n}\|\alpha\|_{L^{2^n}(\PP\otimes ds|_{[0,T]} \otimes \nu)},
    \end{align}
    the last estimate following from  H\"older's inequality with $q=2^n$, $1/p=1-1/2^n$.
     By the triangle inequality the previous two estimates give
     \begin{align}
        \left\|\int_{[0,T] \times E}|\alpha_s(e)| N(ds,de)\right\|_{L^{2^n}(\PP)} 
        \leq \left(b_nc_{n-1}^{1/2}+(\nu(E)T)^{1-1/2^n}\right)\|\alpha\|_{L^{2^n}(\PP\otimes ds|_{[0,T]} \otimes \nu)}
        ,
     \end{align}
    and we conclude~\eqref{eq:PoissonEstimate} for $c_n:=b_nc_{n-1}^{1/2}+(\nu(E)T)^{1-1/2^n}$.
\end{proof}

\bibliography{Literature_Paper}

\end{document}